\documentclass[11pt,letterpaper]{article}
\usepackage{thmtools}
\usepackage{thm-restate}
\usepackage{bm}
\usepackage{bbm}
\usepackage{mathrsfs}           %
\usepackage{vmargin,fancyhdr}   %
\usepackage{tikz}
\usetikzlibrary{cd,positioning,chains,fit,shapes,calc}

\usepackage{subcaption}
\usepackage{algorithm}
\usepackage{algorithmic}
\usepackage{enumerate}

\usepackage{amsmath,amssymb, amsthm}    %
\usepackage{verbatim}           %
\usepackage{xspace}             %
\usepackage{graphicx,float}     %
\usepackage{ifthen,calc}        %
\usepackage{textcomp}           %
\usepackage{fancybox}           %
\usepackage{hhline}             %
\usepackage{float}              %

\usepackage[vflt]{floatflt}     %
\usepackage[small,compact]{titlesec}
\usepackage{setspace}
\usepackage{color}
\definecolor{ForestGreen}{rgb}{0.1333,0.5451,0.1333}
\usepackage{thumbpdf}
\usepackage[letterpaper,
colorlinks,linkcolor=ForestGreen,citecolor=ForestGreen,
backref,
bookmarks,bookmarksopen,bookmarksnumbered]
{hyperref}

\setpapersize{USletter}
\setmarginsrb{1in}{.5in}        %
{1in}{.5in}        %
{.25in}{.25in}     %
{.25in}{.5in}      %
\newcommand{\showccc}[0]{0}
\newcommand{\ccc}[2][nothing]{%
	\ifthenelse{\showccc=0}{}{
		\ensuremath{^{\Lsh\Rsh}}\marginpar{\raggedright\tiny\textsf{%
				\ifthenelse{\equal{#1}{nothing}}{}{\textbf{#1}\\}#2}}}}
\newcounter{hours}\newcounter{minutes}
\newcommand{\hhmm}{%
	\setcounter{hours}{\time/60}%
	\setcounter{minutes}{\time-\value{hours}*60}%
	\ifthenelse{\value{hours}<10}{0}{}\thehours:%
	\ifthenelse{\value{minutes}<10}{0}{}\theminutes}
\ifthenelse{\showccc=0}{\rhead{}}{\rhead{\today \ [\hhmm]}}
\newtheorem{theorem}{Theorem}

\newtheorem{proposition}{Proposition}
\newtheorem{corollary}{Corollary}
\newtheorem{definition}{Definition}

\newtheorem{remark}{Remark}
\newtheorem{lemma}{Lemma}
\newtheorem{fact}{Fact}

\newtheorem{assumption}{Assumption}

\usepackage{float}
\floatstyle{plain}\newfloat{myfig}{t}{figs}[section]
\floatname{myfig}{\textsc{Figure}}
\floatstyle{plain}\newfloat{myalg}{H}{algs}[section]
\floatname{myalg}{}
\newcommand{\defeq}{:=}
\newcommand*{\abs}[1]{\left\lvert#1\right\rvert}      
\newcommand{\norm}[1]{\left\lVert#1\right\rVert}
\newcommand{\normop}[1]{\left\lVert#1\right\rVert_{\textup{op}}}

\newcommand{\inprod}[2]{\left\langle#1, #2\right\rangle}
\newcommand{\eps}{\epsilon}

\newcommand{\R}{\mathbb{R}}

\newcommand{\N}{\mathbb{N}}

\newcommand{\half}{\frac{1}{2}}
\newcommand{\thalf}{\tfrac{1}{2}}
\newcommand{\1}{\mathbbm{1}}
\newcommand{\E}{\mathbb{E}}
\newcommand{\Var}{\textup{Var}}

\newcommand{\Nor}{\mathcal{N}}

\newcommand{\Tr}{\textup{Tr}}

\newcommand{\ma}{\mathbf{A}}
\newcommand{\mb}{\mathbf{B}}

\newcommand{\mm}{\mathbf{M}}
\newcommand{\my}{\mathbf{Y}}
\newcommand{\mx}{\mathbf{X}}

\newcommand{\id}{\mathbf{I}}
\newcommand{\tO}{\widetilde{O}}
\newcommand{\tOmega}{\widetilde{\Omega}}

\newcommand{\lmax}{\lambda_{\textup{max}}}
\newcommand{\lmin}{\lambda_{\textup{min}}}

\newcommand{\Par}[1]{\left(#1\right)}
\newcommand{\Brack}[1]{\left[#1\right]}
\newcommand{\Brace}[1]{\left\{#1\right\}}
\newcommand{\msig}{\boldsymbol{\Sigma}}
\newcommand{\mzero}{\mathbf{0}}

\newcommand{\ndir}{N_{\textup{dir}}}
\newcommand{\kmax}{k_{\max}}
\newcommand{\med}{\textup{med}}
\newcommand{\Tin}{T_{\textup{in}}}
\newcommand{\Tout}{T_{\textup{out}}}
\newcommand{\Tmid}{T_{\textup{mid}}}
\newcommand{\drd}{\delta_{\textup{rd}}}
\newcommand{\nin}{n_{\text{in}}}

\newcommand{\GMultifilter}{\mathsf{FastGaussianMultifilter}}
\newcommand{\NaiveCluster}{\mathsf{NaiveCluster}}
\newcommand{\NaiveClusterPlus}{\mathsf{NaiveClusterPlus}}
\newcommand{\GMultifilterBD}{\mathsf{FastGaussianMultifilterBoundedDiameter}}
\newcommand{\GPartition}{\mathsf{GaussianPartition}}
\newcommand{\GPartitionOneD}{\mathsf{Gaussian1DPartition}}
\newcommand{\Partition}{\mathsf{Partition}}
\newcommand{\PartitionOneD}{\mathsf{1DPartition}}
\newcommand{\GSOrC}{\mathsf{GaussianSplitOrCluster}}
\newcommand{\SOrC}{\mathsf{SplitOrCluster}}
\newcommand{\TBOC}{\mathsf{SplitOrTailBound}}
\newcommand{\Multifilter}{\mathsf{FastMultifilter}}
\newcommand{\MultifilterBD}{\mathsf{FastMultifilterBoundedDiameter}}
\newcommand{\RandDrop}{\mathsf{RandDrop}}
\newcommand{\IteratePostProcess}{\mathsf{IteratePostProcess}}
\newcommand*{\Fix}{\mathsf{Fixing}}

\newcommand{\ClusterUGMM}{\mathsf{ClusterUniformGMM}}
\newcommand{\ClusterGMM}{\mathsf{ClusterRobustGMM}}
\newcommand{\ClusterBFMM}{\mathsf{ClusterRobustBFMM}}
\newcommand{\ClusterBCMM}{\mathsf{ClusterRobustBCMM}}

\newcommand{\alg}{\mathcal{A}}

\newcommand{\mg}{\mathbf{G}}

\newcommand{\mproj}{\mathbf{P}}

\newcommand{\dist}{\mathcal{D}}

\newcommand{\Sym}{\mathbb{S}}
\newcommand{\cov}{\textup{Cov}}
\newcommand{\tcov}{\widetilde{\cov}}

\newcommand{\mus}{\mu^*}
\newcommand{\hmu}{\hat{\mu}}
\newcommand{\bmu}{\bar{\mu}}

\newcommand{\taumed}{\tau_{\textup{med}}}
\newcommand{\tL}{\widetilde{L}}
\newcommand{\tX}{\widetilde{X}}
\newcommand{\ml}{\mathbf{L}}
\newcommand{\tmu}{\tilde{\mu}}

\newcommand*\poly{\operatorname{poly}}

\definecolor{burntorange}{rgb}{0.8, 0.33, 0.0}

  \usepackage{nth}
  \usepackage{intcalc}

\usepackage{url}

\def\colorful{1}

\ifnum\colorful=1

\else

\fi

\newtheorem{question}[theorem]{Question}

\begin{document}

\begin{titlepage}
  \def\thepage{}
  \thispagestyle{empty}
  
  \title{Clustering Mixture Models in Almost-Linear Time \\ via List-Decodable Mean Estimation}  
  \date{}
  \author{
    Ilias Diakonikolas\thanks{University of Wisconsin, Madison, {\tt ilias@cs.wisc.edu}.}
    \and
    Daniel M.\ Kane\thanks{University of California, San Diego, {\tt dakane@cs.ucsd.edu}.}
    \and
    Daniel Kongsgaard\thanks{University of California, San Diego, {\tt dkongsga@ucsd.edu}.} 
    \and
    Jerry Li\thanks{Microsoft Research, {\tt jerrl@microsoft.com}}
    \and
    Kevin Tian\thanks{Stanford University, {\tt kjtian@stanford.edu}. Part of this work was done as an intern at Microsoft Research.}
  }
  
  \maketitle

\abstract{
We study the problem of list-decodable mean estimation, where an adversary can corrupt a majority of the dataset. 
Specifically, we are given a set $T$ of $n$ points in $\mathbb{R}^d$ and a parameter $0< \alpha <\frac 1 2$ 
such that an $\alpha$-fraction of the points in $T$ are i.i.d.\ samples from a well-behaved distribution $\mathcal{D}$ 
and the remaining $(1-\alpha)$-fraction are arbitrary. The goal is to output a small list of vectors,
at least one of which is close to the mean of $\mathcal{D}$. We develop new algorithms for list-decodable mean estimation, 
achieving nearly-optimal statistical guarantees, with running time $O(n^{1 + \eps_0} d)$, for any fixed $\eps_0 > 0$. 
All prior algorithms for this problem had additional polynomial factors in $\frac 1 \alpha$. 
We leverage this result, together with additional techniques, to obtain the first {\em almost-linear time} algorithms 
for clustering mixtures of $k$ separated well-behaved distributions, nearly-matching the statistical guarantees of spectral methods. 
Prior clustering algorithms inherently relied on an application of $k$-PCA, thereby incurring runtimes of $\Omega(n d k)$. 
This marks the first runtime improvement for this basic statistical problem in nearly two decades.

The starting point of our approach is a novel and simpler near-linear time robust mean estimation algorithm in the $\alpha \to 1$ regime, 
based on a one-shot matrix multiplicative weights-inspired potential decrease. We crucially leverage this new algorithmic framework 
in the context of the iterative multi-filtering technique of \cite{DiakonikolasKS18, DiakonikolasKK20}, 
providing a method to simultaneously cluster and downsample points using {\em one-dimensional} 
projections --- thus, bypassing the $k$-PCA subroutines required by prior algorithms.
}

\end{titlepage}
\pagenumbering{gobble}
\newpage
\setcounter{tocdepth}{2}
{
  \hypersetup{linkcolor=black}
  \tableofcontents
}
\newpage
\pagenumbering{arabic}

\section{Introduction}\label{sec:intro}

We develop novel algorithms achieving almost-optimal runtimes for two closely related fundamental problems in high-dimensional statistical estimation: clustering well-separated mixture models and mean estimation in the list-decodable learning (``majority-outlier'') regime. Before we formally state our contributions, we provide the necessary background and motivation for this work.

\paragraph{Clustering well-separated mixture models.} Mixture models are a well-studied class of generative models used widely in practice. Given a family of distributions $\cal{F}$, a mixture model $\cal{M}$ with $k$ components is specified by $k$ distributions $\dist_1, \ldots, \dist_k \in \cal{F}$ and nonnegative mixing weights $\alpha_1, \ldots, \alpha_k$ summing to one, and its law is given by $\sum_{i \in [k]} \alpha_i \dist_i$.
That is, to draw a sample from $\cal{M}$, we first choose $i \in [k]$ with probability $\alpha_i$, and
then draw a sample from $\dist_i$. When the weights are all equal to $\frac 1 k$, we call the mixture \emph{uniform}.
Mixture models, especially Gaussian mixture models, have been widely studied in statistics since pioneering
work of Pearson in 1894 \cite{Pearson94}, and more recently, in theoretical computer science ~\cite{dasgupta1999learning, arora2005learning, vempala2004spectral, achlioptas2005spectral, KannanSV08, BrubakerV08, awasthi2012improved, RegevV17}.

A canonical learning task for mixture models is the {\em clustering problem}.
Namely, given independent samples drawn from $\cal{M}$, the goal is to
approximately recover which samples came from which component.
To ensure that this inference task is information-theoretically possible,
a common assumption is that $\cal{M}$ is ``well-separated'' and ``well-behaved'':
for example, we may assume each component $\dist_i$ is sufficiently concentrated
(with sub-Gaussian tails or bounded moments),
and that component means have pairwise distance at least $\Delta$, for sufficiently large $\Delta$.
The goal is then to efficiently and accurately cluster samples from $\cal{M}$ %
with as small a separation as possible.

The prototypical example is the case of uniform mixtures of bounded-covariance Gaussians, i.e.\
mixtures of the form $\mathcal{M} = \sum_{i \in [k]} \frac 1 k \Nor(\mu_i, \msig_i)$,
where each $\msig_i$ is unknown and satisfies $\normop{\msig_i} \leq \sigma^2$.
Prior to the current work, the fastest known algorithm for this learning problem was due to~\cite{achlioptas2005spectral},
building on \cite{vempala2004spectral}. Notably, \cite{achlioptas2005spectral}  gave a polynomial-time clustering
algorithm when $\Delta = \Omega(\sigma \sqrt k)$. 
Interestingly, the algorithmic approach of~\cite{vempala2004spectral, achlioptas2005spectral} is 
surprisingly simple and elegant:
first, run $k$-PCA on the set of $n$ samples in $\R^d$ to find a $k$-dimensional subspace
(which can be shown to approximately capture the span of the component means),
and then perform a distance-based clustering algorithm in this subspace.
The runtime of this algorithm is dominated by $\tOmega(ndk)$ -- the cost 
of (approximate) $k$-PCA.\footnote{Throughout the paper, when convenient, 
	we hide polylogarithmic factors in the sample size and algorithm failure probabilities with the $\tO$ notation. 
	We reserve the terminology ``almost-linear'' to mean linear up to subpolynomial factors, 
	and the terminology ``nearly-linear'' to mean linear up to polylogarithmic factors.} 
The idea of using $k$-PCA as a subroutine to solve the clustering problem 
is very natural and has also been useful in practice.
Indeed, %
using PCA as a preprocessing step before applying further learning algorithms 
(such as clustering) is so ubiquitous that it is commonly suggested 
by introductory textbooks on machine learning, see e.g.~\cite{murphy2012machine}.

However, in our setting, since the size of the description this problem is $O(n d)$, the runtime of $k$-PCA is off from linear time by a factor of roughly $k$. 
In many real-world settings, this factor of $k$ is quite significant. 
For instance, modern image datasets such as ImageNet~\cite{deng2009imagenet} 
often have hundreds or thousands of different classes and subclasses~\cite{santurkar2020breeds}. 
As a result, many clustering tasks on these datasets often have $k$ of the same order.
The resulting overhead would cause many tasks to be infeasible at scale on these datasets.
Yet, despite considerable attention over the last two decades,\footnote{We note that a recent line of work has developed sophisticated
	polynomial-time clustering algorithms under smaller separation assumptions, see e.g.~\cite{DiakonikolasKS18,hopkins2018mixture,kothari2018robust}.
	These algorithms leverage higher moments of the distribution and consequently require significantly
	higher sample and computational complexity.} no faster algorithm has been developed for the clustering
task. In particular, the runtime of $k$-PCA has remained a bottleneck in this setting.

The preceding discussion motivates the following natural question.

\begin{question}\label{qn:mixtures}
	Can we cluster mixtures of $k$ ``well-separated'' structured distributions without the use of $k$-PCA?
	More ambitiously, is there a clustering algorithm that runs in (almost)-{\em linear time}?
\end{question}
Prior to the current work, this question remained open even for uniform $k$-mixtures of identity covariance Gaussians
with pairwise mean separation as large as $\poly(k)$. In addition to its fundamental interest, a runtime improvement
of this sort may have significant practical implications for clustering at scale in real-world applications, see e.g.~\cite{patinkin2011method,wierzchon2018modern}, where spectral methods are commonly used.
As our main contribution, {\em we resolve Question~\ref{qn:mixtures} for the general class of
mixtures of bounded-covariance distributions under information-theoretically near-optimal separation.}

\paragraph{List-decodable mean estimation.}
In many statistical settings, including machine learning security~\cite{BarrenoNJT10, BiggioNL12, SteinhardtKL17,diakonikolas2019sever}
and exploratory data analysis e.g.\ in biology~\cite{RosenbergPWCKZF02, LiATSCRCBFCM08, PaschouLJD10},
datasets contain arbitrary --- and even adversarially chosen --- outliers.
The central question of the field of robust statistics
is to design estimators tolerant to a small amount of unconstrained contamination.
Classical work in this field~\cite{anscombe1960rejection,tukey1960survey,huber1964robust,tukey1975mathematics, HampelRRS86, huber2004robust} developed
robust estimators for many basic tasks, although with computational costs scaling exponentially in the problem dimension.
More recently, a line of work in computer science, starting with~\cite{diakonikolas2019robust, lai2016agnostic},
developed the first computationally-efficient learning algorithms (attaining near-optimal error)
for various estimation problems. Subsequently, there has been significant progress
in algorithmic robust statistics in a variety of settings (see~\cite{diakonikolas2019recent} for a survey).

In many of these works, it is typically assumed that the fraction of corrupted points is less than $\half$.
Indeed, when more than half the points are corrupted, the problem is ill-posed:
there is not necessarily a uniquely-defined notion of ``uncorrupted samples.''
While outputting a {\em single} accurate hypothesis in this regime is information-theoretically impossible,
one may be able to compute a {\em small list} of hypotheses with the guarantee that {\em at least one of them} is accurate.
This relaxed notion of estimation is known as \emph{list-decodable learning}~\cite{balcan2008discriminative, CharikarSV17}.

\begin{definition}[List-decodable learning]\label{def:ld}
Given a parameter $0<\alpha < \half$ and a distribution family $\cal{F}$ on $\R^d$,
a \emph{list-decodable learning algorithm} takes as input $\alpha$ and a multiset $T$ of $n$ points
such that an unknown $\alpha$ fraction of $T$ are independent samples
from an unknown distribution $\dist \in \mathcal{F}$, and no assumptions are made on the remaining samples.
Given $T$ and $\alpha$, the goal is to output a ``small'' list of hypotheses
at least one of which is close to the target parameter of $\dist$.
\end{definition}

\noindent

Arguably the most fundamental problem in the list-decodable learning setting is mean estimation,
wherein the goal is to output a small list of hypotheses, one of which is close to the true mean.
A natural problem in its own right, list-decodable mean estimation generalizes the problem of 
learning well-separated mixture models (as explained below) and can model important 
applications such as crowdsourcing~\cite{steinhardt2016avoiding, meister2018data} 
or semi-random community detection in stochastic block models \cite{CharikarSV17}. Moreover, it is particularly 
useful in the context of semi-verified learning~\cite{CharikarSV17, meister2018data}, 
where a learner can audit a small amount of trusted data. {\em An important remark is
that the parameter $\alpha \in (0, \half)$ can be quite small in some of these applications and should not necessarily
be thought of as a constant.} In addition to applications in clustering mixture models, 
a concrete example is the crowdsourcing setting with many unreliable responders 
studied in \cite{meister2018data}, where the parameter $\alpha$ is tiny, 
depending inversely-polynomially on other problem parameters such as the dimension.

The parameter $\alpha$ in the list-decodable mean estimation setting plays a very similar role
to the parameter $\frac 1 k$ in learning (uniform) mixture models.
This is no coincidence: list-decodable mean estimation can be thought of as a natural robust generalization
of clustering well-separated mixtures. Indeed, if we run a list-decodable mean estimation algorithm
on a dataset drawn from a uniform mixture of $k$ sufficiently nice and well-separated distributions
with $\alpha$ set to $\frac 1 k$, the output list {\em must} contain
a candidate mean which is close to the mean of each component. This is because from the perspective
of the list-learning algorithm, each component could be the ``true'' unknown distribution $\dist$,
and thus the list must contain a hypothesis close to the mean of this ``true'' distribution.
This small list of hypotheses can then typically be used %
to cluster the original dataset. One conceptually important implication
of this observation is that list-decodable mean estimation algorithms also naturally lead to algorithms
for clustering well-separated mixture models (even in the presence of a small fraction of corrupted samples) ---
a reduction we formalize in this work.

The first polynomial-time algorithm for list-decodable mean estimation, 
when $\cal{F}$ is the family of bounded-covariance distributions, 
was by \cite{CharikarSV17}. The \cite{CharikarSV17} algorithm 
was based on black-box calls to semidefinite program solvers and had a large polynomial runtime.
Since then, a sequence of works \cite{DiakonikolasKK20, CherapanamjeriMY20, DiakonikolasKKLT20}
have obtained substantially improved runtimes for this problem,
while retaining the (near-optimal) statistical guarantees of \cite{CharikarSV17}.
The algorithm by \cite{DiakonikolasKKLT20} runs in time 
$\widetilde{O}(\frac{nd}{\alpha})$ and achieves near-optimal error
(within a polylogarithmic factor). %

Interestingly, as in the case of clustering mixture models, 
the $\widetilde{\Omega}(\frac{nd}{\alpha})$ runtime dependence 
of the \cite{DiakonikolasKKLT20} algorithm is {\em also} due to running a $k$-PCA subroutine 
--- for $k = \Omega(\frac 1 \alpha)$ ---  
to reduce the problem to a $k$-dimensional subspace. In more detail, 
the algorithm of~\cite{DiakonikolasKKLT20} can be viewed as a reduction 
from list-decodable mean estimation to polylogarithmically many calls
to $k$-PCA (for carefully chosen matrices).
Thus, the cost of $k$-PCA appears as a runtime barrier in state-of-the-art algorithms 
for list-decodable mean estimation as well.
In regimes where $\alpha$ is small, the $\widetilde{\Omega}(\frac{nd}{\alpha})$
runtime is significantly sub-optimal in the input size. This leaves open whether 
the extraneous linear dependence on $\alpha^{-1}$ is improvable, and brings us to our second main question.

\begin{question}\label{qn:ld}
Can we perform list-decodable mean estimation with near-optimal statistical guarantees in (almost)-linear time?
\end{question}

In this paper, {\em we similarly resolve Question~\ref{qn:ld} for the class of
bounded-covariance distributions.}

\subsection{Our results} \label{ssec:results}

We answer both Question~\ref{qn:mixtures} and Question~\ref{qn:ld} in the affirmative, up to subpolynomial factors.
Perhaps surprisingly, to resolve the longstanding open problem of clustering mixture models in almost-linear time,
we develop an almost-linear time algorithm for the (much more general) problem of list-decodable mean estimation.
To then solve the clustering problem, we develop a fast post-processing technique that efficiently reduces the clustering task to list-decodable mean estimation.
In light of this development, we begin by presenting our list-decodable estimation result.

\begin{theorem}[informal, see Theorem~\ref{thm:multifilter}]\label{thm:inflistmean}
For any fixed constant $\eps_0 > 0$, there is an algorithm $\Multifilter$ with the following guarantee.
Let $\dist$ be a distribution over $\R^d$ with unknown mean $\mus$ and unknown covariance $\msig$ with
$\normop{\msig} \le \sigma^2$, and let $\alpha \in (0, 1)$.
Given $\alpha$ and a multiset of $n = \Omega(\frac d \alpha)$ points on $\R^d$
such that an $\alpha$-fraction are i.i.d.\ draws from $\dist$, $\Multifilter$ runs in time $O(n^{1 + \eps_0} d)$ and
outputs a list $L$ of $O(\alpha^{-1})$ hypotheses so that with high probability
we have
\[\min_{\hmu \in L} \norm{\hmu - \mus}_2 = O \Par{\frac {\sigma \log \alpha^{-1}} {\sqrt{\alpha}}}.\]
\end{theorem}

\noindent Notably, in the setting of
Theorem~\ref{thm:inflistmean}, a sample complexity of $\Omega(\frac d \alpha)$,
error of $\Omega(\sigma \alpha^{-\half})$, and list size $\Omega(\alpha^{-1})$ are all information-theoretically necessary~\cite{DiakonikolasKS18}. Hence, up to a $\log(\alpha^{-1})$ factor in the error,
Theorem~\ref{thm:inflistmean} achieves optimal statistical guarantees
for this problem in almost-linear time. 

Leveraging Theorem~\ref{thm:inflistmean}, and combining it with a new almost-linear time post-processing
procedure of the resulting list, we achieve our almost-linear runtime for clustering well-separated mixtures under only a second moment bound assumption
--- even in the presence of a small fraction of outliers. In more detail, our algorithm can tolerate a fraction
of outliers proportional to the relative size of the smallest true cluster. For brevity, in this introduction, we will state the natural special case of our clustering result for uniform bounded-covariance mixtures without outliers.
We also achieve similar (indeed, slightly stronger) guarantees when the mixture components 
are sub-Gaussian or have bounded fourth moments.

\begin{theorem}[informal, see Corollaries~\ref{cor:fullgmm},~\ref{cor:fullbfmm},~\ref{cor:fullbcmm}]\label{thm:infcluster}
For any fixed constant $\eps_0 > 0$, there is an algorithm with the following guarantee.
Given a multiset of $n = \Omega (d k)$ i.i.d.\ samples from a uniform mixture model 
$\mathcal{M} = \sum_{i \in [k]} \frac 1 k \dist_i$,
where each component $\dist_i$ has unknown mean $\mu_i$, unknown covariance matrix $\msig_i$
with $\normop{\msig_i} \le \sigma^2$, and $\min_{i, i' \in [k], i \neq i'} \| \mu_i - \mu_{i'}\|_2 =
\tOmega(\sqrt{k}) \, \sigma$, the algorithm runs in time $O(n^{1 + \eps_0} \max(k, d))$, 
and with high probability correctly clusters $99\%$ of the points.
\end{theorem}

Some remarks are in order. First, we note that pairwise mean separation of $\Omega(\sqrt{k} \, \sigma)$
is information-theoretically necessary for accurate clustering to be possible for bounded covariance components.
The algorithm establishing Theorem~\ref{thm:infcluster} 
nearly achieves the optimal separation.
Secondly, and crucially, our clustering algorithm runs in almost-linear time.
Finally, as previously alluded to, our clustering method is robust to outliers, and can handle mixtures with arbitrary weights, with guarantees depending on the smallest weight (see Corollary~\ref{cor:fullbcmm} for a precise statement).

It is worth commenting on the $\max(k, d) $ term appearing in the running time of Theorem~\ref{thm:infcluster}.
Our algorithm runs in almost-linear time as long as $k \leq d$.
For the extreme regime where $k \gg d$, our algorithm has running time
$O(n^{1 + \eps_0} k)$. In this parameter regime, it is plausible
that $\Omega(n k)$ is a runtime bottleneck for the following reason: even if we are given (exactly)
the centers $\mu_i$, $i \in [k]$ for free, $\Omega(n k)$ time seems to be required to simply assign
each of the $n$ points to its closest center.

\begin{remark}[Prior work]\label{rem:lit}
{\em The prior works~\cite{achlioptas2005spectral,awasthi2012improved} obtained polynomial-time clustering
algorithms with similar statistical guarantees as Theorem~\ref{thm:infcluster}, under the (much stronger) assumption
that each component distribution $\dist_i$ has sub-Gaussian tails. For bounded covariance distributions, these algorithms require the stronger mean separation of $\Omega(k \sigma)$~\cite{Aw21}. On the other hand,
the clustering methods obtained in \cite{CharikarSV17} (as an application of their list-decodable mean estimator)
(i) require sub-Gaussian components, and (ii) partition the dataset into $C \cdot k$ for some constant $C>2$ 
--- as opposed to $k$ --- clusters.
In summary, prior work has not explicitly obtained {\em even a polynomial-time}
clustering algorithm in the bounded covariance setting with separation $o(k) \sigma$.}
\end{remark}

\subsection{Technical overview}\label{ssec:techniques}

Here, we describe the techniques developed in this paper at a high level, and how they circumvent several conceptual runtime barriers encountered by prior approaches to list-decodable mean estimation and clustering mixture models. 
Our full proofs are quite technically challenging, and involve several additional steps which we omit here for clarity of exposition. 
Throughout this section, we assume that the ``scale'' of the problem is $\sigma = 1$ for simplicity (e.g.\ distribution covariances are bounded by $\id$).

\subsubsection{Prior approaches and their limitations} \label{sssec:prior-techniques}

In this section, we briefly describe two recent fast algorithms for list-decodable mean estimation, developed by \cite{DiakonikolasKK20} and \cite{DiakonikolasKKLT20},\footnote{We focus on \cite{DiakonikolasKKLT20} instead of \cite{CherapanamjeriMY20} in this technical exposition, as they both apply Ky Fan semidefinite programming machinery to obtain fast runtimes, but the \cite{DiakonikolasKKLT20} approach is more relevant to this paper.} focusing on tools used in their analyses and bottlenecks in extending their techniques to obtain (almost)-linear runtimes.

\paragraph{Multifiltering.} 
Filtering is one of the most popular techniques for robust estimation \cite{diakonikolas2019robust, DiakonikolasKK017, Li18, Steinhardt18, diakonikolas2019recent}. In the minority-outlier setting, filtering is based on the idea of designing certificates of corruption, which either ensure that a current estimate suffices, or can be used to identify a set of points to filter 
on containing more outliers (corrupted points) than inliers (clean points). 
Iterating this process terminates in polynomial time, because (roughly speaking) it eventually removes all outliers.

In the context of list-decodable mean estimation, standard filtering guarantees are insufficient, 
because we cannot afford to remove as many inliers as outliers. 
To overcome this difficulty, \cite{DiakonikolasKS18} introduced the ``multifilter'' in the context of Gaussian mean estimation, 
which was extended to bounded covariance distributions in~\cite{DiakonikolasKK20}. 
At a high level, a multifilter iterates through a tree of candidate subsets, and looks for ways to either ``cluster'' a subset 
or ``split'' it into multiple (overlapping) subsets.\footnote{In \cite{DiakonikolasKK20}, these subsets were replaced by weight functions, but the intuition is very similar in both cases.} To ensure an efficient runtime, a multifilter maintains a potential 
guaranteeing that the tree size does not blow up (i.e.\ there are never too many candidate subsets), 
and carefully chooses to split or cluster based on subset sample statistics, thus 
ensuring that some tree node always retains a large fraction of inliers. 
Previous multifilters chose to split or cluster subsets based on \emph{one-dimensional} projections 
along top eigenvectors of sample covariances, which can be dominated by a single outlier. 
In the worst case, this leads to an iteration count scaling polynomially with the dimension.

\paragraph{Filtering via matrix multiplicative weights.} 
The approach taken by the fastest algorithms for mean estimation 
in both majority-inlier~\cite{DongH019} and majority-outlier~\cite{DiakonikolasKKLT20} settings
is heavily motivated by filtering. In the majority-inlier case, every iteration of the filter is nearly-linear time, 
so the only bottleneck to an overall fast runtime is the number of iterations. 
However, simple hard instances show that only projecting onto the worst directions of empirical covariances may 
lead to an $\Omega(d)$ runtime overhead.
The main idea of~\cite{DongH019} was to choose scores capturing \emph{multiple} bad directions at a time, 
preventing this worst-case behavior. These scores were based on quadratic forms with certain trace-one matrices derived from the \emph{matrix multiplicative weights} (MMW) regret minimization framework from semidefinite programming \cite{WarmuthK06, AroraK07}. By using MMW regret bounds, \cite{DongH019} designs a filter that efficiently decreases 
the empirical covariance operator norm, which is used as a potential to yield convergence 
in {\em polylogarithmically} many iterations.

In the majority-outlier setting, the story is somewhat murkier. 
To overcome complications of prior list-decodable mean estimation algorithms (e.g.\ the multifilter), 
which interleaved ``filtering'' and ``clustering'' steps, \cite{DiakonikolasKKLT20} designed a ``$k$-dimensional filter'', 
for $k = \Theta(\frac 1 \alpha)$, that they called $\mathsf{SIFT}$, decoupling the two goals. 
Specifically, $\mathsf{SIFT}$ uses scores based on $k$-dimensional projections to hone in on a subspace outside of which the empirical mean is accurate. It then efficiently clusters in just this subspace; 
combined with appropriate Ky Fan norm generalizations of MMW, 
the number of iterations is then improved to polylogarithmic. 
However, this approach of decoupling filtering and clustering appears to inherently use 
$k$-dimensional PCA as a subroutine, for $k = \Theta(\frac 1 \alpha)$, 
even just to learn an ``important'' subspace a single time. 
Hence, this approach encounters a similar runtime bottleneck as prior algorithms for clustering mixture models~\cite{vempala2004spectral, achlioptas2005spectral}.

\paragraph{Challenges in combining techniques.} 
As mentioned, the approach of \cite{DiakonikolasKKLT20} seems to inherently run into a runtime barrier at 
$\Omega(\frac {nd} \alpha)$ due to its reliance on $k$-PCA. This suggests that to overcome this barrier, 
we need to develop a new algorithm which both (1) does not disentangle filtering and clustering steps, and (2) relies on univariate projections. It is natural to then try to merge the multifilter with a MMW-based potential to ensure rapid convergence.

Unfortunately, there are several obstacles towards combining these frameworks. 
A primary complication is that the regret minimization approach of \cite{DongH019} 
requires multiple consecutive rounds before it can ensure an appropriate potential decreases. 
This is because of its reliance on MMW, a ``mirror descent'' algorithm which typically does 
not provide monotone guarantees on iterates (and hence requires multiple iterations to bound regret)~\cite{DaskalakisISZ18}. 
It is unclear how to make these arguments work within the multifilter framework, 
which interleaves two types of steps (splitting and clustering) that 
may have incompatible guarantees across iterations. 

Finally, even if it were possible to combine the multifilter with a MMW-based potential analysis, 
there are still various difficulties towards obtaining an almost-linear runtime 
coming from the size of our hypothesis tree. For example, making the decision to split or cluster 
at a node typically requires $\Omega(nd)$ time (e.g.\ to compute scores), 
which we cannot afford to perform more than subpolynomially often. 
This is problematic because our multifilter tree certainly contains $\Omega(\frac 1 \alpha)$ nodes: 
in the uniform mixture model case, our tree must contain hypotheses corresponding to every true cluster.

\subsubsection{Our techniques} \label{sssec:new-techniques}

\paragraph{One-shot potential framework.} 
In order to deal with the first of the two obstacles discussed (the non-monotonicity of MMW regret guarantees), 
our starting point is a framework for fast robust mean estimation (cf.\ Section~\ref{ssec:fastfilter}), 
essentially matching the guarantees of \cite{DongH019} with a more transparent analysis. 
Crucially, our new framework comes with a ``one-shot'' potential function that shows 
monotone progress {\em at every iteration}, making it more amenable to combination with a multifilter 
(which needs to argue how potentials evolve between different types of steps).

In more detail, our new fast algorithm in the majority-inlier setting guarantees monotone progress 
on the ``Schatten-norm'' potential $\Tr(\my_t^2)$, where $\my_t:= \mm_t^{\log(d)}$ and 
$\mm_t = \sum_{i \in T} [w_t]_i (X_i - \mu_t)(X_i - \mu_t)^\top$ is the weighted empirical covariance 
with respect to the current weight vector $w_t$. We then use $\my_t$ to sample carefully chosen 
Gaussian random vectors to locate outliers in multiple univariate directions. By using the guarantees 
of Johnson-Lindenstrauss projections, we can use these univariate filters to ensure 
the next (weighted) empirical covariance matrix satisfies
\begin{equation}\label{eq:keypotfast}
\inprod{\my_t^2}{\mm_{t + 1}} \le O(1)\Tr(\my_t^2) \;.
\end{equation}
Combining \eqref{eq:keypotfast} with a fact from \cite{jambulapati2020robust} shows that 
our potential decays geometrically, resulting in rapid convergence. 
Fortunately, we can use the same potential in the multifilter context, as long as we guarantee that 
\eqref{eq:keypotfast} holds for \emph{every} child of a node (whether a split or cluster step is used). 
In particular, applying \eqref{eq:keypotfast} repeatedly for any path in the multifilter tree implies 
that the depth is $\text{polylog}(d)$. It remains to bound the \emph{width} of the tree (the computational cost per layer), 
while maintaining the invariant that at least one node on every level preserves enough inliers.

\paragraph{Warmup: fast Gaussian multifilter via indicator weights.} 
Recall that our other obstacle towards an almost-linear runtime is that each of the $\Omega(\alpha^{-1})$ nodes 
of our multifilter tree requires $\Omega(nd)$ time to decide on a multifiltering step. 
Our strategy is to reduce the \emph{total number of nodes} across each layer of the tree, 
so that the total cost of multifiltering on all of them is roughly $nd$. We achieve this goal 
by ensuring that our multifilter always maintains nodes which specify subsets of our original data 
(i.e.\ $0$-$1$ weights rather than soft weights $\in [0, 1]$). Hence, each layer of our new multifilter 
trades off the \emph{number of subsets} with the \emph{cost of multifiltering} on each subset. 
Considering the two extreme layers is illustrative of this tradeoff: at the root, 
our algorithm performs a single one-dimensional projection on the entire dataset; 
at the leaves, it performs $O(\alpha^{-1})$ one-dimensional projections, 
each on a subset consisting of roughly an $\alpha$-fraction of the original dataset.

As a warmup, we first show how to achieve this in the case where the ground-truth, $\dist$, 
is a bounded-covariance Gaussian (see Section~\ref{sec:gaussian}), 
so we can exploit strong concentration bounds. In particular, we know that in any linear projection 
almost all of the inliers will lie in an interval of logarithmic length. If almost all of our sample points 
in a subset are clustered within such an interval, we can explicitly remove all samples outside of it. 
On the other hand, if our samples are spread out, we can split them into two (unweighted) subsets 
with sufficient overlap to ensure that at least one of the children subsets will contain almost all the inliers, 
as long as the parent did. We can in fact apply such a partitioning strategy iteratively along each 
univariate projection, until each remaining subset is contained in a short interval; 
this suffices to imply \eqref{eq:keypotfast}. 

\paragraph{From Gaussians to bounded-covariance distributions.} 
Substantially more technical care is required in the bounded-covariance setting 
to achieve an almost-linear runtime without sacrificing the error rate. 
Notably, we will no longer be able to guarantee that the subsets lie in short intervals, 
due to weaker concentration properties. This also means that we cannot deterministically 
remove points, making it more challenging to ensure the weight functions we keep are indicators.

We overcome these challenges in Section~\ref{sec:boundedcov} through several new technical developments. 
We first weaken the outcome guarantee of our recursive partitioning strategy, 
from ensuring each cluster lies in a short interval, to requiring bounded variance, 
which we show suffices to advance on the potential \eqref{eq:keypotfast}. 
Furthermore, we use a randomized dropout strategy in place of the ``clustering'' step of the multifilter, 
and design fast quantile checks to ensure the ``split'' step can be conducted in nearly-constant time. 
By carefully combining these subroutines, we can indeed ensure every child of nodes 
in a layer satisfies \eqref{eq:keypotfast}, and that the total computational cost of splitting 
or clustering on the entire layer is almost-linear. With our earlier depth bound, 
this yields our full runtime guarantee.

\paragraph{Reducing clustering to list-decodable learning.} 
In Section~\ref{sec:clustering}, we demonstrate that several mixture model clustering tasks 
enjoy benefits from the speedups afforded by our list-decodable learning methods. 
In the following, assume we have a list $L$ of size $O(\alpha^{-1})$ and $L \supseteq \{\hmu_i\}_{i \in [k]}$ 
with $\norm{\hmu_i - \mu_i}_2 = \tO(\sqrt {\alpha^{-1}})$ for all $i \in [k]$, 
where $\mu_i$ is the mean of the mixture component $\dist_i$.

For sub-Gaussian components, we build on a clustering algorithm of \cite{DiakonikolasKS18} 
and improve it to run in nearly-linear time via randomized distance comparisons. 
The main idea of the \cite{DiakonikolasKS18} algorithm is to exploit concentration, 
which implies that with high probability, all points drawn from $\dist_i$ have a closest hypothesis 
in $L$ at distance $\tO(\sqrt {\alpha^{-1}})$ from $\mu_i$. By rounding every sample to its nearest hypothesis, 
and assuming separation $\tOmega(\sqrt {\alpha^{-1}})$ between component means, 
we can perform an efficient equivalence class partitioning which clusters the data. 
We observe that this framework is tolerant to a small amount of poorly-behaved points or outliers 
and generalizes to cluster components with bounded fourth moments.

For our most general application of clustering mixtures under only bounded component covariances, as stated in Theorem~\ref{thm:infcluster}, 
the same framework does not apply as a constant fraction of all points may be misbehaved due to weak concentration. 
To address this, we develop a new postprocessing technique, relying on the following observation: 
letting $\mproj$ be the projection onto the $O(\alpha^{-1})$-dimensional subspace spanned by $L$, 
any sample hit by $\mproj$ will lie within distance $O(\sqrt{\alpha^{-1}})$ of its corresponding cluster mean 
in the low-dimensional subspace with constant probability. We use this observation to drop hypotheses 
which are too far away from the true means, and then an appropriate equivalence relation suffices for clustering. 
The runtime bottleneck of this strategy is the computation and application of $\mproj$ to our dataset, 
which can be quite expensive. We show that by instead measuring distances in a $O(\log d)$-dimensional subspace 
formed by random projections within $\mproj$, and clustering based on these estimates, 
we obtain similar clustering performance by exploiting guarantees of Johnson-Lindenstrauss transforms.

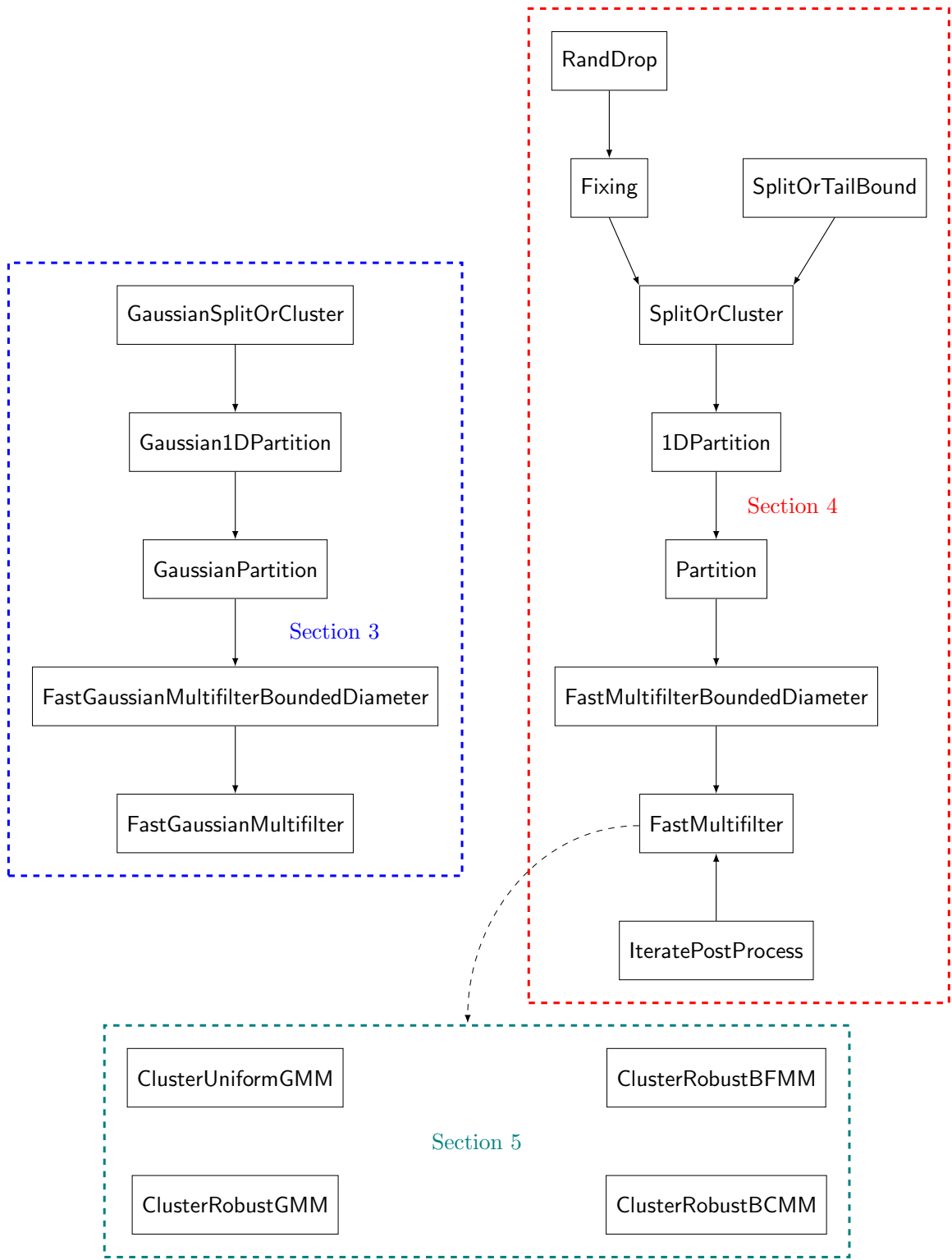
\begin{figure}[ht!]
 	\centering
 	\scalebox{0.95}{%
 		\begin{tikzcd}[%
 			,ampersand replacement=\&,
 			,cells={nodes={draw, minimum width=1.0cm, minimum height=1.0cm}}
 			,column sep=0em, row sep=3em
 			,every arrow/.append style={-latex},
 			execute at end picture={%
 				\draw[blue, very thick, dashed] ([xshift=-4.8em,yshift=1em]C1.north west) rectangle ([xshift=4.8em,yshift=-1em]C2.south east) node[pos=.6,right]{Section 3};
 				\draw[red, very thick, dashed] ([xshift=-1em,yshift=1em]D1.north west) rectangle ([xshift=6em,yshift=-1em]D2.south east) node[pos=.5,right]{Section 4};
 				\draw[teal, very thick, dashed] ([xshift=-1em,yshift=1em]E1.north west) rectangle ([xshift=1em,yshift=-1em]E2.south east) node[pos=0.5]{Section 5};
 				\draw[-latex,dashed] (M.west) to [out=180,in=90] ([xshift=5.5em,yshift=1.1em]E1.north east);
 			}
 			]
 			\&[5em] |[alias=D1]| \RandDrop \ar[d] \\
 			\& \Fix \ar[dr, start anchor = south, end anchor = north west] \&[-5em]\&[-6em] \TBOC \ar[dl, start anchor = south, end anchor = north east] \\
 			|[alias=C1]| \GSOrC \ar[d] \&\& \SOrC \ar[d] \\
 			\GPartitionOneD \ar[d] \&\& \PartitionOneD \ar[d] \\
 			\GPartition \ar[d] \&\& \Partition \ar[d]\\
 			\GMultifilterBD \ar[d] \&\& \MultifilterBD \ar[d] \\
 			|[alias=C2]| \GMultifilter \&\& |[alias=M]| \Multifilter \\
 			\&\& |[alias=D2]| \IteratePostProcess \ar[u] \\
 			|[alias=E1]| \ClusterUGMM \&\& \ClusterBFMM \\
 			\ClusterGMM \&\& |[alias=E2]| \ClusterBCMM
 		\end{tikzcd}
 	}%
 	\caption{An illustration of the dependencies of the different algorithms of this paper together with which section the algorithms are described in. The dashed arrow from $\Multifilter$ to Section~\ref{sec:clustering} is meant to indicate that the algorithms of this section utilize $\Multifilter$ indirectly.}
 	\label{fig:algo-depen}
\end{figure}

\subsection{Related work}\label{ssec:related}
Here we survey the most relevant prior work on learning mixture models
and robust statistics.

\paragraph{Mixture models.} %
The closest line of work to our results studies efficiently clustering mixture models 
under mean-separation conditions, and in particular Gaussian mixtures~\cite{dasgupta1999learning,vempala2004spectral,achlioptas2005spectral,dasgupta2007probabilistic,arora2005learning, kumar2010clustering,awasthi2012improved,regev2017learning,DiakonikolasKS18,hopkins2018mixture,kothari2018robust,mixon2017clustering}.
As mentioned previously, within the class of algorithms with 
runtime $\widetilde{O}(\frac{nd}{\alpha})$,
\cite{achlioptas2005spectral} achieves the best known mean separation condition 
(scaling as $\Theta(\alpha^{-\half})$, where $\alpha$ is the minimum component weight) 
for clustering mixtures of Gaussians with bounded covariance. 
This separation condition is nearly-matched (within a logarithmic factor) 
by our almost-linear time algorithm (Theorem~\ref{thm:infcluster}), 
which in addition is robust to outliers and generalizes to broader distribution families. 
We note that for the special case where the covariances are all known to be {\em exactly} the identity, 
prior to~\cite{achlioptas2005spectral}, \cite{vempala2004spectral} gave a similar algorithm 
attaining the same runtime of $\widetilde{O}(\frac{nd}{\alpha})$, 
under a weaker separation condition (scaling as roughly $\alpha^{-\frac 1 4}$). 
We are not aware of algorithms with this runtime and separation condition 
for clustering Gaussian mixtures when the covariances are only \emph{spectrally bounded by} the identity.

Subsequent work generalized the statistical setting studied in~\cite{achlioptas2005spectral, arora2005learning,kumar2010clustering,awasthi2012improved}, 
by improving on the separation condition using more sophisticated algorithmic tools, see, e.g.~\cite{DiakonikolasKS18,hopkins2018mixture,kothari2018robust, DiakonikolasK20}.
More recent work developed efficient algorithms for clustering mixtures of Gaussians, 
in the presence of a small constant fraction of outliers, under even weaker (algebraic) separation conditions~\cite{bakshi2020outlier,diakonikolas2020robustly, BakshiDHKKK20}. 
Beyond clustering, stronger notions of learning have been studied in this setting, 
including parameter estimation~\cite{moitra2010settling, belkin2015polynomial,hardt2015tight},
proper learning~\cite{feldman2006pac,daskalakis2014faster,acharya2014near,li2017robust,ashtiani2018nearly},
and their robust analogues~\cite{kane2021robust, liu2020settling, bakshi2020robustly}.
All of the aforementioned algorithms are statistically and computationally intensive, 
in particular have sample complexities and runtimes scaling super-polynomially with the number of components.
Finally, we acknowledge a related line of work studying learning in smoothed settings~\cite{hsu2013learning,anderson2014more,bhaskara2014smoothed,ge2015learning} and density estimation~\cite{devroye2012combinatorial,chan2014efficient,acharya2017sample}. 
These latter results are orthogonal to the results of the current paper.

\paragraph{Robust statistics and list-decodable learning.} 
Since the pioneering work from the statistics community in the 1960s and 1970s~\cite{anscombe1960rejection,tukey1960survey,huber1964robust,tukey1975mathematics}, 
there has been a tremendous amount of work on designing robust estimators, e.g.~\cite{HampelRRS86,huber2004robust}.
However, as discussed earlier, the estimators proposed in the statistics community are intractable to compute in high dimensions. 
The first algorithmic progress on robust statistics in high dimensions came in two independent works from the theoretical computer science community~\cite{diakonikolas2019robust,lai2016agnostic}.
Since then, there has been an explosion of work in this area, resulting in 
computationally efficient estimators for a range of increasingly complex tasks, 
including the aforementioned work on robust clustering, amongst many others, e.g.\ \cite{DiakonikolasKK017,balakrishnan2017computationally,cheng2018robust,klivans2018efficient,DiakonikolasKS19,prasad2018robust,diakonikolas2019sever,DiakonikolasKKP19,tran2018spectral}.
For a more complete account, the reader is referred to~\cite{Li18, Steinhardt18, diakonikolas2019recent}.

We also highlight a line of work, relevant to our main result, 
which combines tools from robust statistics with ideas from continuous optimization 
to achieve near-linear runtimes for high-dimensional robust estimation 
tasks~\cite{cheng2019high,cheng2019faster,DongH019,LiY20, jambulapati2020robust}.
Importantly, these algorithms only work in the regime where the fraction of corrupted samples is small, i.e.\ when $\alpha \to 1$.

The list-decodable learning setting we consider (i.e.\ when the trusted proportion of the data is $\alpha < \half$) was first considered in~\cite{balcan2008discriminative,CharikarSV17}.
In particular,~\cite{CharikarSV17} gave the first polynomial-time algorithm 
with near-optimal statistical guarantees for the problem of list-decodable mean estimation 
under bounded covariance assumptions. Shortly thereafter, efficient list-decodable mean estimators with near-optimal
error guarantees were developed under stronger distributional assumptions~\cite{DiakonikolasKS18, kothari2018robust}.
More recently, a line of work developed list-decodable learners for more challenging tasks, 
including linear regression~\cite{raghavendra2020list,karmalkar2019list} and subspace recovery~\cite{raghavendra2020list,bakshi2020list}. 
Similar techniques were also crucial in the recent development of robust clustering algorithms we previously described.

The most directly related prior research to the current paper is the sequence of recent papers 
developing faster algorithms for list-decodable mean estimation~\cite{CherapanamjeriMY20,DiakonikolasKK20,DiakonikolasKKLT20}.
We note that the algorithms in~\cite{CherapanamjeriMY20,DiakonikolasKKLT20} critically use projection of the data 
onto a $O(\frac 1 \alpha)$-dimensional subspace, and therefore are bottlenecked by the cost of this projection, 
yielding $\Omega(\frac {nd}{\alpha})$ runtimes. In the regime that $\alpha$ is a fixed constant, these works 
achieve runtimes which are linear in $n$ and $d$, by reinterpreting the problem of list-decodable mean estimation 
in a way which is amenable to speedups via tools from continuous optimization 
(specifically, regret guarantees over the ``$k$-Fantopes'', which capture Ky Fan norms in hindsight). 
On the other hand, the multifilter approach of~\cite{DiakonikolasKK20} only uses one-dimensional projections. 
However, their algorithm and its analysis does not have a direct interpretation as a continuous optimization method, 
using more problem-specific potentials, which guarantee termination after $n$ iterations. 

In many ways, the algorithm we develop in this paper can be viewed as the synthesis of these two approaches, 
by incorporating the ideas of~\cite{CherapanamjeriMY20,DiakonikolasKKLT20} to find better univariate projections 
for the multifilter of~\cite{DiakonikolasKK20}, and then designing a new potential function inspired 
by regret analyses of matrix multiplicative weights to demonstrate rapid termination of our fast multifilter.

\subsection{Organization}

In Section~\ref{sec:prelims}, we define notation and recall tools from prior work. We also give a technical overview of our one-shot potential approach to fast robust mean estimation. In Section~\ref{sec:gaussian}, we give a fast multifilter implementation in the Gaussian setting, as a simplified introducion to our techniques. In Section~\ref{sec:boundedcov}, we give our full fast multifilter for bounded-covariance distributions, proving Theorem~\ref{thm:inflistmean}. In Section~\ref{sec:clustering}, we give our applications to clustering mixture models, proving Theorem~\ref{thm:infcluster}.
See Figure~\ref{fig:algo-depen} for a pictorial depiction of the organization of the remainder of the paper.

\section{Preliminaries}\label{sec:prelims}

In Section~\ref{ssec:notation}, we define the notation used throughout this paper. Next, we recall some technical tools (primarily from prior work) which we draw upon in Section~\ref{ssec:priortools}. We conclude with a sketch of our potential function approach to filtering in Section~\ref{ssec:fastfilter} by demonstrating how it works for robust mean estimation in the minority-outlier regime, giving an alternative approach to obtaining the runtimes of \cite{DongH019}. This new approach bypasses an explicit matrix multiplicative weights argument in favor of a one-step potential. We ultimately generalize this technique to the list-decodable setting by interlacing it with clustering steps, inspired by the multifilter algorithm of \cite{DiakonikolasKS18, DiakonikolasKK20}.

\subsection{Notation}\label{ssec:notation}

\paragraph{General notation.} For mean $\mu \in \R^d$ and positive semidefinite covariance matrix $\msig \in \R^{d \times d}$, we let $\Nor(\mu, \msig)$ be the standard multivariate Gaussian. For $d \in \N$ we let $[d] \defeq \{j \mid j \in \N, 1 \le j \le d\}$. We refer to the $\ell_p$ norm of a vector argument by $\norm{\cdot}_p$, and overload this to mean the Schatten-$p$ norm of a symmetric matrix argument. The all-ones vector (when the dimension is clear from context) is denoted $\1$. The (solid) probability simplex is denoted $\Delta^n \defeq \{x \in \R^{n}_{\ge 0}, \norm{x}_1 \le 1\}$. We refer to the $i^{\text{th}}$ coordinate of a vector $v$ by $[v]_i$.

\paragraph{Matrices.} Matrices will be in boldface throughout, and when the dimension is clear from context we let $\mzero$ and $\id$ be the zero and identity matrices. The set of $d \times d$ symmetric matrices is $\Sym^d$ and the $d \times d$ positive semidefinite cone is $\Sym_{\ge 0}^d$. We use the Loewner partial ordering $\preceq$ on $\Sym_{\ge 0}^d$. The largest eigenvalue, smallest eigenvalue, and trace of a matrix are given by $\lmax(\cdot), \lmin(\cdot), \Tr(\cdot)$ respectively. We use $\normop{\cdot}$ to mean the ($\ell_2$-$\ell_2$) operator norm, which is the largest eigenvalue for arguments in $\Sym_{\ge 0}^d$. The inner product on $\ma, \mb \in \Sym^d$ is given by $\inprod{\ma}{\mb} \defeq \Tr(\ma \mb)$. 

\paragraph{Distributions.} We often associate a weight vector $w \in \Delta^n$ with a set of points $T \subset \R^d$ with $|T| = n$. Typically we denote this set by $\{X_i\}_{i \in T}$, where we overload $T$ to mean the indices as well as the points. For any $T' \subseteq T$ we let $w_{T'} \in \Delta^n$ be the vector which agrees with $w$ on $T'$ and is $0$ elsewhere. The empirical mean and covariance matrix on any subset are denoted
\[\mu_w(T') \defeq \sum_{i \in T'} \frac{w_i}{\norm{w'_T}_1} X_i,\; \cov_w(T') \defeq \sum_{i \in T'} \frac{w_i}{\norm{w'_T}_1} \Par{X_i - \mu_w(T')}\Par{X_i - \mu_w(T')}^\top.\]
For convenience, we also define the unnormalized convariance matrix by
\[\tcov_w(T') \defeq \sum_{i \in T'} w_i \Par{X_i - \mu_w(T')}\Par{X_i - \mu_w(T')}^\top.\]
We say distribution $\dist$ with $\E_{X \sim \dist}[X] = \mus$ has sub-Gaussian parameter $\sigma$ in all directions if $\E_{X \sim \dist}[\exp(s\inprod{X - \mus}{v})] \le \exp(\frac{\sigma^2 s^2}{2})$ for all unit vectors $v$. In Section~\ref{sec:clustering} we use concentration properties of sub-Gaussian random variables, which are well-known and can be found e.g.\ in the reference \cite{RigolletH17}.

\paragraph{List-decodable mean estimation.} We state the model of list-decodable mean estimation we use throughout the paper; the setting we consider is standard from the literature, and this description is repurposed from \cite{DiakonikolasKKLT20}. Fix a parameter $0 < \alpha < \thalf$. Then a set $T \defeq \{X_i\}_{i \in T} \subset \R^d$ of size $|T| = n = \Theta(d\alpha^{-1})$ is given, containing a subset $S \subset T$ such that the following assumption holds.

\begin{assumption}\label{assume:sexists}
There is a subset $S \subseteq T \subset \R^d$ of size $\alpha n = \Theta(d)$, and a vector $\mus \in \R^d$, such that
\[
\frac{1}{|S|}\sum_{i \in S} (X_i - \mus)(X_i - \mus)^\top \preceq \id.
\]
\end{assumption}

We remark that this assumption is motivated by the statistical model where there is an underlying distribution $\dist$ supported on $\R^d$ with mean $\mus$ and covariance bounded by $\id$, and the dataset $T$ is formed by $\alpha n$ independent draws from $\dist$ combined with $(1 - \alpha)n$ arbitrary points. Up to constants in the ``good'' fraction $\alpha$ and the covariance bound, Proposition B.1 of \cite{CharikarSV17} guarantees Assumption~\ref{assume:sexists} holds with inverse-exponential failure probability. We also note that Proposition 5.4(ii) of \cite{DiakonikolasKS18} shows that the information-theoretic optimal guarantee for list-decodable estimation in the setting of Assumption~\ref{assume:sexists} is to return a list $L$ of candidate means, where $|L| = \Theta(\alpha^{-1})$, and
\begin{equation}\label{eq:optimalrate}\min_{\mu \in L} \norm{\mu - \mus}_2 = \Theta\Par{\frac{1}{\sqrt{\alpha}}}\end{equation}
This setup handles the more general case where the upper bound matrix in Assumption~\ref{assume:sexists} is $\sigma^2\id$ for some positive parameter $\sigma$ by rescaling the space appropriately, and the error guarantee \eqref{eq:optimalrate} becomes worse by a factor of $\sigma$. Because of this, we will set $\sigma = 1$ throughout for simplicity. 

Finally, throughout Sections~\ref{sec:gaussian} and~\ref{sec:boundedcov} we will make the explicit assumption that $d \ge \alpha^{-1}$; for the regime where this is not the case, Algorithm 14 of \cite{DiakonikolasKKLT20} obtains optimal error rates in time $\tO(\alpha^{-2})$ (and in fact, if we tolerate a list size of $O(\alpha^{-1}\log \frac 1 \delta)$ where $\delta \in (0, 1)$ is the failure probability of the algorithm, obtains optimal error in time $\tO(\alpha^{-1})$).

\subsection{Technical tools}\label{ssec:priortools}

We will use a number of technical tools throughout this work which we list here for convenience. The first few are standard facts about covariance matrices which follow directly from computation. 

\begin{fact}[Convexity of covariance]\label{fact:covconvex}
For any $w \in \Delta^n$ associated with $T \subset \R^d$,
\[\mu_w(T)\mu_w(T)^\top \preceq \sum_{i \in T} \frac{w_i}{\norm{w}_1} X_iX_i^\top.\]
This implies that for any $v \in \R^d$,
\[(\mu_w(T) - v)(\mu_w(T) - v)^\top \preceq \sum_{i \in T} \frac{w_i}{\norm{w}_1} (X_i - v)(X_i - v)^\top. \]
\end{fact}

\begin{fact}[Effect of mean shift]\label{fact:meanshift}
For any $w \in \Delta^n$ associated with $T \subset \R^d$, and any $v \in \R^d$,
\[\sum_{i \in T} \frac{w_i}{\norm{w}_1} (X_i - v)(X_i - v)^\top = \cov_w(T) + (\mu_w(T) - v)(\mu_w(T) - v)^\top \succeq \cov_w(T).\]
\end{fact}

\begin{fact}[Alternate covariance characterization]\label{fact:altcovpairs}
For any $w \in \Delta^n$ associated with $T \subset \R^d$,
\[\frac{1}{2\norm{w}_1^2}\Par{\sum_{i, j \in T} w_iw_j(X_i - X_j)(X_i - X_j)^\top} = \cov_w(T).\]
\end{fact}

Next, we need the notion of safe weight removal in the list-decodable setting, adapted from \cite{DiakonikolasKKLT20}. The idea behind safe weight removal is that repeatedly performing a downweighting operation with respect to scores satisfying a certain condition results in weights which preserve some invariant, which we call saturation. We define our notions of safety and saturation, and state a key technical lemma which lets us reason about when saturation is preserved. In the following discussion assume we are in the list-decodable mean estimation setting we defined in Section~\ref{ssec:notation}.

\begin{definition}[$\gamma$-saturated weights]\label{def:saturation}
We call weights $w \in \Delta^n$ \emph{$\gamma$-saturated}, for some $\gamma > 1$, if $w \le \tfrac 1n \1$ entrywise, and
\[\norm{w_S}_1 \ge \alpha \norm{w}_1^{\frac{1}{\gamma}}.\]
\end{definition}

\begin{definition}[$\gamma$-safe scores]\label{def:safety}
We call scores $\{s_i\}_{i \in T} \in \R_{\ge 0}^n$ \emph{$\gamma$-safe with respect to $w$} if $w \in \Delta^n$ and
\[\sum_{i \in S} \frac{w_i}{\norm{w_S}_1} s_i \le \frac{1}{\gamma} \sum_{i \in T}\frac{w_i}{\norm{w_T}_1} s_i.\]
\end{definition}

Roughly speaking, we require this alternative notion of safe scores in the majority-outlier regime because there are less good points we can afford to throw away; see \cite{DiakonikolasKKLT20} for additional exposition. The key property connecting these two definitions is the following.

\begin{lemma}\label{lem:safeimpliessat}
Let $w^{(0)} \in \Delta^n$ be $\gamma$-saturated, and consider any algorithm of the form:
\begin{enumerate}
	\item For $0 \le t < N$:
	\begin{enumerate}
		\item Let $\{s_i^{(t)}\}_{i \in T}$ be $\gamma$-safe with respect to $w^{(t)}$.
		\item Update for all $i \in T$:
		\[w_i^{(t + 1)} \gets \Par{1 - \frac{s_i^{(t)}}{s_{\max}^{(t)}}}w_i^{(t)}, \text{ where } s_{\max}^{(t)} \defeq \max_{i \in T\mid w_i^{(t)} \neq 0} s_i^{(t)}.\]
	\end{enumerate}
\end{enumerate}
Then, $w^{(N)}$ is also $\gamma$-saturated.
\end{lemma}
\begin{proof}
It suffices to prove this in the case $N = 1$ and then use induction. Define
\[\delta_S \defeq \sum_{i \in S} \frac{w^{(0)}_i  - w^{(1)}_i }{\norm{w^{(0)}_S}_1},\; \delta_T \defeq \sum_{i \in T} \frac{w^{(0)}_i  - w^{(1)}_i }{\norm{w^{(0)}}_1}.\]
By using the assumption that $s^{(0)}$ is $\gamma$-safe, we conclude
\[\delta_S = \frac{1}{s_{\max}^{(0)}} \sum_{i \in S} \frac{w_i^{(0)}}{\norm{w^{(0)}_S}_1} s_i^{(0)} \le \frac{1}{\gamma} \cdot \frac{1}{s_{\max}^{(0)}} \sum_{i \in T} \frac{w_i^{(0)}}{\norm{w^{(0)}}_1} s_i^{(0)} = \frac{1}{\gamma}\delta_T.\]
Now, using $\gamma$-saturation of $w^{(0)}$ and $1 - \gamma^{-1}\delta \ge (1 - \delta)^{\gamma^{-1}}$ for all $\delta \in [0, 1]$ and $\gamma > 1$,
\[\norm{w^{(1)}_S}_1 = (1 - \delta_S)\norm{w^{(0)}_S}_1 \ge \Par{1 - \delta_T}^{\gamma^{-1}}\norm{w^{(0)}_S}_1 \ge \alpha \Par{(1 - \delta_T)\norm{w^{(0)}}_1}^{\gamma^{-1}} = \alpha\norm{w^{(1)}}_1^{\gamma^{-1}}.\]
\end{proof}

Finally, we include a technical lemma proved in \cite{CherapanamjeriMY20, DiakonikolasKKLT20}. 

\begin{lemma}[Lemma 2, \cite{DiakonikolasKKLT20}]\label{lem:covbounddist}
Let $w \in \Delta^n$ have $w \le \tfrac 1 n \1$ entrywise, and $\norm{w}_1 \ge \alpha^2$. Then,
\[\norm{\mu_w(T) - \mus}_2 \le \sqrt{2\norm{\cov_w(T)}_{\textup{op}}\frac{\norm{w}_1}{\norm{w_S}_1} + \frac{2}{\alpha}}.\]
\end{lemma}

In light of the lower bound of \cite{DiakonikolasKS18}, Lemma~\ref{lem:covbounddist} shows to learn the mean near-optimally in the list-decodable setting, it suffices to ensure $\norm{w_S}_1 = \Omega(\alpha)$ (i.e.\ we retain a constant fraction of the ``good'' weight) and $\normop{\cov_w(T)} = \tO(1)$ (i.e.\ the weighted covariance of the dataset is bounded).

\subsection{Potential function approach to fast filtering}\label{ssec:fastfilter}

In this section, we outline an example of a potential function approach to fast filtering, an alternative to filtering based on matrix multiplicative weights (MMW) used in recent literature \cite{DongH019, LiY20, DiakonikolasKKLT20}.\footnote{MMW guarantees are also implicitly used in approaches based on packing SDPs, see e.g.\ \cite{cheng2019high, cheng2019faster, CherapanamjeriMY20}. However, \cite{DongH019, LiY20, DiakonikolasKKLT20} use MMW guarantees in a non-black-box way to design filters.} This replacement is very useful in the list-decodable setting, as it greatly simplifies the requirements of our fast multifilter which interlaces clustering and filtering steps.

The example problem we consider in this expository section is the minority-outlier regime for robust mean estimation, when the ``ground truth'' distribution has covariance norm bounded by $\id$. We briefly describe the approach of \cite{DongH019} for this problem, and explain how it can be replaced with our new potential function framework. Throughout this section, fix some $0 < \eps < \thalf$ and assume that amongst the dataset $T \subset \R^d$ of $n$ points, there is a majority subset $S \subset T$ of size $|S| = (1 - \eps)n$ with bounded empirical covariance: $\cov_{\frac 1 n \1}(S) \preceq \id$.

The main algorithmic step in \cite{DongH019} is an efficient subroutine for halving the operator norm of the empirical covariance while filtering more weight from $T \setminus S$ than from $S$. It is well-known in the literature that whenever the operator norm is $O(1)$, the empirical mean is within $O(\sqrt{\eps})$ in $\ell_2$ norm from the ground truth mean (for an example of this derivation, see Lemma 3.2 of \cite{DongH019}). Thus, the key technical challenge is to provide a nearly-linear time implementation of this subroutine. This was accomplished in \cite{DongH019} using MMW-based regret guarantees, with ``gain matrices'' given by covariances and iterative filtering based on MMW responses. The result was a procedure which either terminates with a good mean estimate, or halves the covariance operator norm after $O(\log d)$ rounds of filtering. The latter is an artifact of many regret minimization techniques, which only guarantee progress after multiple rounds. It is natural to ask instead whether an alternative one-shot potential decrease guarantee exists; we now describe such a guarantee.

\paragraph{One-shot potential decrease.} Our algorithm will proceed in a number of iterations, where we modify a weight vector in $\Delta^n$ associated with $T$. We initialize $w^{(0)} \gets \tfrac{1}{n}\1$. In iteration $t$, we will downweight $w^{(t)} \in \Delta^n$ to obtain a new vector $w^{(t + 1)}$ as follows. Define the matrices
\[\mm_t \defeq \tcov_{w^{(t)}}\Par{T} = \sum_{i \in T} w^{(t)}_i (X_i - \mu_{w^{(t)}}(T))(X_i - \mu_{w^{(t)}}(T))^\top,\; \my_t \defeq \mm_t^{\log d}.\]
The potential we will track is $\Phi_t \defeq \Tr(\my_t^2)$. In order to analyze $\Phi_t$, we require two helper facts.

\begin{fact}[Lemma 7, \cite{jambulapati2020robust}]\label{fact:eltapp}
Let $\ma \succeq \mb$ be matrices in $\Sym_{\ge 0}^d$, and let $p \in \N$. Then
\[\Tr(\ma^{p - 1}\mb) \ge \Tr(\mb^p).\]
\end{fact}

\begin{fact}\label{fact:alphasplit}
For any $\gamma \ge 0$ and $\ma \in \Sym_{\ge 0}^d$,
\[\gamma\Tr(\ma^{2\log d}) \le \Tr(\ma^{2\log d + 1}) + d\gamma^{2\log d + 1}.\]
\end{fact}
\begin{proof}
This is immediate since each of the $d$ eigenvalues of $\ma ^{2\log d}$ is either at least $\gamma^{2 \log d}$ or not, and both of these cases are accounted for on the right hand side of the conclusion.
\end{proof}

We now give the potential analysis. Our main goal will be ensuring that
\begin{equation}\label{eq:filtercondex}\inprod{\my_t^2}{\mm_{t + 1}} \le 20\Tr(\my_t^2).\end{equation}
The specific constant in the above equation is not particularly important, but is used for illustration. We now show how \eqref{eq:filtercondex} implies a rapid potential decrease. Observe that
\begin{equation}\label{eq:potargument}
\begin{aligned}
\Phi_{t + 1} &= \Tr\Par{\mm_{t + 1}^{2 \log d}} \le \frac{1}{40} \Tr\Par{\mm_{t + 1}^{2 \log d + 1}} + d(40)^{2 \log d} \\
&\le \frac{1}{40} \Tr\Par{\mm_{t}^{2 \log d} \mm_{t + 1}} + d(40)^{2 \log d} \\
&\le \half \Tr(\my_t^2) + d(40)^{2 \log d} = \half \Phi_t + d(40)^{2 \log d}.
\end{aligned}
\end{equation}
The first line used Fact~\ref{fact:alphasplit} with $\gamma = 40$, the second used Fact~\ref{fact:eltapp} with $\ma = \mm_t$ and $\mb = \mm_{t + 1}$ (noting that if $w^{(t + 1)} \le w^{(t)}$ entrywise, the unnormalized covariance matrices respect the Loewner order by Fact~\ref{fact:meanshift}), and the third line used the assumption \eqref{eq:filtercondex}. This implies that we decrease $\Phi_t$ by a constant factor in every iteration, until it is roughly $d(40)^{2\log d}$, at which point the definition $\Phi_t = \Tr(\mm_t^{2\log d})$ implies that $\normop{\mm_t}$ is bounded by a constant. By using a na\"ive filtering preprocessing step, we can guarantee that $\Phi_0 = d^{O(\log d)}$, and hence the process will terminate in $O(\log^2 d)$ rounds.

\paragraph{Meeting the filter criterion \eqref{eq:filtercondex}.} To complete the outline of this algorithm, we need to explain how to satisfy \eqref{eq:filtercondex} via downweighting, while ensuring that we remove more weight from $T \setminus S$ than $S$. To do so, we define scores
\[s_i^{(t)} \defeq (X_i - \mu_{w^{(t)}}(T))^\top \mm_t^{2 \log d} (X_i - \mu_{w^{(t)}}(T)) \text{ for all } i \in T. \]
We remark that (randomized) constant-factor approximations can be computed to all $s_i^{(t)}$ via Johnson-Lindenstrauss projections in $\tO(nd)$ time, but for this discussion we assume we exactly know all scores. Then, by linearity of trace the condition \eqref{eq:filtercondex} is implied by
\begin{equation}\label{eq:sufficientfilter} \sum_{i \in T} w^{(t + 1)}_i  s_i^{(t)} \le 20\Tr(\my_t^2),\end{equation}
since Fact~\ref{fact:meanshift} implies that $\inprod{\my_t^2}{\mm_{t + 1}} \le \sum_{i \in T} w^{(t + 1)}_i  s_i^{(t)}$. Finally, we note that whenever \eqref{eq:sufficientfilter} does not hold, it must be primarily due to the effect of the outliers $T \setminus S$, because the covariance of $S$ is bounded. Hence, we can simply set
\[w_i^{(t + 1)} = \Par{1 - \frac{s_i^{(t)}}{s_{\max}^{(t)}}}^K w_i^{(t)},\]
where $K$ is the smallest natural number which passes the criterion \eqref{eq:sufficientfilter}. For any smaller $K$, it can be shown that downweighting ``one more time'' preserves the invariant that more outlier mass is removed, precisely because \eqref{eq:sufficientfilter} has not been met. Finally, binary searching to find the smallest value of $K$ meeting \eqref{eq:filtercondex} yields a complete algorithm running in $\tO(nd)$ time (for further details on the implementation of this binary search on $K$, see Theorem 2.4 of \cite{DongH019}).

\paragraph{Generalizing to the majority-outlier regime.} Our algorithms for the list-decodable setting marry this potential function argument with a multifilter, which produces multiple candidate filtered weight vectors on an input weight vector. We will instead show that for the \emph{tree} of weight vectors, where a node has children given by the candidates produced from the multifilter, at least one child both halves the potential and performs only $\gamma$-safe weight removal, for some $\gamma$. After polylogarithmic layers, we will return all empirical means of leaf nodes as our list of estimates. The child on the ``safe branch'' will then have a bounded potential and a $\gamma$-saturated weight vector, which suffices to guarantee an accurate mean estimate.

There are a number of additional complications which arise in this extension, which we briefly mention here as a preface to the following Sections~\ref{sec:gaussian} and~\ref{sec:boundedcov}. In order to process every layer of the multifilter tree in almost-linear time, we need to ensure that the number of datapoints across all the nodes, including repetitions, has not grown by more than a constant factor. The multifilter of \cite{DiakonikolasKK20} gives a variant of this guarantee by tracking the sums of the squared $\ell_1$ norms of weights associated with different nodes as a nonincreasing potential, i.e.\ 
\[\sum_{i \in \mathcal{S}} \norm{w^{(i)}}_1^2,\]
where $\mathcal{S}$ is the set of nodes and each $w^{(i)}$ is a current candidate weight function in $\mathcal{S}$. This is not sufficiently strong of a guarantee in our setting, since even points with very small weights need to be factored into calculations and thus affect runtime. We modify this approach in two ways. First, we replace the downweighting step with a randomly subsampled filter, which we show preserves various safety conditions such as those in Lemma~\ref{lem:safeimpliessat} with high probability. Next, we replace the squared $\ell_1$ potential with one involving $1 + \beta$ powers, for some $\beta \in (0, 1)$, which we prove is compatible with the multifilter. Overall, our filter tree contains polylogarithmically many layers, each of which accounts for sets with total cardinality $O(n^{1 + O(\beta)})$, giving us our final runtime.

\section{Warmup: fast Gaussian multifilter}\label{sec:gaussian}

As a warmup to our later (stronger) developments in Section~\ref{sec:boundedcov}, we give a complete algorithm for list-decodable mean estimation in the Gaussian case, i.e.\ where the ``true'' distribution $\dist$ is drawn from a Gaussian with covariance bounded by $\id$. Conceptually, the types of statements Gaussian concentration (rather than heavy-tailed concentration) allow us to make let us simplify several of the technical difficulties alluded to at the end of Section~\ref{ssec:fastfilter}, in particular the following.
\begin{enumerate}
	\item Instead of a ``covariance bound'' statement such as \eqref{eq:filtercondex} to use in our potential proof, we will simply guarantee that the multifilter returns sets of points which lie in short intervals along a number of random directions given by a Johnson-Lindenstrauss sketch.
	\item Instead of a randomly subsampled filtering step to remove outliers without soft downweighting (to preserve truly small subsets), it will be enough to deterministically set thresholds along $1$-dimensional projections to safely remove the outliers.
	\item The definition of the Gaussian multifilter (see Section~\ref{ssec:gsorc}) will be substantially simpler, since we have more explicit tail bounds to check for outliers.
\end{enumerate}
The strength of the error guarantees of the simpler algorithm in this section are somewhat weaker than those of Section~\ref{sec:boundedcov} even when specialized to the Gaussian case, but we include this section as an introductory exposition of our techniques. We will use the stronger Gaussian concentration assumption in this section, a tightening of Assumption~\ref{assume:sexists}.

Throughout this section, we will assume that $\alpha \in [1/d, 1 / \log^C d]$, for some constant $C > 0$.
We claim this is without loss of generality.
Specifically, for $\alpha^{-1}$ sub-logarithmic in the dimension $d$, the algorithm in the prior work by \cite{DiakonikolasKKLT20} runs in nearly-linear time. 
On the other hand, randomly sampling the dataset solves the list-decodable mean estimation problem near-optimally in time $\tO(\alpha^{-1})$ (see Appendix A of \cite{DiakonikolasKKLT20} for a proof). 

We now formally define the regularity condition which we will use throughout this section.
\begin{assumption}\label{assume:gaussian}
There is a subset $S \subseteq T \subset \R^d$ of size $\alpha n = \Theta(d \cdot \textup{polylog}(d))$, and a vector $\mus \in \R^d$, such that for all unit vectors $v \in \R^d$ and thresholds $t \in \R_{\ge 0}$, 
\[\Pr_{i \sim_{\textup{unif}} S}[\inprod{X_i - \mus}{v} > t] \le \exp(-\Omega(t^2)) + \frac{1}{\Omega\Par{\log^3 d}}.\] Here, the notation $i \sim_{\textup{unif}} S$ means that $i$ is a uniformly random sampled index from $S$.
\end{assumption}

This assumption is standard in the literature, and follows when the true distribution which $S$ is sampled from is Gaussian with identity-bounded covariance (see e.g.\ Definition A.4, Lemma A.5 \cite{DiakonikolasKK017}). We remark that the sample complexity of Assumption~\ref{assume:gaussian} is worse than that of Assumption~\ref{assume:sexists} by a polylogarithmic factor. This lossiness is just to simplify exposition in this warmup section, and indeed in the following Section~\ref{sec:boundedcov} we give an algorithm which recovers stronger guarantees than Theorem~\ref{thm:gaussian}, this section's main export, under only Assumption~\ref{assume:sexists}.

In Section~\ref{ssec:gpartition}, we first give our main subroutine, $\GPartition$, which takes a candidate set and produces a number of children candidate sets which each satisfy a progress guarantee similar to \eqref{eq:filtercondex}. The main difficulty will be in guaranteeing that the children sets are sufficiently small, and that if the parent set was ``good'' (had large overlap with $S$), then at least one child set will as well. We reduce $\GPartition$ to a number of one-dimensional clustering steps, which we implement as $\GSOrC$ in Section~\ref{ssec:gsorc}. Finally, we use the guarantees of $\GPartition$ within our potential-based framework outlined in Section~\ref{ssec:fastfilter}, giving our final algorithm $\GMultifilter$ in Section~\ref{ssec:fullgaussian}. Throughout, sets $S$ and $T$ are fixed and satisfy Assumptions~\ref{assume:sexists} and~\ref{assume:gaussian}.

\subsection{Reducing \texorpdfstring{$\GPartition$}{GaussianPartition} to \texorpdfstring{$\GSOrC$}{GaussianSplitOrCluster}}\label{ssec:gpartition}

Our final algorithm creates a tree of candidate sets. Every node $p$ in the tree is associated with a subset $T_p$. In order to progress down the tree, at a given node $p$ we form children $\{c_\ell\}_{\ell \in [k]}$ with associated sets $\{T_{c_\ell}\}_{\ell \in [k]}$; we call the procedure which produces the children node $\GPartition$, and develop it in this section. There are three key properties of $\GPartition$ which we need.

\begin{enumerate}
	\item The sum of the cardinalities of $\{T_{c_\ell}\}_{\ell \in [k]}$ is not too large compared to $|T_p|$. This is to guarantee that at each layer of the tree, we perform about the same amount of work, namely $\tO(nd)$. We formalize this with a parameter $\beta \in (0, 1]$ throughout the rest of this section, and will guarantee that every time $\GPartition$ is called on a parent node $p$,
	\begin{equation}\label{eq:gsizepot}\sum_{\ell \in [k]} \left|T_{c_\ell}\right|^{1 + \beta} \le |T_p|^{1 + \beta}.\end{equation}
	\item If the parent vertex $p$ has substantial overlap with $S$ (at least $\thalf |S|$ points), then at least one of the produced children continues to retain all but a small fraction of points in $S$.
	\item Defining the matrices
	\begin{equation}\label{eq:gaussianmatdef}
	\begin{gathered}
	\mm_p \defeq \tcov_{\frac 1 n \1}\Par{T_p},\; \my_p \defeq \mm_p^{\log d},\\
	\mm_{c_\ell} \defeq \tcov_{\frac 1 n \1}\Par{T_{c_\ell}},\; \my_{c_\ell} \defeq \mm_{c_\ell}^{\log d} \text{ for all } \ell \in [k],
	\end{gathered}
	\end{equation}
	every $\mm_{c_\ell}$ satisfies the bound
	\begin{equation}\label{eq:gfiltercond}
	\inprod{\my_p^2}{\mm_{c_\ell}} \le R^2 \Tr\Par{\my_p^2},
	\end{equation}
	for some (polylogarithmic) value $R$ we will specify. Note the similarity between this and \eqref{eq:filtercondex}; this will be used in a potential analysis to bound progress on covariance operator norms.
\end{enumerate}

We are now ready to state the algorithm $\GPartition$.

\begin{algorithm}[ht!]
	\caption{$\GPartition(T_p, \alpha, \beta, C, R)$}\label{alg:gpartition}
	\begin{algorithmic}[1]
		\STATE \textbf{Input:} $T_p \subseteq T$, $\alpha \in (0,\thalf)$, $\beta \in (0, 1]$, $C, R \in \R_{\ge 0}$ satisfying (for sufficiently large constants)
		\[R = \Omega\Par{\sqrt{\log\Par{C}} \cdot \frac{\log\log \Par{C\alpha^{-1}}}{\beta}},\;C = \Omega\Par{\log^2 d}.\] 
		\STATE \textbf{Output:} With failure probability $\le \tfrac{1}{d^3}$: subsets $\{T_{c_\ell}\}_{\ell \in [k]}$ of $T_p$, satisfying \eqref{eq:gsizepot}. Every child satisfies \eqref{eq:gfiltercond} (using notation \eqref{eq:gaussianmatdef}). If $|T_p \cap S| \ge (\thalf + \tfrac 1 C)|S|$, at least one child $T_{c_\ell}$ satisfies
		\begin{equation}\label{eq:gkeepmostofs}\left|T_{c_\ell} \cap S\right| \ge \left|T_p \cap S\right| - \frac{1}{C}|S|.\end{equation}
		\STATE Sample $\ndir = \Theta(\log d)$ vectors $\{u_j\}_{j \in [\ndir]} \in \R^d$ each with independent entries $\pm 1$. Following notation \eqref{eq:gaussianmatdef}, let $v_j \gets \my_p u_j$ for all $j \in [\ndir]$.
		\STATE $\mathcal{S}_0 \gets T_p$
		\FOR{$j \in [\ndir]$}
		\STATE $\mathcal{S}_j \gets \emptyset$
		\FOR{$T' \in \mathcal{S}_{j - 1}$}
		\STATE $\mathcal{T} \gets \GPartitionOneD(T', \alpha, v_j, \beta, C\ndir, R)$
		\STATE $\mathcal{S}_j \gets \mathcal{S}_j \cup \mathcal{T}$
		\ENDFOR
		\ENDFOR
		\RETURN $\mathcal{S}_{\ndir}$
	\end{algorithmic}
\end{algorithm}

It heavily relies on a subroutine, $\GPartitionOneD(T', v)$ which takes a subset $T'$ and a vector $v \in \R^d$, and produces children subsets of $T'$ satisfying the first two conditions above, and also guarantees that along the direction $v$, each child subset is contained in a relatively short interval.

\begin{algorithm}[ht!]
	\caption{$\GPartitionOneD(T', \alpha, v, \beta, C, R)$}\label{alg:gpartitiononed}
	\begin{algorithmic}[1]
		\STATE \textbf{Input:} $T' \subseteq T$, $\alpha \in (0,\thalf)$, $v \in \R^d$, $\beta \in (0, 1]$, $C, R \in \R_{\ge 0}$ satisfying (for sufficiently large constants)
		\[R = \Omega\Par{\sqrt{\log\Par{C}} \cdot \frac{\log\log \Par{C\alpha^{-1}}}{\beta}},\;C = \Omega\Par{\log^3 d}.\] 
		\STATE \textbf{Output:} Subsets $\{T''_\ell\}_{\ell \in [k]} \subseteq T'$, such that
		\begin{equation}\label{eq:gpartitiononedsize}\sum_{\ell \in [k]} \left|T''_\ell\right|^{1 + \beta} \le |T'|^{1 + \beta}.\end{equation}
		If $|T' \cap S| \ge (\thalf + \tfrac{1}{C}) |S|$, at least one child $T''_\ell$ satisfies
		\[\left|T''_\ell \cap S \right| \ge \left|T' \cap S \right| - \frac{1}{C} |S|.\]
		Every child has all values $\{\inprod{v}{X_i} \mid i \in T''_\ell\}$ contained in an interval of length $R\norm{v}_2$.
		\STATE $\mathcal{S}_{\text{in}} \gets \{T'\}$, $\mathcal{S}_{\text{out}} \gets \emptyset$
		\WHILE{$\mathcal{S}_{\text{in}} \neq \emptyset$}
		\STATE $T'' \gets$ the first element of $\mathcal{S}_{\text{in}}$
		\STATE $\mathcal{S}_{\text{in}} \gets \mathcal{S}_{\text{in}} \setminus T''$
		\IF{$\GSOrC(T'', \alpha, \beta, R, \tfrac{1}{Cn})$ returns one set $T_{\text{out}}^{(0)}$}
		\STATE $\mathcal{S}_{\text{out}} \gets \mathcal{S}_{\text{out}} \cup \left\{\Tout^{(0)}\right\}$
		\ELSE
		\STATE $T^{(1)}_{\text{out}}, T^{(2)}_{\text{out}} \gets \GSOrC(T'', \alpha, \beta, R, \tfrac{1}{Cn})$
		\STATE $\mathcal{S}_{\text{in}} \gets \mathcal{S}_{\text{in}} \cup \left\{T_{\text{out}}^{(1)}, T_{\text{out}}^{(2)}\right\}$ 
		\ENDIF
		\ENDWHILE
	\end{algorithmic}
\end{algorithm}

Once again, $\GPartitionOneD$ heavily relies on a subroutine, $\GSOrC$, which we implement in Section~\ref{ssec:gsorc}. It takes as input a set $T''$ and either produces one or two subsets of $T''$ as output. If it outputs one set, that set has length at most $R\norm{v}_2$ in the direction $v$; otherwise, $\GPartitionOneD$ simply recurses on the additional two sets. Crucially, $\GSOrC$ guarantees that if $T''$ has substantial overlap with $S$, then so does at least one child; moreover, when $\GSOrC$ returns two sets, they satisfy a size potential such as \eqref{eq:gpartitiononedsize}. We now demonstrate correctness of $\GPartition$, assuming that $\GPartitionOneD$ is correct.

\begin{lemma}\label{lem:gpartition}
The output of $\GPartition$ satisfies the guarantees given in Line 2 of Algorithm~\ref{alg:gpartition}, assuming correctness of $\GPartitionOneD$.
\end{lemma}
\begin{proof}
First, to demonstrate that the subsets satisfy \eqref{eq:gsizepot}, we observe that we can view $\GPartition$ as always maintaining a set of subsets, $\mathcal{S}_j$ (in the beginning, $\mathcal{S}_0 = T_p$). The set $\mathcal{S}_j$ is formed by calling $\GPartitionOneD$ on elements of $\mathcal{S}_{j - 1}$, each of which satisfy \eqref{eq:gpartitiononedsize}, so inductively $\mathcal{S}_{\ndir} = \{T_{c_l}\}_{l \in [k]}$ will satisfy \eqref{eq:gsizepot} with respect to $\mathcal{S}_0 = T_p$ as desired.

Next, by recursively using the guarantee of $\GPartitionOneD$, every $T_{c_l} \in \mathcal{S}_{\ndir}$ will satisfy
\[\text{all values } \left\{\inprod{v_j}{X_i} \mid i \in T_{c_l}\right\} \text{ are contained in an interval of length } R\norm{v_j}_2, \text{ for all } j \in [\ndir].\]
In other words, this set is short along all the directions $\{\my_p u_j = v_j\}_{j \in [\ndir]}$. This lets us conclude
\begin{align*}
\inprod{\my_p^2}{\mm_{c_l}} &= \frac{1}{2n\left|T_{c_l}\right|} \inprod{\my_p^2}{\sum_{i, i' \in T_{c_l}} (X_i - X_{i'})(X_i - X_{i'})^\top} \\
&= \frac{1}{2n\left|T_{c_l}\right|} \sum_{i, i' \in T_{c_l}} \norm{\my_p(X_i - X_{i'})}_2^2 \\
&\le \frac{1.4}{2n\left|T_{c_l}\right|\ndir} \sum_{i, i' \in T_{c_l}} \sum_{j \in [\ndir]} \inprod{\my_p u_j}{X_i - X_{i'}}^2 \\
&\le \frac{1.4}{2n\left|T_{c_l}\right|\ndir} \sum_{i, i' \in T_{c_l}} \sum_{j \in [\ndir]} R^2 \norm{\my_p u_j}_2^2 \\
&\le \frac{1.4}{2\ndir} \sum_{j \in [\ndir]} R^2 \norm{\my_p u_j}_2^2 \le R^2 \Tr\Par{\my_p^2},
\end{align*}
with probability at least $1 - \tfrac{1}{2d^3}$. Here, we used Fact~\ref{fact:altcovpairs} in the first line and linearity of trace in the second line. The third line used the Johnson-Lindenstrauss lemma of \cite{Achlioptas03} which says that for any vector $v$, $\tfrac{1}{\ndir} \sum_{j \in [\ndir]} \inprod{u_j}{v}^2 \in [0.6, 1.4] \norm{v}_2^2$ for a sufficiently large $\ndir = \Theta(\log(d))$ with probability at least $1 - \tfrac{1}{2d^6}$, which we union bound over all $|T_{c_l}|^2 \le n^2 \le d^4$ pairs of points. The fourth line used the radius guarantee of $\GPartitionOneD$, and the fifth used $|T_{c_l}| \le n$ and the Johnson-Lindenstrauss lemma guarantee that $\tfrac{1}{\ndir} \sum_{j \in [\ndir]} \norm{\my_p u_j}_2^2 \in [0.6, 1.4] \Tr(\my_p^2)$ with probability at least $1 - \tfrac{1}{2d^3}$, which can be deduced by the guarantee of \cite{Achlioptas03} applied to the rows of $\my_p$. Union bounding over the two applications of \cite{Achlioptas03} yields the claim.

Finally, to demonstrate that at least one child satisfies \eqref{eq:gkeepmostofs}, suppose $p$ satisfies $|T_p \cap S| \ge (\thalf + \tfrac 1 C)|S|$ (i.e.\ it has substantial overlap with $S$). Then by applying the guarantee of $\GPartitionOneD$ inductively, every $\mathcal{S}_j$ will have at least one element $T'$ satisfying $|T' \cap S| \ge \thalf |S|$. Every call to $\GPartitionOneD$ only removes $\tfrac{1}{C\ndir} |S|$ points in $S$, so overall only $\tfrac{1}{C}|S|$ points are removed.
\end{proof}

\subsection{Implementation of \texorpdfstring{$\GSOrC$}{GaussianSplitOrCluster}}\label{ssec:gsorc}

In this section, we first state $\GSOrC$ and analyze its correctness. We conclude with a full runtime analysis of $\GPartitionOneD$, using our $\GSOrC$ implementation.

\begin{algorithm}[ht!]
	\caption{$\GSOrC(\Tin, \alpha, v, \beta, R, \Delta)$}\label{alg:gsorc}
	\begin{algorithmic}[1]
		\STATE \textbf{Input:} $\Tin \subseteq T$, $\alpha \in (0,\thalf)$, $v \in \R^d$, $\beta \in (0, 1]$, $R \in \R_{\ge 0}$, $\Delta \in (0, 1)$
		\STATE \textbf{Output:} Either one subset $\Tout^{(0)} \subset \Tin$, or two subsets $\Tout^{(1)}, \Tout^{(2)} \subset \Tin$. In the one subset case, $\Tout^{(0)}$ has $\left\{\inprod{v}{X_i} \mid i \in \Tout^{(0)}\right\}$ contained in an interval of length $R\norm{v}_2$. In the two subsets case, they take the form, for some threshold value $\tau \in \R$ and $r \defeq \tfrac{R}{4\kmax}, \kmax = \Theta\Par{\tfrac{\log\log(\frac{1}{\alpha\Delta})}{\beta}}$
		\begin{equation}\label{eq:gtwosetsdef}\Tout^{(1)} \defeq \{X_i \mid \inprod{v}{X_i} \le \tau + r\norm{v}_2\},\; \Tout^{(2)} \defeq \{X_i \mid \inprod{v}{X_i} \ge \tau - r\norm{v}_2\}, \end{equation}
		and satisfy
		\begin{equation}\label{eq:gtwosetssizebound}\left|\Tout^{(1)}\right|^{1 + \beta} + \left|\Tout^{(2)}\right|^{1 + \beta} < \left|\Tin\right|^{1 + \beta}.\end{equation}
		\STATE $Y_i \gets \inprod{v}{X_i}$ for all $i \in \Tin$
		\STATE $\Tout^{(0)} \gets $ indices in the middle $1 - \alpha\Delta$ quantiles of $\{Y_i\}_{i \in \Tin}$
		\IF{$\left\{Y_i \mid i \in \Tout^{(0)}\right\}$ is contained in an interval of length $R\norm{v}_2$}
		\RETURN $\Tout^{(0)}$
		\ELSE
		\STATE $\taumed \gets \med\Par{\Brace{Y_i \mid i \in \Tin}}$, where $\med$ returns the median
		\STATE $\tau_k \gets \taumed + 2kr\norm{v}_2$ for all integers $-\kmax \le k \le \kmax$
		\RETURN $\Tout^{(1)}$, $\Tout^{(2)}$ defined in \eqref{eq:gtwosetsdef} for any threshold $\tau_k$ inducing sets satisfying \eqref{eq:gtwosetssizebound}
		\ENDIF
	\end{algorithmic}
\end{algorithm}

To analyze Algorithm~\ref{alg:gsorc} we first demonstrate that it always returns in at least one case. In particular, we demonstrate that whenever the set $\Tout^{(0)}$ is not sufficiently short, then there will be a threshold parameter $k$ such that the induced sets in \eqref{eq:gtwosetsdef} satisfy the size bound \eqref{eq:gtwosetssizebound}.

\begin{lemma}\label{lem:gsorccorrect}
Suppose Algorithm~\ref{alg:gsorc} does not return on Line 6. Then, there exists a $k \in \mathbb{Z}$ in the range $-\kmax \le k \le \kmax$ such that Algorithm~\ref{alg:gsorc} is able to return on Line 10.
\end{lemma}
\begin{proof}
We instead prove that if there is no such $k$, then we will have a contradiction on the length of the set $\Tout^{(0)}$ in the direction $v$. We first lower bound the length of the $[\thalf, 1 - \tfrac{\alpha\Delta}{2}]$ quantiles of $\{Y_i \mid i \in \Tin\}$ by $\thalf R \norm{v}_2$; the lower bound for the $[\tfrac{\alpha\Delta}{2}, \thalf]$ quantiles will follow analogously. Combining shows that if no threshold works, then the algorithm should have returned $\Tout^{(0)}$.

For any threshold $\tau$, define $g(\tau) \in [0, 1]$ to be the proportion of $\{Y_i \mid i \in \Tin\}$ which are $\ge \tau$. Moreover, define for all $1 \le k \le \kmax$, 
\[\gamma_k \defeq g(\tau_k - r\norm{v}_2) = g(\taumed + (2k - 1) r\norm{v}_2),\]
and note that $\gamma_1 \le \thalf$ by definition, since $\taumed$ was the median. Now, for each $1 \le k \le \kmax$, since $\tau_k$ was not a valid threshold, the sets
\[T_k^{(1)} \defeq \{X_i \mid Y_i \le \taumed + (2k + 1)r\norm{v}_2\},\; T_k^{(2)} \defeq \{X_i \mid Y_i \ge \taumed + (2k - 1)r\norm{v}_2\}\]
do not satisfy the size bound \eqref{eq:gtwosetssizebound}. Normalizing both sides of \eqref{eq:gtwosetssizebound} by $|\Tin|^{1 + \beta}$ and using the definitions of $\{\gamma_k\}$, we obtain the following recursion:
\begin{equation}\label{eq:ggammarecursion}(1 - \gamma_{k + 1})^{1 + \beta} + \gamma_k^{1 + \beta} = \Par{\frac{\left|T_k^{(1)}\right|}{|\Tin|}}^{1 + \beta} + \Par{\frac{\left|T_k^{(2)}\right|}{|\Tin|}}^{1 + \beta} \ge 1 \implies \gamma_{k + 1} \le \gamma_k^{1 + \beta}.\end{equation}
To obtain the above implication, we used $1 - (1 - x)^{1 + \beta} > x^{1 + \beta}$ for all $x, \beta \in [0, 1]$. By repeatedly applying the recursion \eqref{eq:ggammarecursion}, we have
\[\gamma_{\kmax} \le \gamma_1^{(1 + \beta)^{(\kmax - 1)}} \le \Par{\half}^{(1 + \beta)^{(\kmax - 1)}} \le \frac{\alpha\Delta}{2},\]
where we use the definition of $\kmax$ and $\gamma_1 \le \thalf$. Thus, the $[\thalf, 1- \tfrac{\alpha\Delta}{2}]$ quantiles are contained between $\taumed$ and $\taumed + (2\kmax - 1)r\norm{v}_2 \le \taumed + \thalf R\norm{v}_2$. By repeating this argument in the range $-\kmax \le k \le -1$, we obtain a contradiction (as Algorithm~\ref{alg:gsorc} should have returned $\Tout^{(0)}$).
\end{proof}

We next prove that if the input $T'$ to $\GPartitionOneD$ has large overlap with $S$, then the algorithm always returns some child $T''$ which removes at most $\tfrac 1 C |S|$ points from this overlap. This proof uses the implementation of $\GSOrC$ in a white-box way, as well as Assumption~\ref{assume:gaussian}.

\begin{lemma}\label{lem:gsorcpreserves}
Whenever $\GPartitionOneD$ is called on $T'$ with $|T' \cap S| \ge \Par{\thalf + \frac 1 C}|S|$ with parameters $R, C$ satisfying (for sufficiently large constants)
\[R = \Omega\Par{\sqrt{\log(C)} \cdot \frac{\log\log(Cd)}{\beta}},\; C = \Omega\Par{\log^3 d},\]
it produces some child $T''$ satisfying $|T'' \cap S| \ge |T' \cap S| - \frac 1 C |S|$.
\end{lemma}
\begin{proof}
We first discuss the structure of $\GPartitionOneD$. We say a call to $\GSOrC$ is a ``split step'' if it produces two sets, and otherwise we call it a ``cluster step.'' Every output child of $\GPartitionOneD$ is the result of a consecutive number of split steps, and then one cluster step. Also, every split step replaces an interval with its intersections with two half-lines which overlap by $2r\norm{v}_2 = \Omega(\sqrt{\log(C)}\norm{v}_2)$. Assume for simplicity that $\norm{v}_2 = 1$ in this proof; analogous arguments hold for all $v$ by scaling everything appropriately. Finally, we recall that all calls to $\GSOrC$ in $\GPartitionOneD$ are with $\Delta = \frac{1}{Cn}$.

Our key technical claim is that after any number of split steps forming a partition of the real line, there is always some interval such that $\inprod{v}{\mus}$ is $r$ away from both endpoints (in this proof, we allow intervals to have endpoints at $\pm \infty$). This is clearly true at the beginning, since the only interval is $(-\infty, \infty)$. Next, we induct and assume that on the current partition, after some number of split steps, there is an interval $[a, b]$ in the partition such that $\inprod{v}{\mus} \in [a + r, b - r]$. Consider the intersection of this interval with any split step, parameterized by the half-lines $(-\infty, \tau + r]$ and $[\tau - r, \infty)$ for some $\tau \in \R$. If $\inprod{v}{\mus} \ge \tau$, then one of the resulting intervals is
\[\Brack{\max\Par{a, \tau - r}, b}\]
where we note that this interval is non-degenerate by assumption; $\tau \le \inprod{v}{\mus} \le b - r \implies \tau - r \le b$. If the result of the $\max$ is $[a, b]$, then the claim holds; otherwise, the interval is $[\tau - r, b]$ and the claim holds by induction ($\inprod{v}{\mus} \le b - r$) and the assumption $\inprod{v}{\mus} \ge \tau$. The other case when $\inprod{v}{\mus} \le \tau$ follows symmetrically by considering the interval $[a, \min(b, \tau + r)]$.

Now, consider the partition of the real line which is induced by the eventual children outputted by $\GPartitionOneD$, \emph{right before} the last cluster step is applied to them (in other words, this partition is formed only by split steps). Using the above argument, there is some element of this partition $[a, b]$ so that $\inprod{v}{\mus} \in [a + r, b - r]$. Applying Assumption~\ref{assume:gaussian} shows that if we consider the effects of truncating the set $\{Y_i \mid i \in S\}$ at the endpoints of this interval, we remove at most a $\tfrac{1}{2C}$ fraction of the points from $S$. Finally, the interval that is returned is the result of a cluster step applied to this interval. This can only remove at most an $\alpha\Delta \le \tfrac{\alpha}{2C}$ fraction of the overall points, which is at most $\tfrac{1}{2C}|S|$. Combining these two bounds yields the claim.
\end{proof}

Finally, we conclude with a runtime analysis of $\GPartitionOneD$.

\begin{lemma}\label{lem:gponedruntime}
Let $n' \defeq |T'|$ for some $T' \subseteq T$. $\GPartitionOneD$ called on input $T'$ with parameter $C$ can be implemented to run in time
\[O\Par{n' d + (n')^{1 + \beta} \log n' \cdot \frac{\log\log(Cd)}{\beta}}.\]
\end{lemma}
\begin{proof}
We begin by forming all of the one-dimensional projections $\inprod{v}{X_i}$ for all $i \in T'$, and sorting these values. We also store the quantile of each point (i.e.\ the number of points larger than it). The total cost of these operations is $O(n' d + n' \log n')$.

Next, given this total ordering, observe that the structure of $\GPartitionOneD$ means that every set in $\mathcal{S}_{\textup{in}}$ is a subinterval of $T'$, since this is inductively preserved by calls to $\GSOrC$; hence, we can represent every set implicitly by its endpoints. Moreover, given access to the initial quantile information we can implement every call to $\GSOrC$ in time $O(\kmax \log n') = O(\log n' \cdot \tfrac{\log\log(Cd)}{\beta})$, since the cost of checking the length of $\Tout^{(0)}$ is constant, and the cost of checking each candidiate $\tau_k$ is dominated by determining the thresholds of the corresponding induced sets $\Tout^{(1)}$ and $\Tout^{(2)}$. These can be performed via binary searches in $O(\log n')$ time.

It remains to bound the number of calls to $\GSOrC$ throughout the execution of $\GPartitionOneD$. To this end, we bound the number of times $\GSOrC$ can return one set, and the number of times it can return two sets. Every time $\GSOrC$ returns one set, it adds it to $\mathcal{S}_{\text{out}}$, and by using the guarantee \eqref{eq:gtwosetssizebound} recursively, there can only ever be $(n')^{1 + \beta}$ such sets. Similarly, every time it returns two sets it increases $|\mathcal{S}_{\text{in}}| + |\mathcal{S}_{\text{out}}|$ by one, but we know at termination this is at most $(n')^{1 + \beta}$, and this potential never decreases. Thus, the total number of calls to $\GSOrC$ is bounded by $O((n')^{1 + \beta})$, as desired.
\end{proof}

As an immediate corollary, we obtain a runtime bound on $\GPartition$.

\begin{corollary}\label{cor:gpruntime}
Let $n_p \defeq |T_p|$ for some $T_p \subseteq T$. $\GPartition$ called on input $T_p$ with parameter $C$ can be implemented to run in time
\[O\Par{n_p^{1 + \beta} d \log^2(d) + n_p^{1 + \beta} \log^2(d)\cdot \frac{\log\log(Cd))}{\beta}}.\]
\end{corollary}
\begin{proof}
First, consider the cost of computing all vectors $\my_p u_j$. It is straightforward to implement matrix-vector multiplications through $\mm_p$ in time $O(n_p d)$, so this cost is $O(n_p d \log^2 (d))$.
	
We next require a bound on the cost of $\ndir = \Theta(\log d)$ consecutive calls to $\GPartitionOneD$. The cost of each is given by Lemma~\ref{lem:gponedruntime}, and the result follows by summing this cost over all elements of each $\mathcal{S}_j$, which can be bounded since for all $j \in [\ndir]$, the cardinalities of all sets contained in $\mathcal{S}_j$ have $1 + \beta$ powers bounded by $n_p^{1 + \beta}$ by repeatedly using the guarantee \eqref{eq:gpartitiononedsize}.
\end{proof}

\subsection{Full Gaussian algorithm}\label{ssec:fullgaussian}

Finally, we are ready to give our full algorithm for list-decodable mean estimation under Assumptions~\ref{assume:sexists} and~\ref{assume:gaussian}. We begin by reducing the original problem to a number of subproblems of bounded diameter (following \cite{DiakonikolasKKLT20}), and then showing that for each of these subproblems, polylogarithmic calls to $\GPartition$ yield subsets of bounded covariance operator norm. We conclude by recalling that a covariance operator norm bound suffices to yield guarantees on mean estimation.

\begin{algorithm}[ht!]
	\caption{$\GMultifilter(T,\alpha)$}\label{alg:gmultifilter}
	\begin{algorithmic}[1]
		\STATE \textbf{Input:} $T \subset \R^d$, $|T| = n$ satisfying Assumptions~\ref{assume:sexists} and~\ref{assume:gaussian} with parameter $\alpha \in (0,\thalf)$
		\STATE \textbf{Output:} With failure probability $\le \frac 1 d$: $L$ with $|L| = O(\tfrac 1 \alpha)$ such that some $\hmu \in L$ satisfies
		\begin{equation}\label{eq:gaussianerror}\norm{\hmu - \mus}_2 = O\Par{\frac{\log(d)\log\log^{1.5}(d)}{\sqrt \alpha}}.\end{equation}
		\STATE $\{T'_i\}_{i \in [k]} \gets \NaiveCluster(T)$
		\STATE $\alpha_i \gets \frac{|T|}{|T'_i|}\alpha$ for all $i \in [k]$
		\RETURN $\bigcup_{i \in [k]} \GMultifilterBD(T'_i, \alpha_i)$
	\end{algorithmic}
\end{algorithm}

\begin{algorithm}[ht!]
	\caption{$\GMultifilterBD(T, \alpha)$}\label{alg:gmultifilterbd}
	\begin{algorithmic}[1]
		\STATE \textbf{Input:} $T \subset \R^d$, $|T| = n$ satisfying Assumptions~\ref{assume:sexists} and~\ref{assume:gaussian} with parameter $\alpha \in (0,\thalf)$
		\STATE \textbf{Output:} With failure probability $\le \frac 1 d$: $L_{\text{out}}$ with $|L_{\text{out}}| = O(\tfrac 1 \alpha)$ such that some $\hmu \in L_{\text{out}}$ satisfies
		\[\norm{\hmu - \mus}_2 = O\Par{\frac{\log(d)\log\log^{1.5}(d)}{\sqrt \alpha}}.\]
		\STATE $L^{(0)} \gets \{T\}$, $L_{\text{out}} \gets \emptyset$
		\STATE For sufficiently large constants,
		\[R \gets \Theta\Par{\log(d)\log\log^{1.5}(d)},\; C \gets \Theta(\log^2 d),\; D \gets \Theta(\log^2 d)\]
		\FOR{$\ell \in [D]$}
		\STATE $L^{(\ell)} \gets \emptyset$
		\FOR{$T' \in L^{(\ell - 1)}$}
		\STATE Append all elements of $\GPartition(T', \alpha, \frac{1}{\log d}, C, R)$ to $L^{(\ell)}$ with size at least $\frac{\alpha n}{2}$
		\ENDFOR
		\ENDFOR
		\RETURN List of empirical means of all sets in $L^{(D)}$
	\end{algorithmic}
\end{algorithm}

We begin by stating the guarantees of $\NaiveCluster$, used in Line 3 of $\GMultifilter$.

\begin{lemma}[Lemma 12, \cite{DiakonikolasKKLT20}]\label{lem:naivecluster}
There is a randomized algorithm, $\NaiveCluster$, which takes as input $T \subset \R^d$ satisfying Assumption~\ref{assume:sexists} and partitions it into disjoint subsets $\{T'_i\}_{i \in [k]}$ such that with probability at least $1 - \tfrac{1}{d^2}$, all of $S$ is contained in the same subset, and every subset has diameter bounded by $O(d^{12})$. The runtime of $\NaiveCluster$ is $O(nd + n\log n)$.
\end{lemma}

We next demonstrate that if the operator norm of the (unnormalized) covariance matrix of a set of points $T'$ is bounded, and $T'$ has sufficient overlap with $S$, then its empirical mean is close to $\mus$.

\begin{lemma}\label{lem:meanclose}
For $T' \subset T$ with empirical mean $\hmu$, if $|T' \cap S| \ge \thalf |S|$ and $\tcov_{\tfrac 1 n \1}(T') \le R^2$,
\[\norm{\hmu - \mus}_2 = O\Par{(1 + R)\cdot \frac{1}{\sqrt{\alpha}}}.\]
\end{lemma}
\begin{proof}
Let $w$ place weight $\tfrac{1}{n}$ on coordinates in $T'$, and $0$ on all other coordinates. Clearly this $w$ satisfies the assumption of Lemma~\ref{lem:covbounddist}, since its $\ell_1$ norm is simply $\frac{|T'|}{|T|} \ge \thalf \alpha$. The conclusion follows by applying Lemma~\ref{lem:covbounddist}, where we use $\norm{w}_1 \cov_w(T) = \tcov_{\tfrac 1 n\1}(T')$, and the assumed bound.
\end{proof}

We now give a full analysis of $\GMultifilterBD$.

\begin{proposition}\label{prop:gmfboundeddiameter}
$\GMultifilterBD$ meets its output specifications with probability at least $1 - \tfrac 1 d$, within runtime
\[O\Par{nd\log^4(d) + n\log^5(d)\log\log(d)}.\]
\end{proposition}
\begin{proof}
Throughout, we denote $\beta \defeq \frac{1}{\log d}$. There are three main guarantees of the algorithm: that the list size is $O(\tfrac 1 \alpha)$, that some list element satisfies \eqref{eq:gaussianerror}, and that the runtime is as claimed.

We first bound the list size. We can view $\GMultifilterBD$ as producing a tree of subsets, of depth $D$. Each layer of the tree is composed by the sets in $L^{(\ell)}$ where $0 \le \ell \le D$, and $L^{(0)}$ is the root node. The children of each node are the results of calling $\GPartition$ on the associated subset. Moreover, by repeatedly using the guarantee \eqref{eq:gsizepot} inductively, the total cardinality of all sets at layer $\ell$ is bounded by $n^{1 + \beta} = O(n)$. Since we only return means from sets with size at least $\tfrac{\alpha n}{2}$ on layer $D$, there can only be $O(\tfrac 1 \alpha)$ such sets. 

Next, we bound error rate. Consider some leaf node, and its path to the root; call the sets associated with these vertices $T_0, T_1, \ldots T_D$, where $T_D$ is the leaf node and $T_0 = T$ is the original set. Define the potential function at each layer $0 \le \ell \le D$,
\[\Phi_{\ell} \defeq \Tr\Par{\mm_{\ell}^{2\log d}},\text{ where } \mm_{\ell} \defeq \tcov_{\frac 1 n \1}\Par{T_{\ell}}.\]
Note that every parent-child pair along this path satisfies the guarantee \eqref{eq:gfiltercond}. We thus conclude that for each $0 \le \ell < D$, we have the recurrence (analogously to \eqref{eq:potargument})
\begin{align*}
\Phi_{\ell + 1} &= \Tr\Par{\mm_{\ell + 1}^{2\log d}} \le \frac{1}{2R^2} \Tr\Par{\mm_{\ell + 1}^{2\log d + 1}} + d(2R^2)^{2 \log d} \\
&\le \frac{1}{2R^2} \Tr\Par{\mm_{\ell}^{2\log d} \mm_{\ell + 1}} + d(2R^2)^{2 \log d} \\
&\le \half \Tr\Par{\mm_{\ell}^{2\log d}} + d(2R^2)^{2 \log d} = \half \Phi_{\ell} + d(2R^2)^{2 \log d}.
\end{align*}
The first line used Fact~\ref{fact:alphasplit} with $\gamma = 2R^2$, the second used Fact~\ref{fact:eltapp}, and the third used the guarantee \eqref{eq:gfiltercond}. Thus, as long as at a layer $\ell$ we have
\[\Phi_{\ell} > 4d(2R^2)^{2 \log d},\]
we have $\Phi_{\ell + 1} \le \frac 3 4 \Phi_{\ell}$, and so the potential is decreasing by at least a constant factor. The potential $\Phi_0$ is bounded by $d^{O(\log d)}$, because we assumed the input set has polynomially bounded diameter, so within $D = \Omega(\log^2 d)$ layers, every node on layer $D$ must have $\Phi_D \le 4d(2R^2)^{2 \log d}$. This implies that the operator norm of $\tcov_{\frac 1 n \1}(T')$ for every node $T'$ on layer $D$ is $O(R^2)$.

We next show at least one node $T'$ on every layer has $|T' \cap S| \ge \thalf |S|$. By inductively using \eqref{eq:gkeepmostofs} with our chosen value of $C$, summing over the $O(\log^2 d)$ layers guarantees that we only remove at most $\thalf |S|$ points from the intersection throughout the root-to-leaf path, for some path. We can now apply Lemma~\ref{lem:meanclose} to guarantee \eqref{eq:gaussianerror}. To obtain the high-probability bound, note that the number of times we call $\GPartition$ is bounded by $O(\tfrac 1 \alpha \log^2 d)$, since at each layer we prune every node with less than $\tfrac{\alpha n}{2}$ points; there can only be $O(\tfrac 1 \alpha)$ surviving nodes per layer (since the total cardinalities of the layer is bounded by $n^{1 + \beta} = O(n)$), and taking a union bound over all calls to $\GPartition$ shows the failure probability is at most $\tfrac 1 d$.

Finally, we discuss runtime. We simply apply Corollary~\ref{cor:gpruntime} to each layer, which bounds the runtime of each layer by $O(nd\log(d) + n\log^3(d)\log\log(d))$, since the sets on that layer satisfy \eqref{eq:gsizepot} inductively. Summing over all layers yields the desired runtime guarantee.
\end{proof}

\begin{theorem}\label{thm:gaussian}
$\GMultifilter$ meets its output specifications with probability at least $1 - \tfrac 1 d$, within runtime
\[O\Par{nd\log^4(d) + n\log^5(d)\log\log(d)}.\]
\end{theorem}
\begin{proof}
We apply Proposition~\ref{prop:gmfboundeddiameter} to the relevant call of $\GMultifilterBD$. Note that all $\alpha_i \ge \alpha$, giving the error guarantee \eqref{eq:gaussianerror}, and
\[\sum_{i \in [k]} \frac{1}{\alpha_i} = \frac{1}{\alpha},\]
giving the list size guarantee. The runtime follows from $\sum_{i \in [k]} |T'_i| = |T|$.
\end{proof}

\section{Fast bounded covariance multifilter}\label{sec:boundedcov}

In this section, we give our algorithm for list-decodable mean estimation under only Assumption~\ref{assume:sexists}. 
As before, we can assume without loss of generality that $\alpha \in [1/d, 1 / \log^C d]$, for some constant $C > 0$.
We begin by giving our main subroutine, $\Partition$, in Section~\ref{ssec:partition}. The goal of $\Partition$ will be to produce child subsets $\{c_\ell\}_{\ell \in [k]}$ of a given input set $p$, which each satisfy the potential criterion in \eqref{eq:gfiltercond}, reproduced here:
\begin{equation}\label{eq:filtercond}\inprod{\my_p^2}{\mm_{c_\ell}} \le R^2 \Tr\Par{\my_p^2}.\end{equation}
Recall that in Section~\ref{sec:gaussian}, the way we produced children satisfying condition \eqref{eq:filtercond} was by ensuring that along logarithmically many random directions, each child $c_\ell$ lied entirely in short intervals. We will satisfy this guarantee in this section by more directly working with the definition of \eqref{eq:filtercond}, which requires each child to have small variance along the random directions, a looser condition.

To bound the variance of the child subsets, in Section~\ref{ssec:sorc} we develop an algorithm, $\SOrC$, which is patterned off our earlier $\GSOrC$. It either certifies that the input set is already ``close'' to having bounded variance in an input direction, or identifies a split point which produces two subsets which are closer to having this property, while maintaining at least one subset retains most points in $S$. In the first case (the ``cluster'' case), we develop a postprocessing procedure $\Fix$ in Section~\ref{ssec:fixing} which randomly filters points according to safe outlier scores (see Definition~\ref{def:safety}) to make the remaining cluster have truly bounded variance. In the second case (the ``split'' case), we develop a fast threshold checking procedure $\TBOC$ in Section~\ref{ssec:tboc} which identifies a valid split in polylogarithmic time, whenever one exists; here, we note the key difficulty is that we can no longer use a fixed radius for splits, because Gaussian concentration does not hold.

We discuss runtimes of all of these algorithms in Section~\ref{ssec:runtime}, and in particular give a runtime bound on $\Partition$. Finally, we use $\Partition$ to develop our full algorithm, $\Multifilter$, which we analyze in Section~\ref{ssec:fullboundedcov} through a potential argument similar to our analysis of $\GMultifilter$. A post-processing step used in $\Multifilter$ is analyzed in Section~\ref{ssec:postprocess}.

\subsection{Reducing \texorpdfstring{$\Partition$}{Partition} to \texorpdfstring{$\SOrC$}{SplitOrCluster}}\label{ssec:partition}

The goal of this section is to develop $\Partition$, the main subroutine of $\Multifilter$. $\Partition$ has very similar guarantees to the algorithm $\GPartition$ developed in Section~\ref{ssec:gpartition}. It takes as input a ``parent set'' $T_p \subseteq T$ and produces a number of ``children subsets'' $\{T_{c_\ell}\}_{\ell \in [k]}$ such that every child satisfies $\inprod{\my_p^2}{\mm_{c_\ell}} \le R^2 \Tr\Par{\my_p^2}$, where we follow the definitions \eqref{eq:gaussianmatdef}, reproduced here:
\begin{equation}\label{eq:matdef} 
\begin{gathered}
\mm_p \defeq \tcov_{\frac 1 n \1}\Par{T_p},\; \my_p \defeq \mm_p^{\log d},\\
\mm_{c_\ell} \defeq \tcov_{\frac 1 n \1}\Par{T_{c_\ell}},\; \my_{c_\ell} \defeq \mm_{c_\ell}^{\log d} \text{ for all } \ell \in [k],
\end{gathered}
\end{equation} 
This will allow us to conduct a potential analysis to bound the depth of the multifilter tree in Section~\ref{ssec:fullboundedcov}. Moreover, we require two additional guarantees of $\Partition$. 
\begin{enumerate}
	\item The first is the same as \eqref{eq:gsizepot}; namely, for some parameter $\beta \in (0, 1]$, we have
	\begin{equation}\label{eq:sizepot}\sum_{\ell \in [k]} \left|T_{c_\ell}\right|^{1 + \beta} \le \left|T_p\right|^{1 + \beta}. \end{equation}
	This will help us bound the total work done in each layer of the multifilter tree.
	\item The second is ensures at least one child preserves most points in $S$, assuming that the parent $T_p$ has this property. To this end, the tools of Section~\ref{ssec:priortools} will vastly simplify the language of this section. In particular, we will ensure that for $\gamma = O(\log(\frac{1}{\alpha}))$, every filter step in this entire section will be with respect to $\gamma$-safe weights in at least one branch. We then apply Lemma~\ref{lem:safeimpliessat} to conclude that some node at every level of the multifilter tree is $\gamma$-saturated.
\end{enumerate}

For the remainder of Section~\ref{sec:boundedcov}, we will define
\[\gamma \defeq 8\log\Par{\frac 1 \alpha}.\]
We demonstrate one important consequence of a set being $\gamma$-saturated.

\begin{lemma}\label{lem:gammasaturated}
Suppose for a set $T' \subset T$, the weights $w \defeq \tfrac 1 n \1_{T'} \in \Delta^n$ which place $\tfrac 1 n$ on coordinates in $T'$ and $0$ otherwise are $\gamma$-saturated (cf.\ Definition~\ref{def:saturation}). Then,
\[|T' \cap S| \ge \frac{\alpha n}{2}.\]
\end{lemma}
\begin{proof}
Recall that Definition~\ref{def:saturation} gives
\[\norm{w}_1 \ge \norm{w_S}_1 \ge \alpha\norm{w}_1^{\frac 1 \gamma} \implies \norm{w}_1^{1 - \frac 1 \gamma} \ge \alpha \implies \norm{w}_1 \ge \frac{3\alpha}{4}. \]
The first implication was by rearrangement, and the second used $\alpha^{\frac{1}{1 + \gamma^{-1}}} \ge \alpha^{1 + \frac 2 \gamma} \ge \frac{3\alpha}{4}$. Next,
\[\norm{w_S}_1 \ge \alpha\norm{w}_1^{\frac 1 \gamma} \ge \alpha\Par{\frac{3\alpha}{4}}^{\frac 1 \gamma} \ge \frac \alpha 2.\]
The conclusion follows since $\norm{w_S}_1$ counts the elements of $T' \cap S$, normalized by $\frac 1n$.
\end{proof}

Lemmas~\ref{lem:safeimpliessat} and~\ref{lem:gammasaturated} imply that as long as we can guarantee that at every filtering step, at least one child was produced with respect to $\gamma$-safe scores, that child retains half the elements of $S$. We are now ready to state the algorithm $\Partition$, which heavily relies on a subroutine $\PartitionOneD$.

\begin{algorithm}[ht!]
  \caption{$\Partition(T_p, \alpha, \delta, \beta, R)$}\label{alg:partition}
  \begin{algorithmic}[1]
    \STATE \textbf{Input:} $T_p \subset T$, $\alpha \in (0,\thalf)$, $\delta \in (0, 1)$, $\beta \in (0, 1]$, $R \in \R_{\ge 0}$ satisfying (for a sufficiently large constant) 
    \[R = \Omega\Par{\max\Par{\frac{1}{\beta} \cdot \sqrt{\gamma\log\Par{\frac{1}{\alpha\beta}}}, \sqrt{\gamma\log\Par{\frac{\log d}{\delta}}}}}.\]
    \STATE \textbf{Output:} With failure probability $\le \delta$: subsets $\{T_{c_\ell}\}_{\ell \in [k]}$ of $T_p$, satisfying \eqref{eq:sizepot}. Every child satisfies \eqref{eq:filtercond} (using notation \eqref{eq:matdef}). If $w \defeq \frac{1}{n}\1_{T_p}$ is $\gamma$-saturated, then $w' \defeq \frac{1}{n}\1_{T_{c_\ell}}$ is $\gamma$-saturated for at least one child $T_{c_\ell}$.
    \STATE Sample $\ndir = \Theta(\log \tfrac d \delta)$ vectors $\{u_j\}_{j \in [\ndir]} \in \R^d$ each with independent entries $\pm 1$. Following notation \eqref{eq:matdef}, let $v_j \gets \my_p u_j$ for all $j \in [\ndir]$.
    \STATE $\mathcal{S}_0 \gets T_p$
    \FOR{$j \in [\ndir]$}
    \STATE $\mathcal{S}_j \gets \emptyset$
    \FOR{$T' \in \mathcal{S}_{j - 1}$}
    \STATE $\mathcal{T} \gets \PartitionOneD(T', \alpha, v_j, \frac{\delta}{2\ndir}, \beta, R)$
    \STATE $\mathcal{S}_j \gets \mathcal{S}_j \cup \mathcal{T}$
    \ENDFOR
    \ENDFOR
    \RETURN $\mathcal{S}_{\ndir}$
  \end{algorithmic}
\end{algorithm}

As in the Gaussian case, we develop an algorithm $\PartitionOneD$ which in turn is based on a subroutine $\SOrC$, which we implement in Section~\ref{ssec:sorc}. The algorithm $\PartitionOneD$ takes an input direction $v$ and guarantees that along this direction, every child subset produced has small variance (scaled by the length of $v$). $\PartitionOneD$ is implemented by recursively calling $\SOrC$, which takes as input a set $T''$ and produces either one or two subsets, analogously to $\GSOrC$.

\begin{algorithm}[ht!]
  \caption{$\PartitionOneD(T', \alpha, v, \delta, \beta, R)$}\label{alg:partitiononed}
  \begin{algorithmic}[1]
    \STATE \textbf{Input:} $T' \subset T$, $\alpha \in (0,\thalf)$, $v \in \R^d$, $\delta \in (0, 1)$, $\beta \in (0, 1]$, $R \in \R_{\ge 0}$ satisfying (for a sufficiently large constant) \[R = \Omega\Par{\max\Par{\frac{1}{\beta} \cdot \sqrt{\gamma\log\Par{\frac{1}{\alpha\beta}}}, \sqrt{\gamma\log\Par{\frac{\log d}{\delta}}}}}.\]
    \STATE \textbf{Output:} Subsets $\{T''_\ell\}_{\ell \in [k]}$ of $T'$, such that
    \begin{equation}\label{eq:partitiononedsize}\sum_{\ell \in [k]} \left|T''_\ell\right|^{1 + \beta} \le |T'|^{1 + \beta}.\end{equation}
    Every child $T''_\ell$ for $\ell \in [k]$ has 
    \begin{equation}\label{eq:onedfiltercond}\inprod{\tcov_{\frac 1 n \1} (T''_\ell)}{vv^\top} \le \half R^2\norm{v}_2^2.\end{equation}
    If $w = \frac{1}{n}\1_{T'}$ is $\gamma$-saturated, then $w' = \frac{1}{n}\1_{T_\ell''}$ is $\gamma$-saturated for at least one child $T_\ell''$, with failure probability $\leq \delta$.
    \STATE $\mathcal{S}_{\text{in}} \gets T'$, $\mathcal{S}_{\text{out}} \gets \emptyset$
    \WHILE{$\mathcal{S}_{\text{in}} \neq \emptyset$}
    \STATE $T'' \gets$ the first element of $\mathcal{S}_{\text{in}}$
    \STATE $\mathcal{S}_{\text{in}} \gets \mathcal{S}_{\text{in}} \setminus \{T''\}$
    \IF{$\SOrC(T'', \alpha, v, \delta, \beta, R)$ returns one set $\Tout^{(0)}$}
    \STATE $\mathcal{S}_{\text{out}} \gets \mathcal{S}_{\text{out}} \cup \left\{\Tout^{(0)}\right\}$
    \ELSE
    \STATE $T_{\text{out}}^{(1)}, T_{\text{out}}^{(2)} \gets \SOrC(T'', \alpha, v, \delta, \beta, R)$
    \STATE $\mathcal{S}_{\text{in}} \gets \mathcal{S}_{\text{in}} \cup \left\{T_{\text{out}}^{(1)}, T_{\text{out}}^{(2)}\right\}$ %
    \ENDIF
    \ENDWHILE
    \RETURN $\mathcal{S}_{\text{out}}$
  \end{algorithmic}
\end{algorithm}

\begin{lemma}\label{lem:partition}
The output of $\Partition$ satisfies the guarantees given in Line 2 of Algorithm~\ref{alg:partition}, assuming correctness of $\PartitionOneD$.
\end{lemma}
\begin{proof}
We will follow the proof of Lemma~\ref{lem:gpartition}. First, to demonstrate that the subsets satisfy \eqref{eq:sizepot}, inducting on the guarantee \eqref{eq:partitiononedsize} of $\PartitionOneD$ suffices. Similarly, $\gamma$-saturation of some child follows from inducting on the corresponding guarantee of $\PartitionOneD$.

Finally, using the guarantee \eqref{eq:onedfiltercond} of $\PartitionOneD$, every $T_{c_\ell} \in \mathcal{S}_{\ndir}$ satisfies
\[\inprod{v_jv_j^\top}{\tcov_{\frac{1}{n}\1}(T_{c_\ell})} \leq \half R^2\norm{v_j}_2^2, \text{ for all } j \in [\ndir].\]
In other words, the variance is small along all directions $\{\my_p u_j = v_j\}_{j \in [\ndir]}$. We conclude
\begin{align*}
  \inprod{\my_p^2}{\mm_{c_l}} &\le \frac{1.4}{2n\left|T_{c_\ell}\right|\ndir} \sum_{i, i' \in T_{c_l}} \sum_{j \in [\ndir]} \inprod{\my_p u_j}{X_i - X_{i'}}^2 \\
  &= \frac{1.4}{\ndir} \sum_{j \in [\ndir]} \inprod{v_jv_j^\top}{\tcov_{\frac{1}{n}\1}(T_{c_\ell})} \\
                              &\le \frac{1.4}{2\ndir} \sum_{j \in [\ndir]} R^2 \norm{v_j}_2^2 \\
                              &= \frac{1.4}{2\ndir} \sum_{j \in [\ndir]} R^2 \norm{\my_p u_j}_2^2 \le R^2 \Tr\Par{\my_p^2},
\end{align*}
with probability at least $1 - \delta$; the first two lines and the last line follow the proof of Lemma~\ref{lem:gpartition}, and we used the variance guarantee of $\PartitionOneD$ in the third line. We remark that we make sure to take the number of directions $\ndir$ to depend logarithmically on $\delta$ (as opposed to just $d$), so we can apply the guarantees of \cite{Achlioptas03} with probability $1 - \frac \delta 2$ on the first and last lines.
\end{proof}

Next, we state a correctness guarantee for $\PartitionOneD$, assuming correctness of its main subroutine, $\SOrC$. The guarantees of $\SOrC$ are summarized in Line 2 of Algorithm~\ref{alg:sorc}. The salient features are that it takes a set $\Tin$ and either produces one set satisfying \eqref{eq:onedfiltercond} deterministically, or two sets which each are strict subsets of $\Tin$ (and hence remove at least one point) deterministically, and \eqref{eq:partitiononedsize} is always maintained. In the two set case, if $\frac 1 n \1_{\Tin}$ is $\gamma$-saturated then so is at least one output deterministically; in the one set case, the output is saturated with probability at least $1 - \delta$. We will use only these features to analyze $\PartitionOneD$ in Lemma~\ref{lem:partitiononed}.

\begin{lemma}\label{lem:partitiononed}
The output of $\PartitionOneD$ satisfies the guarantees given in Line 2 of Algorithm~\ref{alg:partitiononed}, assuming correctness of $\SOrC$.
\end{lemma}
\begin{proof}
Each run of Lines 4-13 results in either one set (which we call a ``cluster step'') or two sets (which we call a ``split step''). We can view this process as a tree, where a leaf node corresponds to the result of a cluster step, and every node on the leaf-to-node path corresponds to a split step. Every time a split step occurs, it increases $|\mathcal{S}_{\textup{in}}| + |\mathcal{S}_{\textup{out}}|$ by one, so there are at most $(n')^{1 + \beta}$ calls to $\SOrC$ where $|T'| = n'$, and thus the algorithm terminates in finite time.

Next, if $T'$ is not saturated, then there is no failure probability, since the conditions \eqref{eq:partitiononedsize} and \eqref{eq:onedfiltercond} deterministically succeed. Otherwise, from the root of this partition tree, consider the root-to-leaf path which at each node corresponding to a split step takes any child which corresponds to a saturated child (one always exists because split steps deterministically succeed). The only failure probability comes from the success of the leaf-parent to leaf cluster step, which fails with probability $\delta$. We can ignore all other bad events, because we only need to ensure one child is saturated. 
\end{proof}

\subsection{Reducing \texorpdfstring{$\SOrC$}{SplitOrCluster} to \texorpdfstring{$\TBOC$}{TailBoundOrCluster} and $\Fix$}\label{ssec:sorc}

In this section, we state and analyze $\SOrC$, the main subroutine of $\PartitionOneD$.

\begin{algorithm}[ht!]
	\caption{$\SOrC(\Tin, \alpha, v, \delta, \beta, R)$}\label{alg:sorc}
	\begin{algorithmic}[1]
		\STATE \textbf{Input:} $\Tin \subseteq T$, $\alpha \in (0,\thalf)$, $v \in \R^d$, $\delta \in (0,1)$, $\beta \in (0, 1]$, $R \in \R_{\ge 0}$ satisfying (for a sufficiently large constant) 
		\[R = \Omega\Par{\max\Par{\frac{1}{\beta} \cdot \sqrt{\gamma\log\Par{\frac{1}{\alpha\beta}}}, \sqrt{\gamma\log\Par{\frac{\log d}{\delta}}}}}.\]
		\STATE \textbf{Output:} Either one subset $\Tout^{(0)} \subset \Tin$, or two subsets $\Tout^{(1)}, \Tout^{(2)} \subset \Tin$. In the one subset case, $\Tout^{(0)}$ has $\inprod{\tcov_{\frac 1 n \1}\Par{\Tout^{(0)}}}{vv^\top} \leq \thalf R^2 \norm{v}_2^2$. In the two subsets case, they take the form, for some threshold value $\tau \in \R$ and $r \in \R_{\ge 0}$
		\begin{equation}\label{eq:twosetsdef}
                  \begin{aligned}
                    \Tout^{(1)} \defeq \{X_i \mid \inprod{v}{X_i} \le \tau + r\norm{v}_2\}, \;
                    \Tout^{(2)} \defeq \{X_i \mid \inprod{v}{X_i} \ge \tau - r\norm{v}_2\},                    
                  \end{aligned}
                \end{equation}
		and satisfy
		\begin{equation}\label{eq:twosetssizebound}
                  \begin{aligned}
                    \left|\Tout^{(1)}\right|^{1 + \beta} + \left|\Tout^{(2)}\right|^{1 + \beta} < \left|\Tin\right|^{1 + \beta}, \\
                    \min\Par{ 1-\frac{\abs{\Tout^{(1)}}}{\abs{\Tin}}, 1-\frac{\abs{\Tout^{(2)}}}{\abs{\Tin}} } \geq \frac{2\gamma}{r^2}.                    
                  \end{aligned}
                \end{equation}
                In either case if $\frac 1 n \1_{\Tin}$ is $\gamma$-saturated then $\frac 1 n \1_{\Tout}$ is $\gamma$-saturated for at least one child $\Tout$, with failure probability $\leq\delta$ only in the case one set is returned (deterministically otherwise).
		\STATE $Y_i \gets \inprod{v}{X_i}$ for all $i \in \Tin$
		\STATE $\taumed \gets \med\Par{\{Y_i \mid i \in \Tin\}}$, where $\med$ returns the median
		\STATE $I \gets \left[\taumed - c, \taumed +c \right]$ is the smallest interval containing the $1 - \tfrac \alpha 4$ quantiles of $\{Y_i \mid i \in \Tin\}$ for $c \in \R_{\ge 0}$ and $2I \gets \left[\taumed -2c,\taumed +2c\right]$
                \IF{$\inprod{\tcov_{\frac{1}{n}\1}\Par{\Tmid}}{vv^\top} \leq \tfrac{1}{8} R^2\norm{v}_2^2$ where $\Tmid \defeq \{X_i \in \Tin \mid Y_i \in 2I\}$}
                \RETURN $\Fix(\Tin,\alpha,v,\delta, R)$
                \ELSE
		\STATE Run both $\TBOC\Par{\Tin,v,\beta, \taumed \pm \frac{1}{2^k}\cdot \sqrt{\frac{2048`}{\beta^2 \alpha}}\norm{v}_2}$ for integers $k$ with
		\[0 \le k \le \log_2\Par{\frac{2048}{\beta^2 \alpha}}\]
		until one returns $\Par{\Tout^{(1)}, \Tout^{(2)}}$ satisfying \eqref{eq:twosetsdef}, \eqref{eq:twosetssizebound}
		\RETURN $\Tout^{(1)}$, $\Tout^{(2)}$ 
		\ENDIF
	\end{algorithmic}
\end{algorithm}

$\SOrC$ uses two subroutines, $\Fix$ and $\TBOC$, which are respectively used to handle the one child and two children cases. Roughly speaking, $\Fix$ takes as input a set $\Tin$ which ``almost'' has bounded variance in the direction $v$, and slightly filters extreme outliers in a way so that the result has truly bounded variance. On the other hand, $\TBOC$ is used at a candidate threshold $\tau$ to either check that it induces sets $\Tout^{(1)}, \Tout^{(2)}$ satisfying \eqref{eq:twosetssizebound} or satisfies a certain tail bound. By stitching together tail bounds at a small number of quantiles, $\SOrC$ guarantees that at least one of these quantiles was a valid threshold, else we would attain a contradiction as Line 6 of $\SOrC$ would have passed. We state guarantees of $\TBOC$ and $\Fix$ as Lemmas~\ref{lem:tboc} and~\ref{lem:fix}, and prove them in Sections~\ref{ssec:tboc} and~\ref{ssec:fixing} respectively. We then use Lemmas~\ref{lem:tboc} and~\ref{lem:fix} to prove Lemma~\ref{lem:sorc}, which demonstrates correctness of $\SOrC$.

\begin{restatable}{lemma}{restatetboc}\label{lem:tboc}
There is an algorithm, $\TBOC$ (Algorithm~\ref{alg:tboc}), which takes as input $\Tin \subseteq T$, $v \in \R^d$, $\beta \in (0, 1]$, $R \in \R_{\ge 0}$, and $\tau_0 \in \R$, and returns in one of two cases we call the ``split'' case and the ``tail bound'' case. In the split case, it returns $(\tau, r)$ such that the induced sets \eqref{eq:twosetsdef} satisfy \eqref{eq:twosetssizebound}, and if $\frac 1 n \1_{\Tin}$ is $\gamma$-saturated then one of the induced sets $\Tout$ has $\frac 1 n \1_{\Tout}$ is $\gamma$-saturated. Otherwise, for all $t \in \R$ define the upper and lower tail probabilities
\begin{equation}\label{eq:tailbounddef}\rho^+(t) \defeq \Pr_{i \sim_{\textup{unif}} \Tin}\Brack{Y_i \ge t},\; \rho^-(t) \defeq \Pr_{i \sim_{\textup{unif}} \Tin} \Brack{Y_i \le t}.\end{equation}
Then, in the tail bound case $\TBOC$ certifies 
\begin{equation}\label{eq:tailbound}\rho^+(\tau_0) \le \frac{128\gamma}{\beta^2\left|\tau_0 - \taumed\right|^2}\norm{v}_2^2\text{ if }\tau_0 \ge \taumed,\text{ and } \rho^-(\tau_0) \le \frac{128\gamma}{\beta^2|\tau_0 - \taumed|^2}\norm{v}_2^2 \text{ if } \tau_0 \le \taumed.\end{equation}
\end{restatable}

\begin{restatable}{lemma}{restatefix}\label{lem:fix}
There is an algorithm, $\Fix$ (Algorithm~\ref{alg:fix}), which takes as input $\Tin \subseteq T$,  $v \in \R^d$, and $R \in \R_{\ge 0}$ and produces $\Tout$ with the following guarantee with probability at least $1 - \delta$. Define $\Tmid$ as in Line 5 of Algorithm~\ref{alg:sorc}. Then if $\frac1n \1_{\Tin}$ is $\gamma$-saturated, so is $\frac 1n \1_{\Tout}$, and if $\inprod{\tcov_{\frac 1 n \1}(\Tmid)}{vv^\top} \le \frac18 R^2 \norm{v}_2^2$, then $\inprod{\tcov_{\frac 1 n \1}(\Tout)}{vv^\top} \le \frac12 R^2 \norm{v}_2^2$. 
\end{restatable}

Finally, we are ready to prove Lemma~\ref{lem:sorc}, the main export of this section.

\begin{lemma}\label{lem:sorc}
The output of $\SOrC$ satisfies the guarantees given in Line 2 of Algorithm~\ref{alg:sorc}.
\end{lemma}
\begin{proof}
If the check in Line 6 passes (the one subset case), correctness of $\SOrC$ follows immediately from the guarantees of $\Fix$ in Lemma~\ref{lem:fix}. We now focus on the two subset case.

Assume throughout this proof that the $\{Y_i\}_{i \in \Tmid}$ are ordered by distance to $\taumed$, so $|Y_1 - \taumed| \le \ldots \le |Y_{m} - \taumed|$, where $m \defeq |\Tmid|$; we order the remaining elements $i \in \Tin \setminus \Tmid$ arbitrarily. We also define $\taumed$, $\Tmid$, $c$, $I$, and $2I$ as in Lines 4-6 of $\SOrC$. If Line 9 outputs a split for any $k$, then by Lemma~\ref{lem:tboc}, \eqref{eq:twosetsdef}, \eqref{eq:twosetssizebound} are satisfied, and the saturation condition is met for one of the children. It remains to show that some value of $k$ will result in the split case of $\TBOC$. We show this by contradiction; if all $k$ resulted in $\TBOC$ returning with a tail bound guarantee, we prove $\Tmid$ would have passed Line 6. Assume for the remainder of the proof that $\TBOC$ failed to find a split for all $k$.

First, we bound the length of the interval $2I$. Because $\SOrC$ failed to return a split for $k = 1$, by combining the corresponding tail bounds \eqref{eq:tailbound},
\[\Pr_{i \sim_{\textup{unif}} \Tin} \Brack{\left|Y_i - \taumed \right| > \sqrt{\frac{512}{\beta^2 \alpha}} \norm{v}_2} \le \frac \alpha 4.\]
Thus, we conclude
\[2I \subset \Brack{\taumed - C\norm{v}_2, \taumed + C\norm{v}_2},\text{ for } C \defeq \sqrt{\frac{2048}{\beta^2 \alpha}}.\]
We proceed to bound $\inprod{\tcov_{\frac 1 n \1}(\Tmid)}{vv^\top}$ to obtain our desired contradiction. Observe that
\begin{equation}\label{eq:tcovzbound}
\begin{aligned}
\tcov_{\frac 1 n \1}(\Tmid) &\preceq \cov_{\frac 1 n \1}(\Tmid) \\
&= \frac{1}{|\Tmid|} \sum_{i \in \Tmid} \Par{X_i - \mu_{\frac 1 n \1}\Par{\Tmid}}\Par{X_i - \mu_{\frac 1 n \1}\Par{\Tmid}}^\top \\
&\preceq \frac{1}{|\Tmid|} \sum_{i \in \Tmid} \Par{X_i - \bar{X}}\Par{X_i - \bar{X}}^\top \text{ where } \inprod{v}{\bar{X}} = \taumed \\
\implies \inprod{\tcov_{\frac 1 n \1}(\Tmid)}{vv^\top} &\le \frac{1}{|\Tmid|} \sum_{i \in \Tmid} \Par{Y_i - \taumed}^2.
\end{aligned}
\end{equation}
Here, we used the definitions of $\tcov$, $\cov$ in the first two lines, and Fact~\ref{fact:meanshift} in the third. The last follows by definition of $\{Y_i\}_{i \in \Tin}$ and $\taumed$. Define the random variable $Z$ to be the realization of $\abs{Y_i - \taumed}$ for $i$ a uniform draw from $\Tmid$, let $Z_i \defeq \abs{Y_i - \taumed}$, and let $G(t) = \Pr[Z \geq t]$ be the inverse cumulative density function of $Z$. Notice that directly expanding implies that (where we let $m \defeq |\Tmid|$, $Z_0 \defeq 0$, and recall we argued $Z_m \le C\norm{v}_2$ earlier)
  \[
    \E[Z^2] = \sum_{i \in [m]} (Z_i^2 - Z_{i-1}^2)G(Y_i) = \int_0^{Z_{m}} 2tG(t) dt \le \int_0^{C\norm{v}_2} 2tG(t)dt.
  \]
  Define now $K(t)$ for each $\norm{v}_2 \leq t \leq C\norm{v}_2$ to be the smallest $k$ such that $t^{(k)} \defeq \frac{C}{2^k}\norm{v}_2 \leq t$, so $K(C\norm{v}_2) = 0$, $K(t) = 1$ for $t \in \left[\frac{C}{2}\norm{v}_2, C\norm{v}_2\right)$, and so on. By construction, for all relevant $t$,
  \[
    t^{(K(t))} \leq t < 2t^{(K(t))}.
  \]
  Since $G$ is decreasing in its argument, we can write
  \[
    \int_{0}^{C\norm{v}_2} 2tG(t) dt \leq \int_0^{\norm{v}_2} 2t dt + \int_{\norm{v}_2}^{C\norm{v}_2} 2t G\left(t^{(K(t))} \right) dt \le \norm{v}_2^2 + \int_{\norm{v}_2}^{C\norm{v}_2} 2t G\left(t^{(K(t))} \right) dt.
  \]
  Now, for all $0 \le k \le \log_2(\frac{80}{\beta^2\alpha})$, we recall that we assumed both calls to $\TBOC$ with thresholds $\taumed \pm \frac{C}{2^k}\norm{v}_2$ failed to produce a split, and hence certify a tail bound \eqref{eq:tailbound}. Thus,
  \[G\Par{t^{(k)}} \le \frac{512\gamma}{\beta^2\Par{t^{(k)}}^2} \norm{v}_2^2 \le \frac{2048\gamma}{\beta^2t^2} \norm{v}_2^2,\text{ for any } t \text{ with } K(t) = k.\]
  Here, the first inequality used both tail bounds in \eqref{eq:tailbound} and accounted for the fact that the quantizations $\rho^+$, $\rho^-$ are defined over $\Tin$, and $G$ is defined over $\Tmid$ with $|\Tmid| \ge \thalf |\Tin|$; the second inequality used that for any such $t$, $t < 2t^{(k)}$. Putting all these pieces together,
  \[\E[Z^2] \le \norm{v}_2^2 + \int_{\norm{v}_2}^{C\norm{v}_2} \frac{4096\gamma}{\beta^2 t} \norm{v}_2^2 dt \le\norm{v}_2^2 + \frac{4096\gamma}{\beta^2} \log(C) \norm{v}_2^2 = O\Par{\frac{\gamma}{\beta^2} \cdot \log\Par{\frac{1}{\alpha\beta}}}\norm{v}_2^2.\]
  Finally, recall \eqref{eq:tcovzbound} shows that $\inprod{\tcov_{\frac 1 n \1}(\Tmid)}{vv^\top} \le \E[Z^2]$. Thus, the set $\Tmid$ should have passed the check in Line 6 under the assumed lower bound on $R$, yielding the desired contradiction.
\end{proof}

\subsection{Implementation of \texorpdfstring{$\TBOC$}{TailBoundOrCluster}}\label{ssec:tboc}

In this section, we prove Lemma~\ref{lem:tboc} by providing $\TBOC$ and giving its analysis.

\begin{algorithm}[ht!]
  \caption{$\TBOC(\Tin, v, \beta, \tau_0)$}\label{alg:tboc}
  \begin{algorithmic}[1]
    \STATE \textbf{Input:} $\Tin \subseteq T$, $v \in \R^d$, $\beta \in (0, 1]$, $\tau_0 \in \R$
    \STATE \textbf{Output:} Either outputs $(\tau, r)$ such that $\Tout^{(1)} \defeq \{X_i \mid \inprod{v}{X_i} \le \tau + r\norm{v}_2\},
    \Tout^{(2)} \defeq \{X_i \mid \inprod{v}{X_i} \ge \tau - r\norm{v}_2\}$ satisfy 
    \begin{align*}\left|\Tout^{(1)}\right|^{1 + \beta} + \left|\Tout^{(2)}\right|^{1 + \beta} < \left|\Tin\right|^{1 + \beta},\;    \min\Par{ 1-\frac{\abs{\Tout^{(1)}}}{\abs{\Tin}}, 1-\frac{\abs{\Tout^{(2)}}}{\abs{\Tin}} } \geq \frac{2\gamma}{r^2},                    \end{align*}
    or returns ``Tail bound'' guaranteeing that for $\taumed \defeq \med\Par{\Brace{\inprod{v}{X_i} \mid i \in \Tin}}$ (following \eqref{eq:tailbound})
    \[\rho^+(t) \le \frac{128\gamma}{\beta^2\left|\tau_0 - \taumed\right|^2}\norm{v}_2^2\text{ if }\tau_0 \ge \taumed,\text{ and } \rho^-(t) \le \frac{128\gamma}{\beta^2|\tau_0 - \taumed|^2}\norm{v}_2^2 \text{ if } \tau_0 \le \taumed.\]
    \STATE $Y_i \gets \inprod{v}{X_i}$ for all $i \in \Tin$, $\taumed \gets \med(\{Y_i \mid i \in \Tin\})$
    \STATE $j \gets 0$
    \IF{$\tau_0 > \max_{i \in \Tin} Y_i$ or $\tau_0 < \min_{i \in \Tin} Y_i$}
    \RETURN ``Tail bound''
    \ENDIF
    \IF{$\tau_0 \geq \taumed$}
    \WHILE{$\tau_j \geq \taumed$}
    \STATE $g_j \gets \rho^+(\tau_j)$ and $r_j \gets \sqrt{\frac{2\gamma}{g_j}}$
    \IF{$\Tout^{(1)}$, $\Tout^{(2)}$ induced by $(\tau_j - r_j\norm{v}_2, r_j)$ satisfy \eqref{eq:twosetssizebound}}
    \RETURN $(\tau_j-r_j\norm{v}_2,r_j)$
    \ELSE
    \STATE $\tau_{j+1} \gets \tau_j - 2r_j$
    \ENDIF
    \STATE $j \gets j+1$
    \ENDWHILE
    \ELSE
    \WHILE{$\tau_j \leq \taumed$}
    \STATE $\ell_j \gets \rho^-(\tau_j)$ and $r_j \gets \sqrt{\frac{2\gamma}{\ell_j}}$
    \IF{$\Tout^{(1)}$, $\Tout^{(2)}$ induced by $(\tau_j + r_j\norm{v}_2, r_j)$ satisfy \eqref{eq:twosetssizebound}}
    \RETURN $(\tau_j+r_j\norm{v}_2,r_j)$
    \ELSE
    \STATE $\tau_{j+1} \gets \tau_j + 2r_j$
    \ENDIF
    \STATE $j \gets j+1$
    \ENDWHILE    
    \ENDIF
    \RETURN ``Tail bound''
  \end{algorithmic}
\end{algorithm}

\restatetboc*
\begin{proof}
This proof proceeds in two parts which we show separately. First, we demonstrate that if $\tau_0 \ge \taumed$ and none of the runs of Lines 9-17 in Algorithm~\ref{alg:tboc} return with a valid split, then indeed we can certify the tail bound \eqref{eq:tailbound} holds (and a similar guarantee holds for $\tau_0 < \taumed$). Second, we show whenever a split is returned and $\Tin$ is $\gamma$-saturated, then one of the output sets will be as well.

\textit{Correctness of tail bound.} We consider the case $\tau_0 \ge \taumed$ here, as the other case follows symmetrically. Let $K + 1$ be the first index such that $\tau_{K + 1} < \taumed$, so the algorithm checks all pairs $(\tau_j - r_j\norm{v}_2, r_j)$ for $0 \le j \le K$. Suppose that all such induced splits fail to satisfy \eqref{eq:twosetssizebound}. For all $0 \le j \le K$, let $T^{(1)}_j$, $T^{(2)}_j$ be the induced sets by the pair $(\tau_j - r_j\norm{v}_2, r_j)$. Then by construction
\[\frac{\abs{T^{(1)}_j}}{\abs{\Tin}} = 1 - g_j,\; \frac{\abs{T^{(2)}_j}}{\abs{\Tin}} = g_{j + 1}.\]
For all $0 \le j \le K - 1$, recall $(\tau_j - r_j\norm{v}_2, r_j)$ did not pass the check \eqref{eq:twosetssizebound}, but the second expression reads $g_j \ge \frac{2\gamma}{r_j^2}$ (since $g_j \le \thalf \le 1 - g_{j + 1}$), which is true by construction. Thus the first check did not pass and we conclude
\begin{equation}\label{eq:gjbound}g_{j + 1}^{1 + \beta} + \Par{1 - g_j}^{1 + \beta} > 1 \implies g_{j + 1}^{1 + \beta} > g_j.\end{equation}
By applying this inequality inductively with $j = K - 1$, and recalling $g_{K} \le \half$, we have
\begin{equation}\label{eq:kbound}\Par{\half}^{(1 + \beta)^{K}} \ge g_{K}^{(1 + \beta)^{K}} > g_0 \implies 2^{(1 + \beta)^{K}} < \frac{1}{g_0} \implies K = O\Par{\frac{1}{\beta}\log\log\Par{\frac{1}{g_0}}}.\end{equation}
Next, since $\tau_{K + 1} = \tau_0 - 2\sum_{0 \le j \le K} r_j \norm{v}_2 < \taumed$, we have
\[\tau_0 - \taumed < 2\sum_{j = 0}^{K} r_j \norm{v}_2 = \sqrt{8\gamma} \norm{v}_2\sum_{j = 0}^{K} \sqrt{\frac{1}{g_j}} < \sqrt{8\gamma} \norm{v}_2\sum_{j = 0}^{K} A^{\frac{1}{(1 + \beta)^j}},\text{ for } A \defeq \sqrt{\frac{1}{g_0}}.\]
where we used our earlier guarantee \eqref{eq:gjbound} inductively. By Lemma~\ref{lem:supergeometric}, we have the desired tail bound:
\[\tau_0 - \taumed < \sqrt{8\gamma} \norm{v}_2 \cdot \frac{4A}{\beta} \le \frac{\sqrt{128}}{\beta}\sqrt{\frac{\gamma}{g_0}} \implies \rho^+(\tau_0) = g_0 < \frac{128\gamma}{\beta^2\abs{\tau_0-\taumed}^2} \norm{v}_2^2.\]

\textit{Correctness of split.} Suppose that we find $(\tau, r)$ such that for
\[\Tout^{(1)} \defeq \{X_i \mid \inprod{v}{X_i} \le \tau + r\norm{v}_2\},\;
\Tout^{(2)} \defeq \{X_i \mid \inprod{v}{X_i} \ge \tau - r\norm{v}_2\},\]
we have
\begin{equation}\label{eq:tinscores}\min\Par{1 - \frac{\left|\Tout^{(1)}\right|}{|\Tin|},1 - \frac{\left|\Tout^{(2)}\right|}{|\Tin|}} \ge \frac{2\gamma}{r^2}.\end{equation}
We show that the downweighting $\frac 1 n \1_{\Tin} \to \frac 1 n \1_{\Tout^{(i)}}$ is a weight removal with respect to $\gamma$-safe scores for one of $i = 1, 2$. Let $\tau^* \defeq \inprod{v}{\mus}$ where $\mus$ is the ``true mean vector'' in Assumption~\ref{assume:sexists}. Clearly, either $\tau^* \ge \tau$ or $\tau^* \le \tau$; suppose without loss of generality $\tau^* \ge \tau$ as the other case follows symmetrically. Define the scores $\{s_i\}_{i \in \Tin}$ to be $1$ if $i \in \Tin \setminus \Tout^{(2)}$ and $0$ otherwise; then the downweighting $\frac 1 n \1_{\Tin} \to \frac 1 n \1_{\Tout^{(2)}}$ is of the form in Lemma~\ref{lem:safeimpliessat}, with respect to these scores. Note that
\[\frac{1}{\gamma} \sum_{i \in \Tin} \frac{1}{|\Tin|} s_i = \frac{1}{\gamma} \Par{1 - \frac{\abs{\Tout^{(2)}}}{\abs{\Tin}}} \ge \frac{2}{r^2}\]
by the assumption \eqref{eq:tinscores}. To apply Lemma~\ref{lem:safeimpliessat}, it remains to show that
\begin{equation}\label{eq:stinbound}\sum_{i \in S \cap \Tin} \frac{1}{\abs{S \cap \Tin}} s_i \le \frac{2}{r^2}.\end{equation}
However, we can extend the definition of the scores $\{s_i\}_{i \in S}$ to include points in $S \setminus \Tin$, so that $s_i$ is the indicator function of $\inprod{v}{X_i} < \tau - r\norm{v}_2$ for all $i \in S$, which is consistent with our definitions $\{s_i\}_{i \in \Tin}$. Then by Chebyshev's inequality and Assumption~\ref{assume:sexists}, using $\inprod{v}{\mus} \ge \tau$,
\begin{equation}\label{eq:sbound}\sum_{i \in S} \frac{1}{|S|} s_i \le \Pr_{i \sim_{\textup{unif}} S}\Brack{\inprod{v}{X_i - \mus}^2 > r^2\norm{v}_2^2} \le \frac{1}{r^2}.  \end{equation}
By Lemma~\ref{lem:gammasaturated}, if $\Tin$ is $\gamma$-saturated then $\abs{S \cap \Tin} \ge \half|S|$; combining with \eqref{eq:sbound} yields \eqref{eq:stinbound}.
\end{proof}

\begin{lemma}\label{lem:supergeometric}
Let $A > \sqrt{2}$, $\beta \in (0, 1]$, and let $K$ be such that $A^{\frac{1}{(1+\beta)^K}} > \sqrt{2}$. Then, we have
\[\sum_{j = 0}^K A^{\frac{1}{(1 + \beta)^j}} \le \frac{4A}{\beta}.\]
\end{lemma}
\begin{proof}
Define $f(x) = A^{\frac{1}{(1 + \beta)^x}}$ for any $0 \le x \le K$, and note this is a decreasing function in $x$. Thus, since by direct computation the antiderivative of $A^{B^x}$ is $\frac{1}{\log B}\text{Ei}(B^x\log(A))$ where $\text{Ei}$ is the exponential integral,
\begin{align*}\sum_{j = 0}^K A^{\frac{1}{(1 + \beta)^j}} &\le A + \int_0^K f(x)dx = A + \frac{1}{\log(1+\beta)} \Par{\text{Ei}\Par{\log A} - \text{Ei}\Par{\frac{\log A}{(1+\beta)^K}}} \\
&\le A + \frac{2}{\beta}\Par{\text{Ei}\Par{\log A} + 1}.
\end{align*}
In the last line, we used $\log(1+\beta) \ge \frac \beta 2$ for $\beta \in (0, 1]$, $\frac{\log A}{(1 + \beta)^K} > \log(\sqrt 2)$ by assumption, and $\text{Ei}$ is increasing with $\text{Ei}(\log(\sqrt 2)) > -1$. The conclusion follows from $\text{Ei}(\log A) + 1 \le \frac{3}{2}A$ for $A > \sqrt{2}$.
\end{proof}

\subsection{Fixing a cluster via fast filtering}\label{ssec:fixing}

In this section, we prove Lemma~\ref{lem:fix} by providing $\Fix$ and giving its analysis. Before stating the algorithm, we provide a helper result which analyzes the effect of a ``randomized dropout scheme'' with respect to safe scores, and shows with high probability it is still safe.

\begin{algorithm}[ht!]
	\caption{$\RandDrop(T', \drd, s)$}\label{alg:randdrop}
	\begin{algorithmic}[1]
		\STATE \textbf{Input:} $T' \subseteq T$, $\drd \in (0, 1)$, $4\gamma$-safe scores $\{s_i\}_{i \in T'}$ with respect to $w \defeq \frac 1 n \1_{T'}$ such that $s_{\max} \defeq \max_{i \in T'} s_i \le 24|T' \cap S|$, and
		 \begin{equation}\label{eq:tscoresbound}\sum_{i \in T'} \frac{1}{|T'|} s_i \ge 288\gamma \log\Par{\frac{2}{\drd}}.\end{equation}
		\STATE \textbf{Output:} With failure probability $\le \drd$, outputs $T'' \subseteq T'$ such that if $w$ is $\gamma$-saturated, then $\frac1 n \1_{T''}$ is $\gamma$-saturated.
		\STATE $T'' \gets \emptyset$
		\FOR{$i \in T'$}
		\STATE $T'' \gets T'' \cup \{X_i\}$ with probability $1 - \frac{s_i}{s_{\max}}$
		\ENDFOR
		\RETURN $T''$
	\end{algorithmic}
\end{algorithm}

In other words, $\RandDrop$ removes points from $T'$ with probability proportional to their score.

\begin{lemma}\label{lem:randdrop}
The output of $\RandDrop$ satisfies the guarantees given in Line 2 of Algorithm~\ref{alg:randdrop}.
\end{lemma}
\begin{proof}
For all $i \in T'$, let $Z_i$ be the random variable defined as
\[Z_i = \begin{cases}
1 & \text{with probability } \frac{s_i}{s_{\max}} \\
0 & \text{with probability } 1 - \frac{s_i}{s_{\max}}
\end{cases}.\]
Note that the number of points removed from $T'$ and $T' \cap S$ are respectively $\sum_{i \in T'} Z_i$ and $\sum_{i \in T' \cap S} Z_i$. We now obtain high-probability bounds on both of these totals.

First, we lower bound $\sum_{i \in T'} Z_i$. Observe that $\E\Brack{\sum_{i \in T'}Z_i} = \sum_{i \in T'} \frac{s_i}{s_{\max}} $, and each $Z_i$ is Bernoulli. Thus we can apply a Chernoff bound to obtain
\[\Pr\Brack{\sum_{i \in T'} Z_i < \half \sum_{i \in T'} \frac{s_i}{s_{\max}}} \le  \exp\Par{-\frac{1}{8}\sum_{i \in T'} \frac{s_i}{s_{\max}}} \le \frac{\drd}{2},\]
where we used $s_{\max} \le 24|T'|$ and the assumed lower bound \eqref{eq:tscoresbound} to conclude
\[\sum_{i \in T'} \frac{s_i}{s_{\max}} \ge \sum_{i \in T'} \frac{s_i}{24|T'|} \ge 8\log\Par{\frac{2}{\drd}}.\]
Next, we upper bound $\sum_{i \in T' \cap S} Z_i$. We claim with failure probability at most $\frac{\drd}{2}$,
\begin{equation}\label{eq:claimtprimes}\sum_{i \in T' \cap S} Z_i \le \frac{1}{2\gamma} \frac{\abs{T' \cap S}}{|T'|} \sum_{i \in T'} \frac{s_i}{s_{\max}}. \end{equation}
Define $\mu \defeq \sum_{i \in T' \cap S} \frac{s_i}{s_{\max}}$ to be the expectation of the left hand side of \eqref{eq:claimtprimes}, and set
\[\Delta \defeq \frac{1}{\mu}\Par{\frac{1}{2\gamma} \frac{\abs{T' \cap S}}{|T'|} \sum_{i \in T'} \frac{s_i}{s_{\max}}} - 1\]
so that $(1 + \Delta)\mu$ is the right hand side of \eqref{eq:claimtprimes}. Recall that we assumed that $\{s_i\}_{i \in T}$ were $4\gamma$-safe; rearranging this definition (cf.\ Definition~\ref{def:safety}) yields
\[\mu \le \frac{1}{4\gamma} \cdot \frac{\abs{T' \cap S}}{|T'|} \sum_{i \in T'} \frac{s_i}{s_{\max}}\implies \Delta\mu = \frac{1}{2\gamma} \frac{\abs{T' \cap S}}{|T'|} \sum_{i \in T'} \frac{s_i}{s_{\max}} - \mu \ge \frac{1}{4\gamma} \cdot \frac{\abs{T' \cap S}}{|T'|} \sum_{i \in T'} \frac{s_i}{s_{\max}}.\]
However, since $\frac{\abs{T' \cap S}}{s_{\max}} \ge \frac{1}{24}$ by assumption, we use \eqref{eq:tscoresbound} and the above equation to conclude
\[\Delta\mu \ge 3\log\Par{\frac{2}{\drd}}.\]
Finally, a Chernoff bound shows the failure probability of \eqref{eq:claimtprimes} is at most $\exp\Par{-\frac{\Delta\mu}{3}} \le \frac{\drd}{2}$, as desired. Thus with probability at least $1 - \drd$,
\[\sum_{i \in T' \cap S}\frac{1}{|T' \cap S|} Z_i \le \frac{1}{\gamma} \sum_{i \in T'}\frac{1}{|T'|}  Z_i.\]
Now observe that the $\{Z_i\}_{i \in T'}$ meet the definition of $\gamma$-safe scores (Definition~\ref{def:safety}). Thus, Lemma~\ref{lem:safeimpliessat} applies with weights $\frac 1 n \1_{T'}$ and $\frac 1 n \1_{T''}$ and we obtain the conclusion.
\end{proof}

We are now ready to state the algorithm $\Fix$ and prove its guarantees in Lemma~\ref{lem:fix}.

\begin{algorithm}[ht!]
  \caption{$\Fix(\Tin, \alpha, v, \delta, R)$}\label{alg:fix}
  \begin{algorithmic}[1]
    \STATE \textbf{Input:} $\Tin \subseteq T$, $\alpha \in (0,\thalf)$, $v \in \R^d$, $\delta \in (0,1)$, $R \in \R_{\ge 0}$ satisfying (for a sufficiently large constant)
    \[R = \Omega\Par{\sqrt{\gamma\log\Par{\frac{\log d}{\delta}}}},\]
    such that for $\Tmid$ defined in Algorithm~\ref{alg:sorc}, $\inprod{\tcov_{\frac 1 n \1}(\Tmid)}{vv^\top} \le \frac 1 8 R^2\norm{v}_2^2$
    \STATE \textbf{Output:} Outputs $\Tout \subset \Tin$ with
    \[\inprod{\tcov_{\frac 1 n \1}(\Tout)}{vv^\top} \le \half R^2\norm{v}_2^2.\]
    If $\frac 1 n \1_{\Tin}$ is $\gamma$-saturated, so is $\frac 1 n \1_{\Tout}$, with failure probability $\le \delta$.
    \IF{$\inprod{\tcov_{\frac 1 n \1}(\Tin)}{vv^\top} \le \half R^2\norm{v}_2^2$}
    \RETURN $\Tin$
    \ENDIF
    \STATE $Y_i \gets \inprod{v}{X_i}$ for all $i \in \Tin$, $\taumed \gets \med\Par{\Brace{Y_i \mid i \in \Tin}}$
    \STATE $I \gets [\taumed - c, \taumed + c]$ is the smallest interval containing the $1 - \frac{\alpha}{4}$ quantiles of $\{Y_i \mid i \in \Tin\}$ for $c \in \R_{\ge 0}$ and $2I \gets [\taumed - 2c, \taumed + 2c]$
    \STATE Define scores $\{s_i\}_{i \in \Tin}$ by
    \[s_i \gets \begin{cases}
    0 & Y_i \in I \\
    (Y_i - (\taumed - c))^2 & Y_i \le \taumed - c \\
    (Y_i - (\taumed + c))^2 & Y_i \ge \taumed + c
    \end{cases}\]
    \STATE $\drd \gets \frac{\delta}{\Omega\Par{\log d \cdot \log \frac d \delta}}$ for a sufficiently large constant
    \STATE $\Tout \gets \Tin \setminus \{X_i \mid s_i \ge 12\norm{v}_2^2 |S|\}$
    \WHILE{$\inprod{\tcov_{\frac 1 n \1}(\Tout)}{vv^\top} > \half R^2\norm{v}_2^2$}
    \STATE $\Tout \gets \RandDrop(\Tout, \drd, \frac{s}{\norm{v}_2^2})$
    \ENDWHILE
    \RETURN $\Tout$
  \end{algorithmic}
\end{algorithm}

\restatefix*
\begin{proof}
This proof proceeds in three parts. First, we show that whenever the average score is small:
\[\frac{1}{|\Tout|}\sum_{i \in \Tout} s_i \le 288\gamma \log\Par{\frac{2}{\drd}}\norm{v}_2^2,\]
then the check in Line 11 will fail and the algorithm will terminate. Next, we show calls to $\RandDrop$ meet its input criteria so its conclusion holds inductively (correctness of Line 10 is also handled here). Finally, we show that $\Fix$ fails with probability at most $\delta$. Assume throughout that $\frac 1 n \1_{\Tin}$ is $\gamma$-saturated; else there is nothing to prove. We also use the following notation throughout:
\begin{equation}\label{eq:muvarvdef}\Var_v\Par{T'} \defeq \frac{1}{\abs{T'}} \sum_{i \in T'} \Par{Y_i - \mu_v\Par{T'}}^2,\text{ where } \mu_v\Par{T'} \defeq \inprod{v}{\mu_{\frac 1 n \1}\Par{T'}},\text{ for all } T' \subseteq T.\end{equation}

\textit{Small average score implies termination.} We show that whenever the average score is small: $\frac{1}{|\Tout|}\sum_{i \in \Tout} s_i \le 288\gamma \log\Par{\frac{2}{\drd}}\norm{v}_2^2$, we terminate since this implies
\[\inprod{\tcov_{\frac 1 n \1}(\Tout)}{vv^\top} \le \half R^2\norm{v}_2^2.\]
 To show this, we first prove
\begin{equation}\label{eq:boundnotin2i}\begin{aligned}s_i > \frac{1}{16}\Par{Y_i - \mu_{2I}}^2 \text{ for all } i \in \Tout, Y_i \not\in 2I,\\
\text{ where } \mu_{2I} \defeq \frac{1}{\abs{\Tin \cap \Brace{i \mid Y_i \in 2I}}}\sum_{i \in \Tout \cap \Brace{i \mid Y_i \in 2I}}Y_i.\end{aligned}\end{equation}
In other words, $\mu_{2I}$ is the mean of points in $2I$. To see this, if $Y_i = \taumed - 2c - \Delta$ for $\Delta > 0$,
\[s_i = \Par{c + \Delta}^2,\; \Par{Y_i - \mu_{2I}}^2 \le \Par{4c + \Delta}^2 < 16 s_i.\]
The case when $Y_i = \taumed + 2c + \Delta$ is handled similarly, which covers all $Y_i \not\in 2I$. Then, following notation \eqref{eq:muvarvdef},
\begin{align*}
\Var_v(\Tout) &= \sum_{i \in \Tout} \frac{1}{\abs{\Tout}} \Par{Y_i - \mu_v(\Tout)}^2 \le \sum_{i \in \Tout} \frac{1}{\abs{\Tout}} \Par{Y_i - \mu_{2I}}^2 \\
&= \sum_{i \in \Tout \cap \Brace{i \mid Y_i \in 2I}} \frac{1}{|\Tout|} \Par{Y_i - \mu_{2I}}^2 + \sum_{i \in \Tout \cap \Brace{i \mid Y_i \notin 2I}} \frac{1}{|\Tout|} \Par{Y_i - \mu_{2I}}^2 \\
&\le \sum_{i \in \Tout \cap \Brace{i \mid Y_i \in 2I}} \frac{1}{|\Tout|} \Par{Y_i - \mu_{2I}}^2 + 16\sum_{i \in \Tout} \frac{1}{|\Tout|} s_i \\
&\le \sum_{i \in \Tout \cap \Brace{i \mid Y_i \in 2I}} \frac{1}{|\Tout|} \Par{Y_i - \mu_{2I}}^2 + O\Par{\gamma \log\Par{\frac{1}{\drd}}}\norm{v}_2^2.
\end{align*}
Here, the first line used Fact~\ref{fact:meanshift}, the third used \eqref{eq:boundnotin2i} and that all scores are nonnegative, and the last used our assumption on the average score in $\Tout$. Thus,
\begin{align*}
\inprod{\tcov_{\frac 1 n \1}\Par{\Tout}}{vv^\top} &= \sum_{i \in \Tout} \frac{1}{n} \Par{Y_i - \mu_v(\Tout)}^2 \\
&\le \sum_{i \in \Tout \cap \{i \mid Y_i \in 2I\}} \frac{1}{n} (Y_i - \mu_{2I})^2 + O\Par{\gamma \log\Par{\frac{1}{\drd}}}\norm{v}_2^2 \\
&= \inprod{\tcov_{\frac 1 n \1}\Par{\Tout \cap \{i \mid Y_i \in 2I\}}}{vv^\top} + O\Par{\gamma \log\Par{\frac{1}{\drd}}}\norm{v}_2^2 \\
&\le \frac{1}{8} R^2\norm{v}_2^2 + O\Par{\gamma \log\Par{\frac{1}{\drd}}}\norm{v}_2^2 \le \half R^2 \norm{v}_2^2.
\end{align*}
In the second line we used $n \ge |\Tout|$ to handle the second term, the third line used the definition of $\tcov$, and the fourth line used the assumed bound on $\inprod{\tcov_{\frac 1 n \1}(\Tmid)}{vv^\top}$, and
\[\tcov_{\frac 1 n \1}\Par{\Tout \cap \{i \mid Y_i \in 2I\}} \preceq \tcov_{\frac 1 n \1}\Par{\Tmid}\]
since $\Tout \cap \{i \mid Y_i \in 2I\} \subseteq \Tmid$ and Fact~\ref{fact:meanshift} implies that dropping terms from the covariance formula and shifting to the mean only decreases Loewner order. The last line used the lower bound on $R$.

\textit{Correctness of calls to $\RandDrop$.} We first bound the average score in $\Tin \cap S$ at the beginning of the algorithm. Let $\Var_v(\Tin \cap S)$ denote the variance of $\Tin \cap S$ in the direction $v$ following \eqref{eq:muvarvdef}. We claim that the mean of $\Tin \cap S$ lies close to $I$: in particular,
\begin{equation}\label{eq:meanclose}\mu_{\frac 1 n \1}(\Tin \cap S) \in \Brack{\taumed - c - \sqrt{2\Var_v\Par{\Tin \cap S}}, \taumed + c + \sqrt{2\Var_v\Par{\Tin \cap S}}}.\end{equation}
If this were not the case, we would have a contradiction:
\begin{align*}
\Var_v\Par{\Tin \cap S} &\ge \frac{1}{\abs{\Tin \cap S}} \sum_{i \in \Tin \cap S \mid Y_i \in I} \Par{Y_i - \mu_v\Par{\Tin \cap S}}^2 \\
&> \frac{\abs{\Tin \cap S \cap \{i \mid Y_i \in I\}}}{\abs{\Tin \cap S}} \Par{2\Var_v\Par{\Tin \cap S}} \ge \Var_v\Par{\Tin \cap S}.
\end{align*}
The second inequality used that every summand is at least $2\Var_v\Par{\Tin \cap S}$ if \eqref{eq:meanclose} does not hold, and the last used that $I$ contains a $1 - \frac{\alpha}{4}$ proportion of the points in $\Tin$, and by Lemma~\ref{lem:gammasaturated} $\Tin \cap S$ contains at least $\frac{\alpha}{2}$ of the points in $\Tin$. Now using \eqref{eq:meanclose} and the definition of the scores,
\begin{equation}\label{eq:sscoresinit}
\begin{aligned}
\sum_{i \in \Tin \cap S} \frac{1}{\abs{\Tin \cap S}} s_i &\le \sum_{i \in \Tin \cap S} \frac{1}{\abs{\Tin \cap S}} \Par{\abs{Y_i - \mu_v\Par{\Tin \cap S}} + \sqrt{2\Var_v\Par{\Tin \cap S}}}^2 \\
&\le 6\Var_v\Par{\Tin \cap S} \le 12\norm{v}_2^2.
\end{aligned}
\end{equation}
In the last line, we used that $\Tin \cap S$ contains at least half the points in $S$ by Lemma~\ref{lem:gammasaturated}, so Assumption~\ref{assume:sexists} applies with a normalizing factor at most twice as large. This shows that Line 10 of Algorithm~\ref{alg:fix} preserves saturation, since it can only remove points in $\Tin \setminus S$ (if any point in $\Tin \cup S$ had a score larger than $12|S|\norm{v}_2^2$, it would violate \eqref{eq:sscoresinit}). This also shows that if at any point in running Algorithm~\ref{alg:fix} we have a $\gamma$-saturated subset $\Tout \subset \Tin$, then
\begin{equation}\label{eq:sscorestout}\sum_{i \in \Tout \cap S} \frac{1}{\abs{\Tout \cap S}} s_i \le 24\norm{v}_2^2.\end{equation}
This is because compared to \eqref{eq:sscoresinit}, we can at most double the normalizing factor by Lemma~\ref{lem:gammasaturated}, and all scores are nonnegative. By combining with the first part of this proof, whenever Line 11 passes,
\begin{equation}\label{eq:avgscorelarge}\frac{1}{|\Tout|} \sum_{i \in \Tout} s_i > 288\gamma\log\Par{\frac{2}{\drd}}\norm{v}_2^2,\end{equation}
and hence the scores are $4\gamma$-safe as required by $\RandDrop$. The second requirement of $\RandDrop$ is that $s_{\max} \le 24|\Tout \cap S|\norm{v}_2^2$, which is taken care of by Line 10 as $|S| \le 2|\Tout \cap S|$ by Lemma~\ref{lem:gammasaturated}. Finally, Line 11 implies \eqref{eq:avgscorelarge} by the first part of this proof, which is the third condition of $\RandDrop$.

\textit{Bounding failure probability.} We bound the failure probability in two steps. First, we show with probability at least $1 - \frac \delta 2$, there are at most (for a suitable constant)
\[N \defeq O\Par{\log d \cdot \log\frac{d}{\delta}}\]
calls to $\RandDrop$. Then, we union bound to show that all these calls to $\RandDrop$ pass with probability at least $1 - \frac \delta 2$. Combining gives the overall failure probability to $\Fix$.

To see the bound on $N$, observe that after Line 10, the largest score is at most $O(\norm{v}_2^2 d)$, and the algorithm ends when the largest score is at most a constant (since then \eqref{eq:avgscorelarge} clearly fails, at which point we terminate on Line 11 by the first part of this proof). Thus, the largest score can only halve at most $O(\log d)$ times. However, observing the implementation of $\RandDrop$, any point with score at least half the largest is dropped with probability at least $\thalf$, and hence after $O(\log \frac{d}{\delta})$ rounds, the largest score will halve with probability at least $1 - \frac{\delta}{\Omega(\log d)}$. Union bounding over all the phases of halving the max score implies after $N$ loops the algorithm terminates with probability $1 - \frac{\delta}{2}$.

Since there are at most $N$ calls to $\RandDrop$, it suffices to set $\drd = \frac{\delta}{2N}$ to check that all calls to $\RandDrop$ pass with probability $1 - \frac{\delta}{2}$. If all calls pass, we have the desired conclusion.
\end{proof}

\subsection{Runtime analysis}\label{ssec:runtime}

We now give a runtime bound for $\PartitionOneD$, and use it to obtain a similar bound on $\Partition$.

\begin{lemma}\label{lem:partitiononedruntime}
Let $n' \defeq |T'|$, where $T'$ is the input to $\PartitionOneD$. Then $\PartitionOneD$ can be implemented to run in time
\[O\Par{n' d + (n')^{1 + \beta}\Par{\frac{1}{\beta}\log\log d \cdot \log\Par{\frac{1}{\alpha\beta}} + \log d \log \frac d \delta}}.\]
\end{lemma}
\begin{proof}
We begin by computing all the points $Y_i \defeq \inprod{v}{X_i}$ for $i \in T'$ and sorting them, and store all quantiles (i.e.\ the number of points less than any given $Y_i$), which takes time $O(n' d + n' \log n')$.

Next, we bound the cost of running $\Fix$ on an input $\Tin$ of size $\nin$. Given access to quantile information, and since $\Fix$ is only ever called on a set which is formed after applying some number of splits to the original dataset $T'$, it is straightforward to implement Lines 3-10 in time $O(\nin)$. Moreover, each loop in Lines 11-13 costs $O(\nin)$ time, and by the proof of Lemma~\ref{lem:fix}, there are at most $O(\log d \log \frac{d}{\delta})$ loops. Thus overall the runtime of $\Fix$ is
\[O\Par{\nin \log d \log \frac{d}{\delta}}.\]
We now consider the cost of running $\TBOC$ with a given threshold $\tau_0$. If $\tau_0$ does not lie in the interval of $\{Y_i\}_{i \in \Tin}$, then the runtime is $O(1)$. Otherwise, consider the case $\tau_0 \ge \taumed$ (note $\taumed$ can be computed in constant time given quantile information). Since at least one point is larger than $\tau_0$, $g_0 \ge \frac{1}{n}$, and hence \eqref{eq:kbound} shows the number of threshold checks is bounded by $O(\frac{1}{\beta} \log\log d)$. Each threshold check takes constant time (we just need to compute the cardinalities of the induced $\Tout^{(1)}, \Tout^{(2)}$) and computing the next $g_j$ and $r_j$ takes constant time given quantile information, so the cost of $\TBOC$ is
\[O\Par{\frac{1}{\beta}\log\log d}.\]
Correspondingly, the cost of each run of Lines 8-11 of $\SOrC$ is bounded by
\[O\Par{\frac{1}{\beta}\log\log d \cdot \log\Par{\frac{1}{\alpha\beta}}}.\]
Now, consider the structure of $\PartitionOneD$. Lemma~\ref{lem:partitiononed} shows that there are at most $(n')^{1 + \beta}$ split steps total, so the total cost of all split steps (which run Lines 8-11 of $\SOrC$) is
\[O\Par{\frac{(n')^{1 + \beta}}{\beta}\log\log d \cdot \log\Par{\frac{1}{\alpha\beta}}}.\]
Finally, consider all nodes in the $\PartitionOneD$ which are parents of leaves. The sums of cardinalities of all such nodes is bounded by $(n')^{1 + \beta}$, so the cost of running $\Fix$ on all these nodes is
\[O\Par{(n')^{1 + \beta} \log d \log \frac{d}{\delta}}.\]
\end{proof}

As an immediate corollary, we obtain a runtime bound on $\Partition$.

\begin{corollary}\label{cor:pruntime}
Let $n_p \defeq |T_p|$ for some $T_p \subseteq T$. $\Partition$ called on input $T_p$ with parameter $C$ can be implemented to run in time%
\[O\Par{n_p^{1 + \beta} d \log d \log \frac d \delta + n_p^{1 + \beta}\Par{\frac{1}{\beta}\log\log d \cdot \log\Par{\frac{1}{\alpha\beta}}\log\frac d \delta + \log d \log^2 \frac d \delta}}.\]
\end{corollary}
\begin{proof}
The proof is identical to Corollary~\ref{cor:gpruntime}, where we use Lemma~\ref{lem:partitiononedruntime} to bound the cost over all elements of each $\mathcal{S}_j$, and there are $\ndir = \Theta(\log \frac d \delta)$ calls to $\PartitionOneD$. 
\end{proof}

\subsection{Full bounded covariance algorithm}\label{ssec:fullboundedcov}

Finally, we give our full algorithm for list-decodable mean estimation under Assumption~\ref{assume:sexists}. As in Section~\ref{ssec:fullgaussian}, we will reduce to the bounded diameter case via the algorithm $\NaiveCluster$ (cf.\ Lemma~\ref{lem:naivecluster}); we reproduce its guarantees for arbitrary failure probabilities as $\NaiveClusterPlus$.

\begin{lemma}[Lemma 12, \cite{DiakonikolasKKLT20}]\label{lem:naiveclusterplus}
There is a randomized algorithm, $\NaiveClusterPlus(T, \delta)$, which takes as input $T \subset \R^d$ satisfying Assumption~\ref{assume:sexists} and partitions it into disjoint subsets $\{T'_i\}_{i \in [k]}$ such that with probability at least $1 - \delta$, all of $S$ is contained in the same subset, and every subset has diameter bounded by $O(\frac{d^8}{\delta^2})$. The runtime of $\NaiveClusterPlus$ is $O(nd + n\log n)$.
\end{lemma}

We also require a post-processing procedure to reduce the list size, which we call $\IteratePostProcess$. We state its guarantees in Lemma~\ref{lem:postprocess}, and defer the description and analysis to Section~\ref{ssec:postprocess}.

\begin{restatable}{lemma}{restatepostprocess}\label{lem:postprocess}
There is an algorithm, $\IteratePostProcess$ (Algorithm~\ref{alg:postprocess}), which takes as input $T$ satisfying Assumption~\ref{assume:sexists} and a list $L \subset \R^d$ of length $m \le n$ such that
\[\min_{\hmu \in L} \norm{\hmu - \mus}_2 \le \Delta,\; \Delta = \Omega\Par{\frac{1}{\sqrt \alpha}}\]
and returns with probability at least $1 - \delta$ a subset $L' \subset L$ of size $O(\frac 1 \alpha)$ such that $\min_{\hmu \in L'} \norm{\hmu - \mus}_2 = O(\Delta)$, within runtime
\[O\Par{\Par{(m + n)d + \alpha m^2 n}\log \frac{d}{\delta}}.\]
\end{restatable}

\begin{algorithm}[ht!]
	\caption{$\Multifilter(T, \alpha, \delta, \beta)$}\label{alg:multifilter}
	\begin{algorithmic}[1]
		\STATE \textbf{Input:} $T \subset \R^d$, $|T| = n$ satisfying Assumption~\ref{assume:sexists} with parameter $\alpha \in (0,\thalf)$, $\delta \in (0, 1)$, $\beta \in (0, 1]$ %
		\STATE \textbf{Output:} With failure probability $\le \delta$: $L$ with $|L| = O(\tfrac{1}{\alpha})$ such that some $\hmu \in L$ satisfies
		\begin{equation}\label{eq:error}\norm{\hmu - \mus}_2 = O\Par{\sqrt{\frac{\log\Par{\frac 1 \alpha}}{\alpha}} \cdot \max\Par{\frac{1}{\beta}\sqrt{\log\Par{\frac{1}{\alpha\beta}}}, \sqrt{\log\log d}}}.\end{equation}
		\STATE $\delta_{\text{outer}} \gets \half$
		\STATE $N_{\text{runs}}\gets \lceil 2 \log \frac 2 \delta\rceil$
		\STATE $L \gets \emptyset$
		\FOR{$j \in [N_{\text{runs}}]$}
		\STATE $\{T'_i\}_{i \in [k]} \gets \NaiveClusterPlus(T, \frac {\delta_{\text{outer}}} 3)$
		\STATE $\alpha_i \gets \frac{|T|}{|T'_i|}\alpha$ for all $i \in [k]$
		\STATE $L \gets L \cup \IteratePostProcess\Par{T, \bigcup_{i \in [k]} \MultifilterBD(T'_i, \alpha_i, \frac {\delta_{\text{outer}}} 3, \beta), \frac {\delta_{\text{outer}}} 3}$
		\ENDFOR
		\RETURN $\IteratePostProcess\Par{T, L, \frac \delta 2}$
	\end{algorithmic}
\end{algorithm}

\begin{algorithm}[ht!]
	\caption{$\MultifilterBD(T, \alpha, \delta, \beta)$}\label{alg:multifilterbd}
	\begin{algorithmic}[1]
		\STATE \textbf{Input:} $T \subset \R^d$, $|T| = n$ satisfying Assumption~\ref{assume:sexists} with parameter $\alpha \in (0,\thalf)$, $\delta \in (0, 1)$, $\beta \in (0, 1]$
		\STATE \textbf{Output:} With failure probability $\le \delta$: $L_{\text{out}}$ with $|L_{\text{out}}| = O(\tfrac{n^\beta}{\alpha})$ such that some $\hmu \in L_{\text{out}}$ satisfies
		\[\norm{\hmu - \mus}_2 = O\Par{\sqrt{\frac{\log\Par{\frac 1 \alpha}}{\alpha}} \cdot \max\Par{\frac{1}{\beta}\sqrt{\log\Par{\frac{1}{\alpha\beta}}}, \sqrt{\log\Par{\frac{\log d}{\delta}}}}}.\]
		\STATE $L^{(0)} \gets \{T\}$, $L_{\text{out}} \gets \emptyset$
		\STATE For sufficiently large constants,
		\[R \gets \Theta\Par{\max\Par{\frac 1 \beta \cdot \sqrt{\log\Par{\frac 1 \alpha}\log\Par{\frac{1}{\alpha\beta}}}, \sqrt{\log\Par{\frac 1 \alpha}\log\Par{\frac d \delta}}}},\; D \gets \Theta\Par{\log d \log \frac d \delta}\]
		\FOR{$\ell \in [D]$}
		\STATE $L^{(\ell)} \gets \emptyset$
		\FOR{$T' \in L^{(\ell - 1)}$}
		\STATE Append all elements of $\Partition(T', \alpha, \frac{\delta}{n^{1 + \beta} D}, \beta, R)$ to $L^{(\ell)}$
		\ENDFOR
		\ENDFOR
		\RETURN List of empirical means of all sets in $L^{(D)}$ with size at least $\tfrac{\alpha n}{2}$
	\end{algorithmic}
\end{algorithm}

\begin{proposition}\label{prop:multifilterbd}
$\MultifilterBD$ meets its output specifications with probability at least $1 - \delta$, within runtime
\[O\Par{n^{1 + \beta} d \log^2 d \log^2 \frac d \delta + n^{1 + \beta}\Par{\frac{1}{\beta}\log\log d \cdot \log\Par{\frac{1}{\alpha\beta}}\log d\log^2\frac d \delta + \log^2 d \log^3 \frac d \delta}}.\]
\end{proposition}
\begin{proof}
The proof of the error rate is identical to that in Proposition~\ref{prop:gmfboundeddiameter}, where the initial potential $\Phi_0$ is bounded by $(\frac d \delta)^{O(\log d)}$ via Lemma~\ref{lem:naiveclusterplus}, which implies the operator norm of $\tcov_{\frac 1 n \1}(T')$ for every node $T'$ on layer $D$ is $O(R^2)$, and inductively at least one such node has $|T' \cap S| \ge \half |S|$ by virtue of being $\gamma$-saturated and applying Lemma~\ref{lem:gammasaturated}. The failure probability follows since $\Partition$ is called at most $n^{1 + \beta} D$ times, as there are at most $n^{1 + \beta}$ elements of each $L^{(\ell)}$. Finally, the list size follows since Lemma~\ref{lem:gammasaturated} implies every leaf node contains $\frac{\alpha n}{2}$, but the total size across leaves is at most $n^{1 + \beta}$.

Finally, to obtain the runtime bound we can sum the guarantee of Corollary~\ref{cor:pruntime} across each of the $D$ layers, and use the potential to bound the sum of all $n_p^{1 + \beta}$ across the layer.
\end{proof}

We are now ready to state our main claim on list-decodable mean estimation. For simplicity, we state the result for $\beta \ge \frac 1 {\log d}$, as otherwise there are no runtime or statistical gains asymptotically.

\begin{theorem}\label{thm:multifilter}
For $\frac 1 {\log d} \le \beta \le 1$, and $\log^{\Omega(1)}(d) \le  \alpha^{-1} \le d$, $\Multifilter$ returns a list of size $O(\frac 1 \alpha)$ such that
\[\min_{\hmu \in L} \norm{\hmu - \mus}_2 = O\Par{\frac 1 \beta \cdot \frac{\log\Par{\frac 1 \alpha}}{\alpha}},\]
with probability at least $1 - \delta$, within runtime
\[O\Par{n^{1 + 2\beta} d \log^4 d \log \frac 1 \delta + nd \log^2 \frac 1 \delta \log \frac d \delta}.\]
\end{theorem}
\begin{proof}
We first analyze Lines 6-10 of $\Multifilter$. We claim that each of the $N_{\text{runs}}$ times these lines run, there is a $\ge \half$ probability that some $\hmu$ will be added to $L$ satisfying \eqref{eq:error}, within runtime
\[O\Par{n^{1 + 2\beta} d \log^4 d + n^{1 + \beta}\Par{\frac{1}{\beta}\log\log d \cdot \log\Par{\frac{1}{\alpha\beta}}\log^3 d + \log^5 d}} = O\Par{n^{1 + 2\beta} d \log^4 d}.\]
To see this, we apply Proposition~\ref{prop:multifilterbd} to the relevant call of $\MultifilterBD$. The correctness follows identically to the proof of Theorem~\ref{thm:gaussian}, except that the size of the list of candidate means $\bigcup_{i \in [k]} \MultifilterBD(T'_i, \alpha_i, \frac {\delta_{\textup{outer}}} 3, \beta)$ is $m = O(\frac{n^{\beta}}{\alpha})$. By Lemma~\ref{lem:postprocess}, after applying $\IteratePostProcess$ the error rate is not affected by more than a constant, and the list size is $O(\frac 1 \alpha)$. The runtime of this last step is dominated by $O(\alpha m^2 n \log d) = O(n^{1 + 2\beta} d \log d)$.

Next, this implies that after all runs of Lines 6-10 have finished running (with independent internal randomness), there is a $\ge 1 - \frac \delta 2$ probability that $L$ contains an element $\hmu$ satisfying \eqref{eq:error}. At this point, the size of the list is $m = O(\frac{\log \delta^{-1}}{\alpha})$, so Line 11 takes time $O(nd \log^2 \frac 1 \delta \log \frac d \delta)$ by Lemma~\ref{lem:postprocess}. 
\end{proof}
This theorem, combined with the previously discussed fact that we can assume that $\alpha \in [1/d, 1 / \log^{\Omega (1)} d]$, gives our desired conclusion.

\subsection{Cleaning up the list}\label{ssec:postprocess}

In this section, we provide the subroutine $\IteratePostProcess$ used in $\Multifilter$, and prove Lemma~\ref{lem:postprocess}, which shows correctness of this subroutine. At a high level, $\IteratePostProcess$ first finds a greedy cover of the input list $L$ at distance $O(\Delta)$. Then, while the greedy cover has size at least $4k$ for $k \defeq \lceil \frac 1 \alpha \rceil$, it iteratively prunes away $2k$ out of $4k$ hypotheses by testing that there are enough datapoints closest to retained hypotheses; otherwise, it returns the greedy cover.

\begin{algorithm}[ht!]
	\caption{$\IteratePostProcess(T, \alpha, L, \delta, \Delta)$}\label{alg:postprocess}
	\begin{algorithmic}[1]
		\STATE \textbf{Input:} $T \subset \R^d$, $|T| = n$ satisfying Assumption~\ref{assume:sexists} with parameter $\alpha \in (0,\thalf)$, $\delta \in (0, 1)$, $L$ with $|L| = m \le n$ such that 
		\[\min_{\hmu \in L} \norm{\hmu - \mus}_2 \le \Delta,\; \Delta = \Omega\Par{\frac{1}{\sqrt \alpha}}.\]
		\STATE \textbf{Output:} With failure probability $\le \delta$: $L' \subset L$ with $|L'| = O(\tfrac{1}{\alpha})$ such that 
		\[\min_{\hmu \in L'} \norm{\hmu - \mus}_2 = O(\Delta).\]
		\STATE $\mg \in \R^{d \times c} \gets $ entrywise $\pm \frac{1}{\sqrt c}$ uniformly at random, for $c = \Theta(\log\frac{d}{\delta})$ (Johnson-Lindenstrauss matrix \cite{Achlioptas03})
		\STATE $k \gets \lceil \frac 1 \alpha \rceil$
		\STATE $L' \gets $ maximal subset of $L$ such that $\forall \hmu \neq \hmu' \in L'$, $\norm{\mg^\top(\hmu - \hmu')}_2 \ge 5\Delta$
		\WHILE{$|L'| \ge 4k$}
		\STATE $L_{\text{head}} \gets$ first $4k$ elements of $L'$
		\STATE $L_{\text{prune}} \gets$ elements of $L_{\text{head}}$ which are nearest neighbors of $<\frac{\alpha n}{2}$ elements of $T$, where $\hmu \in L_{\text{head}}$ is the nearest neighbor of $X_i \in T$ if $\norm{\mg^\top(\hmu - X_i)}_2$ is minimal amongst $L_{\text{head}}$ 
		\STATE $L \gets L \setminus L_{\text{prune}}$
		\STATE $L' \gets $ maximal subset of $L$ such that $\forall \hmu \neq \hmu' \in L'$, $\norm{\mg^\top(\hmu - \hmu')}_2 \ge 5\Delta$
		\ENDWHILE
		\RETURN $L'$
	\end{algorithmic}
\end{algorithm}

\restatepostprocess*
\begin{proof}
We first prove correctness, and then prove the runtime bound.

\textit{Correctness.} By the Johnson-Lindenstrauss lemma as analyzed in \cite{Achlioptas03}, with probability at least $1 - \delta$ every pair of points in $L \cup T \cup \{\mus\}$ has their distance preserved to a $1.1$ multiplicative factor under multiplication by $\mg^\top$. Condition on this event for the remainder of the proof.

Let $\bmu$ be the element of the input $L$ which is guaranteed to be within distance $\Delta$ of $\mus$. We will first show that $\bmu$ is never removed from $L$ by the loop in Lines 6-11. If $\bmu$ is not a part of $L_{\text{head}}$ in a given loop, clearly this is true, so suppose $\bmu \in L_{\text{head}}$, and let $\hmu$ be some other element in $L_{\text{head}}$ with $\norm{\mg^\top (\bmu - \hmu)}_2 \ge 5\Delta$; by definition of $L_{\text{head}}$ as a subset of $L'$, all such $\hmu$ satisfy this. Our goal will be to show that at least half of the points in $S$ have nearest neighbor $\bmu$; to do so, it suffices to show that $\norm{\mg^\top(X_i - \hmu))}_2 \ge \norm{\mg^\top(X_i - \bmu))}_2$ with probability at most $\frac{1}{8k}$ over $i \sim S$ for each $\hmu \neq \bmu$, so the $< 4k$ other hypotheses in $L_{\text{head}}$ can only remove $\frac{\alpha n}{2}$ of the points in $S$ from having nearest neighbor $\bmu$, and hence $\bmu$ will not be pruned.

We now show the key claim: that for all $\hmu \neq \bmu \in L_{\text{head}}$,
\[\Pr_{i \sim_{\textup{unif}} S}\Brack{\norm{\mg^\top\Par{X_i - \hmu}}_2 \le \norm{\mg^\top\Par{X_i - \bmu}}_2} \le \frac{1}{8k}.\]
Observe that by the triangle inequality, for any $i \in S$ satisfying the event above,
\begin{align*}
2\norm{\mg^\top\Par{X_i - \bmu}}_2 &\ge \norm{\mg^\top\Par{X_i - \bmu}}_2 + \norm{\mg^\top\Par{X_i - \hmu}}_2 \\
&\ge \norm{\mg^\top\Par{\hmu - \bmu}}_2 \ge 5\Delta \\
\implies \norm{X_i - \mus}_2 &\ge \norm{X_i - \bmu}_2 - \norm{\bmu - \mus}_2 \ge \Delta.
\end{align*}
Here we used $\norm{X_i - \bmu}_2 \ge \frac{1}{1.1}\norm{\mg^\top(X_i - \bmu)}_2 \ge 2\Delta$, and $\norm{\bmu - \mus}_2 \le \Delta$ by assumption. By Chebyshev's inequality and Assumption~\ref{assume:sexists}, we conclude for sufficiently large $\Delta = \Omega(\frac{1}{\sqrt{\alpha}})$,
\[\Pr_{i \sim_{\textup{unif}} S}\Brack{\norm{\mg^\top\Par{X_i - \hmu}}_2 \le \norm{\mg^\top\Par{X_i - \bmu}}_2} \le \Pr_{i \sim_{\textup{unif}} S}\Brack{\norm{\bmu - \mus}_2 \ge \Delta} \le \frac{1}{8k}.\]
Finally, we have shown that when the algorithm exits on Line 12, $L'$ is a maximal separated subset of a pruned list $L$ containing $\bmu$. If $\bmu \in L'$, the guarantee is immediate; otherwise, there must have been some other $\hmu \in L'$ with $\norm{\mg^\top(\hmu - \bmu)}_2 \le 5\Delta$, else $\bmu$ would have been added. For this $\hmu$,
\[\norm{\hmu - \mus}_2 \le 1.1\norm{\mg^\top(\hmu - \mus)}_2 \le 1.1\norm{\mg^\top(\hmu - \bmu)}_2 + 1.1\norm{\mg^\top(\bmu - \mus)}_2 = O(\Delta).\]
The list size bound follows from Line 6, as the returned $L'$ has at most $4k$ elements.

\textit{Runtime.} First, the cost of computing all projections $\mg^\top X$ for $X \in L \cup T$ is $O((m + n)d\log \frac d \delta)$. Next, Lines 6-11 can only be looped over at most $O(\alpha m)$ times, since every loop removes $2k$ elements from $L$ which originally has size $m$. It remains to argue about the complexity of each loop.

The cost of computing a maximal subset in Lines 5 and 10 is $O(m^2 \log \frac d \delta)$, since distance comparisons under multiplication by $\mg$ take $O(\log \frac d \delta)$ and it suffices to greedily loop over the list. Similarly, the cost of computing nearest neighbors of all elements in $T$ in Line 8 is $O(mn \log \frac d \delta)$, which is the dominant term. Combining these components yields the claim.
\end{proof}

\subsection{(Slightly) improving the error rate}\label{ssec:shavesqrtlog}

We give a brief discussion of how it is possible to shave a $\sqrt{\log \alpha^{-1}}$ factor from the error guarantees of Theorem~\ref{thm:multifilter}, bringing it to within a $\sqrt{\log \alpha^{-1}}$ factor from optimal when $\beta$ is a constant. At a high level, this extraneous factor is due to our insistence that all weight removals be $\gamma$-safe, for some $\gamma = \Theta(\log \alpha^{-1})$. This causes the thresholds required for termination of our subroutines (e.g.\ for $\SOrC$ to enter the $\Fix$ stage) to be inflated by roughly a $\gamma$ factor.

We can remove this factor by using $2$-safe scores instead of $\Theta(\log \alpha^{-1})$-safe scores, an idea introduced by \cite{DiakonikolasKKLT20} to obtain improved estimation rates over the multifilter of \cite{DiakonikolasKK20}. The idea is to restart the algorithm in \emph{phases}, where each phase corresponds to the total maintained weight being stable up to a factor of $2$ (in our case, this means subset sizes are stable up to factors of $2$).

We now summarize the changes to our algorithm. We will run the ``outer loop'' subroutine $\MultifilterBD$ (which can be viewed as constructing a multifilter tree) up until a depth of $O(\log^2 d \log \frac 1 \alpha)$ is reached, in batches of $O(\log^2 d)$ each corresponding to a stable phase. Each batch will either meet the relevant termination condition (bounded covariance, such that e.g.\ \eqref{eq:onedfiltercond} is trivially satisfied), or make progress by entering the next phase via safe weight removals.

Correspondingly, the condition \eqref{eq:filtercond} required to make improvements on the potential will be scaled differently, according to the size of the relevant set $T_p$ at some node $p$. In particular, suppose we are in a phase when $\half n' < |T_p| \le n'$. Then we will aim to guarantee
\[\inprod{\my_p^2}{\mm_{c_\ell}} \le R^2 \sqrt{\frac{n'}{n}} \Tr(\my_p^2),\]
where $R$ has the same value as in $\MultifilterBD$ up to removing a $\sqrt{\log \alpha^{-1}}$. This allows us to terminate when the operator norm of some (unnormalized) $\tcov$ matrix is $O(R^2 \sqrt{\frac{n'}{n}})$, at which point Lemma~\ref{lem:covbounddist} concludes a distance of
\[O\Par{\sqrt{R^2 \sqrt{\frac{n}{n'}} \cdot \frac{n'}{|T_p \cap S|}}} = O\Par{R \cdot \frac{\sqrt[4]{n' n}}{\sqrt{|T_p \cap S|}}} = O\Par{\frac{R}{\sqrt \alpha}},\]
where we use that $2$-saturation of $T_p$ implies $\frac{|T_p \cap S|}{n} \ge \alpha \sqrt{\frac{n'}{2n}}$ (cf.\ Lemma 1, \cite{DiakonikolasKKLT20}). 

Because of the complications this type of argument introduces, e.g.\ every one of our subroutines needs an extra exit condition (when the maintained subset enters the next phase), we omit a formal treatment in this paper. However, we remark that to remove the entire $\log \alpha^{-1}$ factor from our error likely requires new ideas. This is because both branches of our key subroutine $\SOrC$, namely $\TBOC$ and $\Fix$, require this overhead. The former is because integrating variance tail bounds decaying as $t \cdot \frac 1 {t^2}$ out to $O(\alpha)$ quantiles (cf.\ Lemma~\ref{lem:sorc}) introduces a gap of $\log(\frac 1 \alpha)$. The latter is because we employ randomize dropout to maintain subsets (rather than weights); our dropout method requires a threshold of roughly $\log \log d$ (cf.\ Lemma~\ref{lem:randdrop}) to obtain high-probability guarantees after union bounding $\text{polylog}(d)$ times. For $\alpha^{-1} = \log^{\Omega(1)} d$, this is again a $\log(\frac 1 \alpha)$ gap.

\section{Clustering mixture models}\label{sec:clustering}

We define a \emph{mixture model} to be a mixture $\sum_{i \in [k]} \alpha_i \dist_i$ where $\{\alpha_i\}_{i \in [k]} \in \R^k_{\ge 0}$, $\sum_{i \in [k]} \alpha_i = 1$, and all $\dist_i$ are supported on $\R^d$. In Sections~\ref{ssec:gmmuniform} and~\ref{ssec:gmmrobust} we handle the case where all distributions are sub-Gaussians: $\dist_i$ has mean $\mu_i$, and sub-Gaussian parameter $\le 1$ in all directions (cf.\ Section~\ref{ssec:notation}). We begin with the uncorrupted, uniform mixture case as a warmup in Section~\ref{ssec:gmmuniform}, and show how our method tolerates non-uniformity and adversarial outliers in Section~\ref{ssec:gmmrobust}. We then give a simple extension of our algorithm to handle mixtures where each component has bounded fourth moment in Section~\ref{ssec:bfmmm}, and finally tackle the case of bounded-covariance mixture models in Section~\ref{ssec:bcmm}.

Broadly, all of our clustering algorithms follow the same design framework. We first demonstrate using concentration and existence of a good hypothesis (the list-decodable learning guarantee), that the ``nearest hypothesis'' to every non-adversarial point is close to the true mean. We next prune our hypotheses down by only keeping those with a substantial number of nearby points; by arguing that the number of adversarial points (or points that appear adversarial due to anti-concentration) is small, no large ``coalition'' of bad points can be formed, and hence all kept hypotheses are near a true mean. Finally, assuming enough separation between true means, we can define a partition of the points based on their nearest hypotheses. In Section~\ref{ssec:bcmm}, we will use a more direct clustering process in the subspace spanned by candidates, combined with fast projected distance approximations, in order to obtain a tighter separation guarantee.

Throughout, we will frequently use that by Chernoff, the sum of any Bernoulli random variables whose expectation is $\Omega(d)$ will deviate from its expectation by at most any multiplicative constant with probability at least $1 - \exp(-\Omega(d))$. For example, for a dataset of size $n = \Omega(dk)$ drawn from a uniform mixture $\sum_{i \in [k]} \frac 1 k \dist_i$, each component $\dist_i$ will contribute between $0.99 \frac n k$ and $1.01 \frac n k$ points with probability at least $1 - k\exp(-\Omega(d))$, or more simply $1 - \exp(-\Omega(d))$ for $k = O(d)$.

\subsection{Clustering uniform (sub-)Gaussian mixture models}\label{ssec:gmmuniform}

We first consider the simple setting where all $\alpha_i = \frac 1 k$ and all $\dist_i$ has mean $\mu_i$ and sub-Gaussian parameter $\le 1$ in all directions. We assume access to a \emph{list-decoding} algorithm $\alg$ which returns a list $L$ of length $O(k)$, such that for each $i \in [k]$, $L$ contains $\hmu_i$ such that $\norm{\hmu_i - \mu_i}_2 \le \Delta$, for some $\Delta = \Omega(\sqrt{k})$ (in particular, $\Multifilter$ suffices for $\alg$). Finally, we assume access to a dataset $\mx = \{X_j\}_{j \in [n]}$ of size $n = \Theta(dk)$ drawn from the mixture model independently of $\alg$, where we say that each $X_j$ is ``associated with'' an index $i \in [k]$ (designating the component it is drawn from). We will now demonstrate how to cluster a dataset using calls to $\alg$, assuming a sufficiently large separation between the means of any two mixture components.

\begin{algorithm}[ht!]
	\caption{$\ClusterUGMM(\mx, L, \Delta, k, \delta)$}\label{alg:clusterugmm}
	\begin{algorithmic}[1]
		\STATE \textbf{Input:} $\mx = \{X_j\}_{j \in [n]} \sim \sum_{i \in [k]} \frac 1 k \dist_i$ where $\dist_i$ has mean $\mu_i$ and sub-Gaussian parameter $\le 1$ in all directions, and $n = \Theta(dk)$, $L$ of size $O(k)$ containing (for all $i \in [k]$) $\hmu_i \in L$ with $\norm{\hmu_i - \mu_i}_2 \le \Delta$ for $\Delta = \Omega(\sqrt{k})$, $\delta \in (0, 1)$
		\STATE $\mg \in \R^{d \times c} \gets $ entrywise $\pm \frac 1 {\sqrt c}$ uniformly at random, for $c = \Theta(\log \frac n \delta)$ (Johnson-Lindenstrauss matrix \cite{Achlioptas03})
		\STATE Let $m: [n] \to L$ map each $X_j$ to the element $\hmu \in L$ minimizing $\norm{\mg^\top (X_j - \hmu)}_2$
		\STATE Define an equivalence relation $\sim$ on $\mx$ by $X_i \sim X_j$ iff $\norm{\mg^\top(m(i) - m(j))}_2 \le 18\Delta$; if this is not an equivalence relation, then return any labeling
		\RETURN Labeling of $\mx$ associated with $\sim$
	\end{algorithmic}
\end{algorithm}

We begin with the following observation.

\begin{lemma}\label{lem:allpairsclose}
	Consider some $X_j \in \mx$ associated with $i \in [k]$. With probability at least $1 - \exp(-\Omega(\Delta^2)))$, for every pair $\hmu, \hmu' \in L$, $\hmu \neq \hmu'$ letting $v_{\hmu\hmu'}$ be the unit vector in the direction $\hmu - \hmu'$,
	\[\inprod{v_{\hmu\hmu'}}{X_j} < \inprod{v_{\hmu\hmu'}}{\mu_i} + \Delta.\]
\end{lemma}
\begin{proof}
	This is a standard application of sub-Gaussian concentration (on the one-dimensional distribution $\Nor(\langle v_{\hmu\hmu'}, \mu_i\rangle, \langle \msig_i, v_{\hmu\hmu'}v_{\hmu\hmu'}^\top\rangle)$), where we union bound across $O(k^2)$ pairs of elements in $L$. We simplify by using $\Delta = \Omega(\sqrt k)$, so the exponential term dominates the $k^2$ union bound overhead.
\end{proof}

Next, we give our key structural lemma regarding the map $m$.

\begin{lemma}\label{lem:associatedclose}
Following notation of Algorithm~\ref{alg:clusterugmm}, with probability at least $1 - \delta - n\exp(-\Omega(\Delta^2))$, every $X_j$ associated with $i \in [k]$ satisfies $\norm{m(j) - \mu_i}_2 \le 7\Delta$.
\end{lemma}
\begin{proof}
With probability at least $1 - \delta$, all pairwise distances between $\mx \cup L$ and itself are preserved by multiplication through $\mg^\top$ up to a $1 \pm 0.1$ factor \cite{Achlioptas03} (which we will call the ``Johnson-Lindenstrauss guarantee'' henceforth); condition on this event for the remainder of the proof. Suppose for contradiction that $\norm{m(j) - \mu_i}_2 > 7\Delta$, and let $\hmu_i \in L$ denote any (fixed) hypothesis which is promised to satisfy $\norm{\hmu_i - \mu_i}_2 \le \Delta$.\footnote{In the case multiple such hypotheses exist, any satisfactory (but fixed) $\{\hmu_i\}_{i \in [k]} \subseteq L$ will do.} By the triangle inequality, $\norm{m(j) - \hmu_i}_2 > 6\Delta$. Then letting $v$ be the unit vector in the direction of $m(j) - \hmu_i$, 
\begin{gather*}
\inprod{v}{\mu_i} \le \inprod{v}{\hmu_i} + \Delta,\; \inprod{v}{m(j)} > \inprod{v}{\hmu_i} + 6\Delta \\
\implies \inprod{v}{\mu_i} < \inprod{v}{m(j)} - 5\Delta.
\end{gather*}
Next, $\norm{\mg^\top(X_j - m(j))}_2 \le \norm{\mg^\top(X_j - \hmu_i)}_2$ implies $\norm{X_j - m(j)}_2 \le 2\norm{X_j - \hmu_i}_2$ by the Johnson-Lindenstrauss guarantee, or $\inprod{v}{X_j} \ge \inprod{v}{\frac{2\hmu_i + m(j)}{3}}$. Combining with the above displayed equation, 
\[\inprod{v}{X_j} \ge \frac 2 3 \inprod{v}{\hmu_i} + \frac 1 3 \inprod{v}{m(j)} > \inprod{v}{\mu_i}  + \Delta.\]
Applying Lemma~\ref{lem:allpairsclose} and union bounding over all $j \in [n]$ concludes the proof.
\end{proof}

This implies that with high probability, Algorithm~\ref{alg:clusterugmm} (given a list $L$ meeting its prerequisites) succeeds in correctly labelling all data points, assuming $\Omega(\Delta)$ separation between component means.
\begin{lemma}\label{lem:cugmmcorrect}
Suppose every pair $i, i' \in [k]$, $i \neq i''$ satisfies $\norm{\mu_i - \mu_{i'}}_2 > 34\Delta$. Then with probability at least $1 - \delta - n\exp(-\Omega(\Delta^2))$, Algorithm~\ref{alg:clusterugmm} (assuming its preconditions) outputs a correct clustering of all points (up to label permutation).
\end{lemma}
\begin{proof}
Assume the result of Lemma~\ref{lem:associatedclose} and that multiplication through $\mg^\top$ preserves all pairwise distances between $\mx \cup L$ and itself up to a $1 \pm 0.1$ factor throughout this proof. 

We first prove that for any two $X_j$, $X_{j'}$ associated to the same $i \in [k]$, Line 4 of Algorithm~\ref{alg:clusterugmm} sets $X_j \sim X_{j'}$. To see this, Lemma~\ref{lem:associatedclose} and the triangle inequality give $\norm{m(j) - m(j')}_2 \le 14\Delta$, so this will pass Line 4 by the Johnson-Lindenstrauss guarantee. Next, suppose $X_j$ is associated with $i \in [k]$ and $X_{j'}$ is associated with $i' \in [k]$ with $i \neq i'$, and suppose for contradiction $X_j \sim X_{j'}$. By the Johnson-Lindenstrauss guarantee, $\norm{m(j) - m(j')}_2 \le 20\Delta$, which yields by Lemma~\ref{lem:associatedclose} and the triangle inequality that $\norm{\mu_i - \mu_{i'}}_2 \le 34\Delta$, contradicting the separation assumption. 
\end{proof}

We conclude with the following guarantee on $\ClusterUGMM$. 

\begin{corollary}\label{cor:fullugmm}
Suppose every pair $i, i' \in [k]$, $i \neq i'$ satisfies $\norm{\mu_i - \mu_{i'}}_2 = \Omega(\sqrt{k} \log k)$ for an appropriate constant. There is an algorithm drawing $n = \Theta(dk)$ samples from the mixture $\sum_{i \in [k]} \frac 1 k \dist_i$ where $\dist_i$ has mean $\mu_i$ and sub-Gaussian parameter $\le 1$ in all directions, and returns a correct clustering of all points (up to label permutation) with probability at least
\[1 - \delta - n\exp\Par{-\Omega\Par{k\log^2 k}} - k\exp(-\Omega(d)).\]
The algorithm runs in time, for any fixed $\eps_0 > 0$,
\[O\Par{n^{1 + \eps_0} d \log^4 n \log^4 \frac n \delta + k^2 \log^4 \frac n \delta + nk\log \frac n \delta}.\]
\end{corollary}
\begin{proof}
We begin by stating the algorithm. We take $\frac 1 {10}$ of the dataset and run $\Multifilter$\footnote{If $\alpha^{-1} = \log^{o(1)} d$, we instead run Algorithm 8 of \cite{DiakonikolasKKLT20} to obtain the desired error guarantee, which fits within the runtime budget by Theorem 4 of \cite{DiakonikolasKKLT20}. Similarly, if $\alpha^{-1} = \Omega(d)$, we instead run Algorithm 14 of \cite{DiakonikolasKKLT20} which fits within the runtime budget by Proposition 9 of that paper.} on it to produce $L$ satisfying the prerequisites of Algorithm~\ref{alg:clusterugmm}, and then cluster the remaining $\frac 9 {10}$ of the dataset using $L$ via $\ClusterUGMM$; then, we take a disjoint $\frac 1 {10}$  and cluster the remaining $\frac 9 {10}$ using $\ClusterUGMM$. We then match labels based on which clusters overlap on at least $\half$ of their points between these two runs. The runtime follows from Theorem~\ref{thm:gaussian}, and the runtime of $\ClusterUGMM$, which is clearly $O(nk \log \frac n \delta)$ since Line 3 dominates, as Line 4 can be greedily implemented using distance comparisons between only $L$ once the map $m$ has been formed.

Next, for correctness, Theorem~\ref{thm:multifilter} and Proposition B.1 of \cite{CharikarSV17} (which says Assumption~\ref{assume:sexists} is met for both $\Multifilter$ runs with probability $\ge 1 - \exp(\Omega(d))$) imply both runs of $\Multifilter$ correctly return lists satisfying the precondition of Algorithm~\ref{alg:clusterugmm}; here we note that the dataset partition ensures independence of lists used and datasets clustered. Then, Lemma~\ref{lem:cugmmcorrect} implies both clusterings are completely correct on $\frac 9 {10}$ of the data. The conclusion follows from standard binomial concentration, which implies that the $\frac 8 {10}$ of the data which was held-out contains at least $\half$ the points associated with each $i \in [k]$ in the overall dataset, with probability at least $1 - \exp(-\Omega(d))$.
\end{proof}

We remark that for $n$ which grows super-exponentially in $k$ (such that the failure probability guarantee of Corollary~\ref{cor:fullugmm} becomes vacuous), it is straightforward to obtain an appropriate high-probability guarantee for clustering all points by assuming that the minimum pairwise cluster separation scales as $\sqrt{\log n}$. A similar remark also applies to Corollary~\ref{cor:fullgmm}.
 
\subsection{Robustly clustering (sub-)Gaussian mixture models}\label{ssec:gmmrobust}

In this section, we generalize Corollary~\ref{cor:fullugmm} to non-uniform corrupted mixture models. In particular, we consider an adversarially corrupted mixture model 
\begin{equation}\label{eq:genmm}\mathcal{M} = (1 - \eps)\sum_{i \in [k]} \alpha_i \dist_i + \eps \dist_{\textup{adv}},\end{equation}
where for all $i \in [k]$, $\dist_i$ has mean $\mu_i$ and sub-Gaussian parameter $\le 1$ in all directions. Moreover, for some fixed known $\alpha$, we assume all $\alpha_i \ge \alpha$ and $\eps \le \frac{\alpha}{4}$. By definition of $\alpha$, note that we must have $\alpha = O(\frac 1 k)$. In this section, we assume our list-decoding subroutine $\alg$ returns a list of size $O(\alpha^{-1})$, and guarantees estimation error $\Delta$. We now state our algorithm.

\begin{algorithm}[ht!]
	\caption{$\ClusterGMM(\mx, L, \Delta, k, \delta, \alpha)$}\label{alg:clustergmm}
	\begin{algorithmic}[1]
		\STATE \textbf{Input:} $\mx = \{X_j\}_{j \in [n]} \sim (1 - \eps) \sum_{i \in [k]} \alpha_i \dist_i + \eps \dist_{\textup{adv}}$ where $\dist_i$ has mean $\mu_i$ and sub-Gaussian parameter $\le 1$ in all directions, all $\alpha_i \ge \alpha$, $\eps \le \frac \alpha 4$, and $n = \Theta(\frac d \alpha)$, $L$ of size $O(\alpha^{-1})$ containing (for all $i \in [k]$) $\hmu_i \in L$ with $\norm{\hmu_i - \mu_i}_2 \le \Delta$ for $\Delta = \Omega(\sqrt{\alpha^{-1}})$, $\delta \in (0, 1)$
		\STATE $\mg \in \R^{d \times c} \gets $ entrywise $\pm \frac 1 {\sqrt c}$ uniformly at random, for $c = \Theta(\log \frac n \delta)$ (Johnson-Lindenstrauss matrix \cite{Achlioptas03})
		\STATE Let $m: [n] \to L$ map each $X_j$ to the element $\hmu \in L$ minimizing $\norm{\mg^\top (X_j - \hmu)}_2$
		\STATE $\mathcal{S}_{\hmu} \gets \{j \in [n] \mid m(j) = \hmu\}$ for all $\hmu \in L$, $\mathcal{B}_{\hmu} \gets \bigcup_{\hmu' \in L \mid \norm{\mg^\top(\hmu - \hmu')}_2 \le 16\Delta} \mathcal{S}_{\hmu'}$
		\STATE $L' \gets \{\hmu \in L \mid |\mathcal{B}_{\hmu}| \ge 0.9\alpha n\}$
		\STATE Define an equivalence relation $\sim$ on $\mx'$ by $X_i \sim X_j$ iff $\norm{\mg^\top(m(i) - m(j))}_2 \le 55\Delta$, for $\mx' \defeq \{X_i \in \mx \mid m(i) \in L'\}$; if this is not an equivalence relation, then return any labeling
		\RETURN Labeling of $\mx'$ associated with $\sim$, along with $\mx \setminus \mx'$ as ``unlabeled''
	\end{algorithmic}
\end{algorithm}

We again refer to $X_j$ as associated with some $i \in [k]$ if it is a draw from $\dist_i$, and as ``adversarial'' if it is drawn from $\dist_{\text{adv}}$. We begin with the following consequences of Lemma~\ref{lem:associatedclose}.

\begin{corollary}\label{cor:lprimedichotomy}
	With probability at least $1 - \delta - n\exp(-\Omega(\Delta^2)) - k\exp(-\Omega(d))$, both of the following events hold. For all $i \in [k]$, every $\hmu \in L$ with $\norm{\hmu - \mu_i}_2 \le 7\Delta$ is in $L'$. Moreover, every $\hmu \in L'$ has $\norm{\hmu - \mu_i}_2 \le 25\Delta$ for some $i \in [k]$. 
\end{corollary}
\begin{proof}
	By standard binomial concentration, for every $i \in [k]$, the set of $j \in [n]$ associated with $i$ has size at least $0.9 \alpha n$ with probability $1 - k\exp(-\Omega(d))$. Condition on this event, all pairwise distances between $\mx \cup L$ and itself being preserved up to $1 \pm 0.1$ by multiplication through $\mg^\top$, and the conclusion of Lemma~\ref{lem:allpairsclose} for the remainder of this proof.
	
	We begin with the first claim: let $\norm{\hmu - \mu_i}_2 \le 7\Delta$. For every $X_j$ associated with $i$, Lemma~\ref{lem:associatedclose} shows $\norm{m(j) - \mu_i}_2 \le 7\Delta$, implying by the Johnson-Lindenstrauss guarantee and triangle inequality, $\norm{\mg^\top (m(j) - \hmu)}_2 \le 16\Delta$. Hence $X_j$ counts towards $\mathcal{B}_{\hmu}$, which then captures all points associated with $i$, so $\hmu \in L'$. 	
	For the second claim, suppose $\norm{\hmu - \mu_i}_2 > 25\Delta$; then, no $\hmu' \in L$ with $\norm{\hmu' - \mu_i}_2 \le 7\Delta$ will count towards $\mathcal{B}_{\hmu}$, since such $\hmu'$ has $\norm{\mg^\top(\hmu - \hmu')}_2 > 16\Delta$. By Lemma~\ref{lem:associatedclose}, the only points that can contribute to $\mathcal{B}_{\hmu}$ are then the ones drawn from $\mathcal{D}_{\textup{adv}}$, which by standard binomial concentration will be less than $0.6 \alpha n$ with probability $1 - \exp(-\Omega(d))$, and hence $\hmu \not\in L'$.
\end{proof}

The following conclusion then follows immediately in the vein of Lemma~\ref{lem:cugmmcorrect}.

\begin{corollary}\label{cor:cgmmcorrect}
Suppose every pair $i, i' \in [k]$, $i \neq i'$ satisfies $\norm{\mu_i - \mu_{i'}}_2 > 55\Delta$. Then with probability at least $1 - \delta - n\exp(-\Omega(\Delta^2)) - k\exp(-\Omega(d))$, Algorithm~\ref{alg:clustergmm} (assuming its preconditions) outputs a correct clustering of all points associated with some component $i \in [k]$ (up to label permutation).
\end{corollary}
\begin{proof}
The proof is identical to Lemma~\ref{lem:cugmmcorrect}, where we note for each $\hmu \in L'$, there is a unique $i \in [k]$ promised by Corollary~\ref{cor:lprimedichotomy} where $\norm{\hmu - \mu_i}_2 \le 25\Delta$ (by our separation assumption). Thus, $\sim$ defines an equivalence relation, since every $\hmu \in L'$ is mapped to the unique equivalence class associated with $\hmu_i$. This successfully recovers the true clusters, since every point $X_j$ associated with some $i \in [k]$ will be in $\mx'$ by Corollary~\ref{cor:lprimedichotomy} and Lemma~\ref{lem:associatedclose}, and thus it will be classified correctly by $\sim$.
\end{proof}

Finally, we give a complete guarantee on $\ClusterGMM$.

\begin{corollary}\label{cor:fullgmm}
Suppose every pair $i, i' \in [k]$, $i \neq i'$ satisfies $\norm{\mu_i - \mu_{i'}}_2 = \Omega(\sqrt{\alpha^{-1}} \log \alpha^{-1})$ for an appropriate constant. There is an algorithm drawing $n = \Theta(\frac d \alpha)$ samples from the mixture $(1 - \eps) \sum_{i \in [k]} \alpha_i \dist_i + \eps \dist_{\textup{adv}}$ where $\dist_i$ has mean $\mu_i$ and sub-Gaussian parameter $\le 1$ in all directions, and returns a correct clustering of all points drawn from some $\dist_i$ (up to label permutation) with probability at least
\[1 - \delta - n\exp\Par{-\Omega\Par{\alpha^{-1} \log^2 \alpha^{-1}}} - k\exp\Par{-\Omega(d)}.\]
The algorithm runs in time, for any fixed $\eps_0 > 0$,
\[O\Par{n^{1 + \eps_0} d \log^4 n \log^4 \frac n \delta + \frac 1 {\alpha^2} \log^4 \frac n \delta+ \frac n \alpha \log \frac n \delta}.\]
\end{corollary}
\begin{proof}
The proof is the same as Corollary~\ref{cor:fullugmm}, where we note that the mislabeled points from $\dist_{\textup{adv}}$ can only affect the overlapping labels between the two runs by at most a $1.1 \eps \le  \frac {1.1} 4 \alpha$ fraction, a minority compared to the number of correct labels (due to the true mixture), which in both runs will contribute at least $\frac 4 5 \alpha_i$ of the points assigned by $\sim$ to the true clustering. Hence, between runs the clusters assigned by $\sim$ to the points drawn from $\dist_i$ will overlap on at least half their points, and we can use this agreement to correctly label all points not drawn from $\dist_{\textup{adv}}$ (where throughout, we conditioned on all high-probability events in Corollary~\ref{cor:fullugmm}'s proof holding).
\end{proof}

\subsection{Mixture models with bounded fourth moments}\label{ssec:bfmmm}

In this section, we consider non-uniform corrupted mixture models under a weaker distributional assumption based on bounded fourth moments. We will no longer be able to correctly cluster all (non-adversarial) points, but will instead guarantee that amongst non-adversarial points, at least a $1 - o(\alpha)$ fraction are correctly classified. Concretely, we consider the adversarially corrupted mixture model \eqref{eq:genmm}, where for all $i \in [k]$, and some constant $C = O(1)$,
\begin{equation}\label{eq:bfm}\E_{X \sim \dist_i}\Brack{\inprod{v}{X - \mu_i}^4} \le C \text{ for all } v \in \R^d,\;\norm{v}_2 = 1.\end{equation}
We remark that this fourth-moment bound also implies the covariance is bounded by $O(1)\id$ for all components, by Jensen's inequality. We again assume that all $\alpha_i \ge \alpha$, and $\eps \le \frac \alpha 4$. Our algorithm for this setting will be exactly the same as $\ClusterGMM$, but for convenience we restate it under the new distributional assumptions. We also assume $\Delta = \omega(\sqrt {\alpha^{-1}})$ for notational simplicity. 

\begin{algorithm}[ht!]
	\caption{$\ClusterBFMM(\mx, L, \Delta, k, \delta, \alpha)$}\label{alg:clusterbfmm}
	\begin{algorithmic}[1]
		\STATE \textbf{Input:} $\mx = \{X_j\}_{j \in [n]} \sim (1 - \eps) \sum_{i \in [k]} \alpha_i \dist_i + \eps \dist_{\textup{adv}}$ where $\dist_i$ has mean $\mu_i$ and satisfies \eqref{eq:bfm}, all $\alpha_i \ge \alpha$, $\eps \le \frac \alpha 4$, and $n = \Theta(\frac d \alpha)$, $L$ of size $O(\alpha^{-1})$ containing (for all $i \in [k]$) $\hmu_i \in L$ with $\norm{\hmu_i - \mu_i}_2 \le \Delta$ for $\Delta = \omega(\sqrt{\alpha^{-1}})$, $\delta \in (0, 1)$
		\STATE $\mg \in \R^{d \times c} \gets $ entrywise $\pm \frac 1 {\sqrt c}$ uniformly at random, for $c = \Theta(\log \frac n \delta)$ (Johnson-Lindenstrauss matrix \cite{Achlioptas03})
		\STATE Let $m: [n] \to L$ map each $X_j$ to the element $\hmu \in L$ minimizing $\norm{\mg^\top (X_j - \hmu)}_2$
		\STATE $\mathcal{S}_{\hmu} \gets \{j \in [n] \mid m(j) = \hmu\}$ for all $\hmu \in L$, $\mathcal{B}_{\hmu} \gets \bigcup_{\hmu' \in L \mid \norm{\mg^\top(\hmu - \hmu')}_2 \le 15\Delta} \mathcal{S}_{\hmu'}$
		\STATE $L' \gets \{\hmu \in L \mid |\mathcal{B}_{\hmu}| \ge 0.9 \alpha n\}$
		\STATE Define an equivalence relation $\sim$ on $\mx'$ by $X_i \sim X_j$ iff $\norm{\mg^\top(m(i) - m(j))}_2 \le 55\Delta$, for $\mx' \defeq \{X_i \in \mx \mid m(i) \in L'\}$; if this is not an equivalence relation, then return any labeling
		\RETURN Labeling of $\mx'$ associated with $\sim$, along with $\mx \setminus \mx'$ as ``unlabeled''
	\end{algorithmic}
\end{algorithm}

Our analysis of $\ClusterBFMM$ follows from the following key observation: define $X_j \in \mx$ as ``adversarial'' if it is drawn from $\dist_{\textup{adv}}$, and say it is ``pseudo-adversarial'' if it is drawn from some component $i \in [k]$, but satisfies $\norm{m(j) - \mu_i}_2 > 7\Delta$. Then, by the bounded fourth moment assumption \eqref{eq:bfm}, there are few pseudo-adversarial points, made rigorous as follows.

\begin{lemma}\label{lem:pseudobfm}
With probability at least $1 - \delta - \exp(-\Omega(d))$, the number of pseudo-adversarial points in $\mx$ is $o(\alpha n)$.
\end{lemma} 
\begin{proof}
We first bound the probability that some $X_j \sim \dist_i$ will be pseudo-adversarial. By the proof of Lemma~\ref{lem:associatedclose}, if $X_j$ were pseudo-adversarial, it must be the case that in the direction $v$ corresponding to $m(j) - \hmu_i$, $\inprod{v}{X_j} > \inprod{v}{\mu_i} + \Delta$. However, by \eqref{eq:bfm} and Markov,
\[\Pr_{X \sim \dist_i}\Brack{\inprod{v}{X - \mu_i}^4 \ge \Delta^4} \le \frac{C}{\Delta^4} = o\Par{\alpha^2}.\]
Union bounding over all possible directions $v$ (one for each of the $O(\alpha^{-1})$ elements of $L$ other than $\hmu_i$), this implies the chance that $X$ drawn from \emph{any} $\dist_i$ is pseudo-adversarial is $o(\alpha)$. Hence, the expected number of pseudo-adversarial points in our entire dataset is $o(\alpha n)$, and the conclusion follows by applying standard binomial concentration.
\end{proof}

Now, Corollary~\ref{cor:lprimedichotomy} holds for $\ClusterBFMM$ simply by lumping together the pseudo-adversarial points and the adversarial points in its proof (in particular, no hypothesis more than $25\Delta$ away from a true mean can capture any points in the dataset other than adversarial and pseudo-adversarial ones). We now state analogs of Corollaries~\ref{cor:cgmmcorrect} and~\ref{cor:fullgmm}.

\begin{corollary}\label{cor:cbfmcorrect}
Suppose every pair $i, i' \in [k]$, $i \neq i'$ satisfies $\norm{\mu_i - \mu_{i'}}_2 > 55\Delta$. Then with probability at least $1 - \delta - n\exp(-\Omega(\Delta^2)) - k\exp(-\Omega(d))$, Algorithm~\ref{alg:clusterbfmm} (assuming its preconditions) outputs a correct clustering of a $1 - o(\alpha)$ proportion of points associated with some component $i \in [k]$ (up to label permutation).
\end{corollary}
\begin{proof}
The proof is identical to Corollary~\ref{cor:cgmmcorrect}, where we only guarantee correctly labeling points which are not pseudo-adversarial; $\sim$ will correctly cluster all such points by separation and Corollary~\ref{cor:lprimedichotomy}. 
\end{proof}

\begin{corollary}\label{cor:fullbfmm}
Suppose every pair $i, i' \in [k]$, $i \neq i'$ satisfies $\norm{\mu_i - \mu_{i'}}_2 = \Omega(\sqrt {\alpha^{-1}} \log \alpha^{-1})$ for an appropriate constant. There is an algorithm drawing $n = \Theta(\frac d \alpha)$ samples from the mixture $(1 - \eps) \sum_{i \in [k]} \alpha_i \dist_i + \eps \dist_{\textup{adv}}$ where $\dist_i$ has mean $\mu_i$ and satisfies the bounded fourth moment condition \eqref{eq:bfm}, and returns a correct clustering of a $1 - o(\alpha)$ fraction of all points drawn from some $\dist_i$ (up to label permutation) with probability at least
\[1 - \delta - k\exp\Par{-\Omega(d)}.\]
The algorithm runs in time, for any fixed $\eps_0 > 0$,
\[O\Par{n^{1 + \eps_0} d \log^4 n \log^4 \frac n \delta + \frac 1 {\alpha^2} \log^4 \frac n \delta+ \frac n \alpha \log \frac n \delta}.\]
\end{corollary}
\begin{proof}
The proof is identical to Corollary~\ref{cor:fullgmm}, using Corollary~\ref{cor:cbfmcorrect} instead of Corollary~\ref{cor:cgmmcorrect}, and using that the pseudo-adversarial points will not substantially affect cluster overlap guarantees. 
\end{proof}

\subsection{Bounded-covariance mixture models}\label{ssec:bcmm}

In this section, we finally handle the case of non-uniform corrupted mixture models under only a bounded-covariance assumption; in particular, we study \eqref{eq:genmm} where every component $\dist_i$ for $i \in [k]$ has covariance bounded by $\id$. We again assume that all $\alpha_i \ge \alpha$, and $\eps \le \frac \alpha 4$.

\begin{algorithm}[ht!]
	\caption{$\ClusterBCMM(\mx, L, \Delta, k, \delta, \alpha)$}\label{alg:clusterbcmm}
	\begin{algorithmic}[1]
		\STATE \textbf{Input:} $\mx = \{X_j\}_{j \in [n]} \sim (1 - \eps) \sum_{i \in [k]} \alpha_i \dist_i + \eps \dist_{\textup{adv}}$ where $\dist_i$ has mean $\mu_i$ and has covariance bounded by $\id$, all $\alpha_i \ge \alpha$, $\eps \le \frac \alpha 4$, and $n = \Theta(\frac d \alpha)$, $L$ of size $O(\alpha^{-1})$ containing (for all $i \in [k]$) $\hmu_i \in L$ with $\norm{\hmu_i - \mu_i}_2 \le \Delta$ for $\Delta = \Omega(\sqrt{\alpha^{-1}} \log \alpha^{-1})$, $\delta \in (0, 1)$
		\STATE $\mg \in \R^{d \times c} \gets $ entrywise $\pm \frac 1 {\sqrt c}$ uniformly at random, for $c = \Theta(\log \frac n \delta)$ (Johnson-Lindenstrauss matrix \cite{Achlioptas03})
		\STATE $\mproj \gets \ml^\top (\ml \ml^\top)^{-1} \ml$ (i.e.\ projection matrix onto the subspace spanned by $L$) where $\ml \in \R^{O(\alpha^{-1}) \times d}$ is the vertical concatenation of $L$
		\STATE $\tX_j \gets \mg^\top \mproj X_j$ for all $j \in [n]$, $\tmu_i \gets \mg^\top \hmu_i$ for all $\hmu_i \in L$
		\STATE $\tL \gets \Brace{\hmu_i \in L \,\Big\vert \left|\mathcal{S}(\hmu_i)\right| \ge \frac{\alpha n}{2}}$, where $\mathcal{S}(\hmu_i) \defeq \Brace{X_j \in \mx \,\Big\vert \norm{\tX_j - \tmu_i}_2 \le 2.5\Delta}$
		\STATE Define an equivalence relation $\sim$ on $\tL$ by $\hmu_i \sim \hmu_j$ iff $\norm{\tmu_i - \tmu_j}_2 \le 20\Delta$; if this is not an equivalence relation, then return any labeling
		\RETURN Labeling of $\mx$ where $\tX_j$ and $\tX_{j'}$ have the same label iff $\tX_j \in \mathcal{S}(\hmu_i)$ and $\tX_{j'} \in \mathcal{S}(\hmu_{i'})$ for $\hmu_i \sim \hmu_{i'}$
	\end{algorithmic}
\end{algorithm}

We briefly describe Algorithm~\ref{alg:clusterbcmm}. Line 3 forms the projection matrix onto the span of $L$. Line 5 ``prunes'' the set $L$ to a set $\tL$ where we only keep candidates with enough data points close by (in the subspace spanned by $L$). Line 6 then appropriately partitions the candidates, and Line 7 partitions the dataset based on the candidate partition. In the following proof we use $\mproj \hmu = \hmu$ for all $\hmu \in L$, and that all $\mu_i$ have a point in the subspace projected to by $\mproj$ at most $\Delta$ away.

\begin{lemma}\label{lem:boundcovmain}
With probability at least $1 - \delta - k\exp(-\Omega(d))$, if $\Delta = \Omega(\sqrt{\alpha^{-1}} \log \alpha^{-1})$, and every pair $i, i' \in [k]$, $i \neq i'$ satisfies $\norm{\mu_i - \mu_{i'}}_2 > 20\Delta$, Algorithm~\ref{alg:clusterbcmm} returns a correct clustering of a $1 - o(1)$ fraction of all pointns drawn from some $\dist_i$. The algorithm runs in time
\[O\Par{n \cdot \Par{d + \frac 1 \alpha} \cdot \log \frac n \delta}.\]
\end{lemma}
\begin{proof}
Throughout this proof, condition on the event that distances between $\mproj \mx \cup L$ are preserved up to $1 \pm 0.1$ by multiplication through $\mg^\top$ and that there are at most $\frac{\alpha n}{3}$ adversarial points; we will union bound this conditioning with an event of probability $1 - k\exp(-\Omega(d))$.

We first claim that any $\hmu \in L$ satisfying $\norm{\hmu - \mu_i}_2 \le \Delta$ for some $i \in [k]$ is contained in $\tL$, with probability at least $1 - \exp(-\Omega(d))$: call this claim $\dagger$. To see $\dagger$, with probability at least $1 - \exp(-\Omega(d))$, a $\frac{9}{10}$ fraction of $X_j \in \mx$ sampled from $\dist_i$ satisfy $\norm{\mproj(X_j - \mu_i)}_2 \le \Delta$. This follows since Chebyshev's inequality, $\Delta = \omega(\sqrt{\alpha^{-1}})$, and $\mproj$ projecting onto a $O(\alpha^{-1})$-dimensional subspace imply the probability $\norm{\mproj(X_j - \mu_i)}_2 \le \Delta$ is $o(1)$. By triangle inequality and $\mproj \preceq \id$, for such $X_j$,
\[\norm{\mproj(X_j - \hmu)}_2 \le \norm{\mproj(X_j - \mu_i)}_2 + \norm{\mproj(\hmu - \mu_i)}_2 \le \Delta + \norm{\hmu - \mu_i}_2 \le 2\Delta.\]
Since $\mg^\top$ preserves distances to a $1.1$ factor, this implies $\norm{\tX_j - \tmu}_2 \le 2.5\Delta$ as desired.

We next claim that any $\hmu \in L$ satisfying $\norm{\hmu - \mu_i}_2 > 7\Delta$ for all $i \in [k]$ will not be contained in $\tL$, with probability at least $1 - \exp(-\Omega(d))$. To see this, we claim $\mathcal{S}(\hmu)$ can only contain an $o(\alpha)$ fraction of the points drawn from $\dist_i$; summing over all $\dist_i$, and adding in all $\le \frac{\alpha n}{3}$ adversarial points, shows $|\mathcal{S}(\hmu)| < \frac{\alpha n}{2}$. This latter claim follows by first observing that for $X_j \in \mathcal{S}(\hmu)$,
\begin{align*}\norm{\mproj(X_j - \mu_i)}_2 &\ge \norm{\mu_i - \hmu}_2 - \norm{\mu_i - \mproj\mu_i}_2 - \norm{\mproj(X_j - \hmu)}_2 \\
&> 7\Delta - \Delta - 3\Delta = 3\Delta \ge \norm{\mproj(X_j - \hmu)}_2.
\end{align*}
The first inequality used that $\mproj \hmu = \hmu$, and the triangle inequality twice. In the second inequality, we used $\norm{\mu_i - \mproj \mu_i}_2 \le \Delta$ since the subspace projected to by $\mproj$ contains a point at most $\Delta$ away from $\mu_i$, and $\norm{\mproj(X_j - \hmu)}_2 \le 3\Delta$ by the assumption on $\mg$ and $X_j \in \mathcal{S}(\hmu)$. This implies that the projection of $X_j$ is closer to $\mproj \hmu = \hmu$ than $\mproj \mu_i$. Letting $\dist_i'$ be the projection of $\dist_i$ by $\mproj$, $X \sim \dist'_i$ is closer to $\hmu$ than $\mproj \mu_i$ with probability at most $\frac{1}{\Omega(\Delta^2)} = o(\alpha)$ by arguments in Lemma~\ref{lem:associatedclose} and Chebyshev. Hence, the probability that $X_j \in \mathcal{S}(\hmu)$ is $o(\alpha)$, as desired.

Now, under the success of all above conditioning events, by the minimum separation assumption of $20\Delta$, each surviving $\hmu \in \tL$ has a unique $\mu_i$ such that $\norm{\hmu - \mu_i}_2 \le 7\Delta$, and each designated $\hmu_i$ with $\norm{\hmu_i - \mu_i}_2 \le \Delta$ survives. Hence, the equivalence partition succeeds and captures all $\hmu \in \tL$ at distance at most $7\Delta$ from a $\mu_i$. Moreover, for every pair $\hmu, \hmu' \in \tL$ such that $\norm{\hmu - \mu_i}_2 \le 7\Delta$ and $\norm{\hmu' - \mu_{i'}}_2 \le 7\Delta$ for some $i \neq i'$,
\[X \in \mathcal{S}(\hmu) \cap \mathcal{S}(\hmu') \implies \norm{\hmu - \hmu'}_2 \le 6\Delta. \]
Since the minimum separation between $\mu_i$ and $\mu_{i'}$ is $20\Delta$, this is a contradiction by the triangle inequality. So, no sets $\mathcal{S}(\hmu)$ and $\mathcal{S}(\hmu')$ intersect, where the hypotheses are close to different means. Thus the labeling in Line 7 is well-defined. It remains to show that if $X_j \sim \dist_i$, we will have $X_j \in \mathcal{S}(\hmu_i)$ with probability $1 - o(1)$ (and hence a $1 - o(1)$ fraction of points is labeled correctly). This was in fact shown earlier, when we proved the claim $\dagger$.

Finally, it is clear that all operations can be performed in the desired runtime: $\mg^\top \mproj$ can be computed explicitly in time $O(d^2\alpha^{-1} + d^2 \log \frac n \delta) = O(nd + d^2 \log \frac n \delta)$ via na\"ive matrix multiplication, and all projection and distance computations (multiplying all points and hypotheses by $\mg^\top \mproj$, computing all $\mathcal{S}(\hmu)$, and checking all pairwise distances in Line 6) take time $O(n \cdot (d + \frac 1 \alpha)\cdot \log \frac n \delta)$. 
\end{proof}

\begin{corollary}\label{cor:fullbcmm}
Suppose every pair $i, i' \in [k]$, $i \neq i'$ satisfies $\norm{\mu_i - \mu_{i'}}_2 = \Omega(\sqrt{\alpha^{-1}} \log \alpha^{-1})$ for an appropriate constant. There is an algorithm drawing $n = \Theta(\frac d \alpha)$ samples from the mixture $(1 - \eps)\sum_{i \in [k]} \alpha_i \dist_i + \eps \dist_{\textup{adv}}$ where $\dist_i$ has mean $\mu_i$ and covariance bounded by $\id$, and returns a correct clustering of a $1 - \eps_1$ fraction of all points drawn from some $\dist_i$ (up to label permutation) for any fixed $\eps_1 > 0$, with probability at least 
\[1 - \delta - k\exp(-\Omega(d)).\]
The algorithm runs in time, for any fixed $\eps_0 > 0$,
\[O\Par{n^{1 + \eps_0} d \log^4 n \log^4 \frac n \delta + \frac 1 {\alpha^2} \log^4 \frac n \delta + \frac n \alpha \log \frac n \delta}.\]
\end{corollary}
\begin{proof}
Since the goal is to correctly label a $1 - \eps_1$ fraction of all points, we can use a $\frac{\eps_1}{2}$ fraction of our dataset as holdout for learning an independent list $L$. We then use this list to cluster a $1 - \frac{\eps_1}{2}$ fraction of the remaining points via Lemma~\ref{lem:boundcovmain}.
\end{proof}

\subsection*{Acknowledgments}
Ilias Diakonikolas is supported by NSF Medium Award CCF-2107079,
NSF Award CCF-1652862 (CAREER), a Sloan Research Fellowship, and
a DARPA Learning with Less Labels (LwLL) grant.
Daniel Kane is supported by NSF Medium Award CCF-2107547, 
NSF CAREER Award ID 1553288, and a Sloan fellowship. 
Kevin Tian is supported by a Google Ph.D.\ Fellowship, a Simons-Berkeley VMware Research Fellowship, NSF CAREER Award CCF-1844855, 
NSF Grant CCF-1955039, and the Alfred P.\ Sloan Foundation. 

\newpage
\bibliographystyle{alpha}	
\bibliography{nd}

\newcommand{\etalchar}[1]{$^{#1}$}
\begin{thebibliography}{ABDH{\etalchar{+}}18}

\bibitem[ABDH{\etalchar{+}}18]{ashtiani2018nearly}
Hassan Ashtiani, Shai Ben-David, Nicholas~JA Harvey, Christopher Liaw, Abbas
  Mehrabian, and Yaniv Plan.
\newblock Nearly tight sample complexity bounds for learning mixtures of
  gaussians via sample compression schemes.
\newblock In {\em Proceedings of the 32nd International Conference on Neural
  Information Processing Systems}, pages 3416--3425, 2018.

\bibitem[ABG{\etalchar{+}}14]{anderson2014more}
Joseph Anderson, Mikhail Belkin, Navin Goyal, Luis Rademacher, and James Voss.
\newblock The more, the merrier: the blessing of dimensionality for learning
  large gaussian mixtures.
\newblock In {\em Conference on Learning Theory}, pages 1135--1164. PMLR, 2014.

\bibitem[Ach03]{Achlioptas03}
Dimitris Achlioptas.
\newblock Database-friendly random projections: Johnson-lindenstrauss with
  binary coins.
\newblock {\em J. Comput. Syst. Sci.}, 66(4):671--687, 2003.

\bibitem[ADLS17]{acharya2017sample}
Jayadev Acharya, Ilias Diakonikolas, Jerry Li, and Ludwig Schmidt.
\newblock Sample-optimal density estimation in nearly-linear time.
\newblock In {\em Proceedings of the Twenty-Eighth Annual ACM-SIAM Symposium on
  Discrete Algorithms}, pages 1278--1289. SIAM, 2017.

\bibitem[AJOS14]{acharya2014near}
Jayadev Acharya, Ashkan Jafarpour, Alon Orlitsky, and Ananda~Theertha Suresh.
\newblock Near-optimal-sample estimators for spherical gaussian mixtures.
\newblock {\em arXiv preprint arXiv:1402.4746}, 2014.

\bibitem[AK05]{arora2005learning}
Sanjeev Arora and Ravi Kannan.
\newblock Learning mixtures of separated nonspherical gaussians.
\newblock {\em The Annals of Applied Probability}, 15(1A):69--92, 2005.

\bibitem[AK07]{AroraK07}
Sanjeev Arora and Satyen Kale.
\newblock A combinatorial, primal-dual approach to semidefinite programs.
\newblock In {\em Proceedings of the 39th Annual {ACM} Symposium on Theory of
  Computing, San Diego, California, USA, June 11-13, 2007}, pages 227--236,
  2007.

\bibitem[AM05]{achlioptas2005spectral}
Dimitris Achlioptas and Frank McSherry.
\newblock On spectral learning of mixtures of distributions.
\newblock In {\em International Conference on Computational Learning Theory},
  pages 458--469. Springer, 2005.

\bibitem[Ans60]{anscombe1960rejection}
Frank~J Anscombe.
\newblock Rejection of outliers.
\newblock {\em Technometrics}, 2(2):123--146, 1960.

\bibitem[AS12]{awasthi2012improved}
Pranjal Awasthi and Or~Sheffet.
\newblock Improved spectral-norm bounds for clustering.
\newblock In {\em Approximation, Randomization, and Combinatorial Optimization.
  Algorithms and Techniques}, pages 37--49. Springer, 2012.

\bibitem[Awa21]{Aw21}
Pranjal Awasthi.
\newblock Personal communication, September 2021.

\bibitem[BBV08]{balcan2008discriminative}
Maria-Florina Balcan, Avrim Blum, and Santosh Vempala.
\newblock A discriminative framework for clustering via similarity functions.
\newblock In {\em Proceedings of the fortieth annual ACM symposium on Theory of
  computing}, pages 671--680, 2008.

\bibitem[BCMV14]{bhaskara2014smoothed}
Aditya Bhaskara, Moses Charikar, Ankur Moitra, and Aravindan Vijayaraghavan.
\newblock Smoothed analysis of tensor decompositions.
\newblock In {\em Proceedings of the forty-sixth annual ACM symposium on Theory
  of computing}, pages 594--603, 2014.

\bibitem[BDH{\etalchar{+}}20]{BakshiDHKKK20}
Ainesh Bakshi, Ilias Diakonikolas, Samuel~B. Hopkins, Daniel Kane, Sushrut
  Karmalkar, and Pravesh~K. Kothari.
\newblock Outlier-robust clustering of gaussians and other non-spherical
  mixtures.
\newblock In {\em 61st {IEEE} Annual Symposium on Foundations of Computer
  Science, {FOCS} 2020}, pages 149--159. {IEEE}, 2020.

\bibitem[BDJ{\etalchar{+}}20]{bakshi2020robustly}
Ainesh Bakshi, Ilias Diakonikolas, He~Jia, Daniel~M Kane, Pravesh~K Kothari,
  and Santosh~S Vempala.
\newblock Robustly learning mixtures of $ k $ arbitrary gaussians.
\newblock {\em arXiv preprint arXiv:2012.02119}, 2020.

\bibitem[BDLS17]{balakrishnan2017computationally}
Sivaraman Balakrishnan, Simon~S Du, Jerry Li, and Aarti Singh.
\newblock Computationally efficient robust sparse estimation in high
  dimensions.
\newblock In {\em Conference on Learning Theory}, pages 169--212, 2017.

\bibitem[BK20a]{bakshi2020list}
Ainesh Bakshi and Pravesh Kothari.
\newblock List-decodable subspace recovery via sum-of-squares.
\newblock {\em arXiv preprint arXiv:2002.05139}, 2020.

\bibitem[BK20b]{bakshi2020outlier}
Ainesh Bakshi and Pravesh Kothari.
\newblock Outlier-robust clustering of non-spherical mixtures.
\newblock {\em arXiv preprint arXiv:2005.02970}, 2020.

\bibitem[BNJT10]{BarrenoNJT10}
Marco Barreno, Blaine Nelson, Anthony~D. Joseph, and J.~D. Tygar.
\newblock The security of machine learning.
\newblock {\em Mach. Learn.}, 81(2):121--148, 2010.

\bibitem[BNL12]{BiggioNL12}
Battista Biggio, Blaine Nelson, and Pavel Laskov.
\newblock Poisoning attacks against support vector machines.
\newblock In {\em Proceedings of the 29th International Conference on Machine
  Learning, {ICML} 2012, Edinburgh, Scotland, UK, June 26 - July 1, 2012},
  2012.

\bibitem[BS15]{belkin2015polynomial}
Mikhail Belkin and Kaushik Sinha.
\newblock Polynomial learning of distribution families.
\newblock {\em SIAM Journal on Computing}, 44(4):889--911, 2015.

\bibitem[BV08]{BrubakerV08}
S.~Charles Brubaker and Santosh~S. Vempala.
\newblock Isotropic {PCA} and affine-invariant clustering.
\newblock In {\em 49th Annual {IEEE} Symposium on Foundations of Computer
  Science, {FOCS} 2008, October 25-28, 2008, Philadelphia, PA, {USA}}, pages
  551--560, 2008.

\bibitem[CDG19]{cheng2019high}
Yu~Cheng, Ilias Diakonikolas, and Rong Ge.
\newblock High-dimensional robust mean estimation in nearly-linear time.
\newblock In {\em Proceedings of the Thirtieth Annual ACM-SIAM Symposium on
  Discrete Algorithms}, pages 2755--2771. SIAM, 2019.

\bibitem[CDGW19]{cheng2019faster}
Yu~Cheng, Ilias Diakonikolas, Rong Ge, and David~P Woodruff.
\newblock Faster algorithms for high-dimensional robust covariance estimation.
\newblock In {\em Conference on Learning Theory}, pages 727--757, 2019.

\bibitem[CDKS18]{cheng2018robust}
Yu~Cheng, Ilias Diakonikolas, Daniel Kane, and Alistair Stewart.
\newblock Robust learning of fixed-structure bayesian networks.
\newblock In {\em Advances in Neural Information Processing Systems}, pages
  10283--10295, 2018.

\bibitem[CDSS14]{chan2014efficient}
Siu-On Chan, Ilias Diakonikolas, Rocco~A Servedio, and Xiaorui Sun.
\newblock Efficient density estimation via piecewise polynomial approximation.
\newblock In {\em Proceedings of the forty-sixth annual ACM symposium on Theory
  of computing}, pages 604--613, 2014.

\bibitem[CMY20]{CherapanamjeriMY20}
Yeshwanth Cherapanamjeri, Sidhanth Mohanty, and Morris Yau.
\newblock List decodable mean estimation in nearly linear time.
\newblock In {\em 61st IEEE Annual Symposium on Foundations of Computer
  Science, FOCS 2020}, 2020.

\bibitem[CSV17]{CharikarSV17}
Moses Charikar, Jacob Steinhardt, and Gregory Valiant.
\newblock Learning from untrusted data.
\newblock In {\em Proceedings of the 49th Annual {ACM} {SIGACT} Symposium on
  Theory of Computing, {STOC} 2017, Montreal, QC, Canada, June 19-23, 2017},
  pages 47--60, 2017.

\bibitem[Das99]{dasgupta1999learning}
Sanjoy Dasgupta.
\newblock Learning mixtures of gaussians.
\newblock In {\em 40th Annual Symposium on Foundations of Computer Science
  (Cat. No. 99CB37039)}, pages 634--644. IEEE, 1999.

\bibitem[DDS{\etalchar{+}}09]{deng2009imagenet}
Jia Deng, Wei Dong, Richard Socher, Li-Jia Li, Kai Li, and Li~Fei-Fei.
\newblock Imagenet: A large-scale hierarchical image database.
\newblock In {\em 2009 IEEE conference on computer vision and pattern
  recognition}, pages 248--255. Ieee, 2009.

\bibitem[DHKK20]{diakonikolas2020robustly}
Ilias Diakonikolas, Samuel~B Hopkins, Daniel Kane, and Sushrut Karmalkar.
\newblock Robustly learning any clusterable mixture of gaussians.
\newblock {\em arXiv preprint arXiv:2005.06417}, 2020.

\bibitem[DHL19]{DongH019}
Yihe Dong, Samuel~B. Hopkins, and Jerry Li.
\newblock Quantum entropy scoring for fast robust mean estimation and improved
  outlier detection.
\newblock In {\em Advances in Neural Information Processing Systems 32: Annual
  Conference on Neural Information Processing Systems 2019, NeurIPS 2019, 8-14
  December 2019, Vancouver, BC, Canada}, pages 6065--6075, 2019.

\bibitem[DISZ18]{DaskalakisISZ18}
Constantinos Daskalakis, Andrew Ilyas, Vasilis Syrgkanis, and Haoyang Zeng.
\newblock Training gans with optimism.
\newblock In {\em 6th International Conference on Learning Representations,
  {ICLR} 2018, Vancouver, BC, Canada, April 30 - May 3, 2018, Conference Track
  Proceedings}, 2018.

\bibitem[DK14]{daskalakis2014faster}
Constantinos Daskalakis and Gautam Kamath.
\newblock Faster and sample near-optimal algorithms for proper learning
  mixtures of gaussians.
\newblock In {\em Conference on Learning Theory}, pages 1183--1213. PMLR, 2014.

\bibitem[DK19]{diakonikolas2019recent}
Ilias Diakonikolas and Daniel~M Kane.
\newblock Recent advances in algorithmic high-dimensional robust statistics.
\newblock {\em arXiv preprint arXiv:1911.05911}, 2019.

\bibitem[DK20]{DiakonikolasK20}
Ilias Diakonikolas and Daniel~M. Kane.
\newblock Small covers for near-zero sets of polynomials and learning latent
  variable models.
\newblock In {\em 61st {IEEE} Annual Symposium on Foundations of Computer
  Science, {FOCS} 2020}, pages 184--195. {IEEE}, 2020.

\bibitem[DKK{\etalchar{+}}17]{DiakonikolasKK017}
Ilias Diakonikolas, Gautam Kamath, Daniel~M. Kane, Jerry Li, Ankur Moitra, and
  Alistair Stewart.
\newblock Being robust (in high dimensions) can be practical.
\newblock In {\em Proceedings of the 34th International Conference on Machine
  Learning, {ICML} 2017, Sydney, NSW, Australia, 6-11 August 2017}, pages
  999--1008, 2017.

\bibitem[DKK{\etalchar{+}}19a]{diakonikolas2019robust}
Ilias Diakonikolas, Gautam Kamath, Daniel Kane, Jerry Li, Ankur Moitra, and
  Alistair Stewart.
\newblock Robust estimators in high-dimensions without the computational
  intractability.
\newblock {\em SIAM Journal on Computing}, 48(2):742--864, 2019.

\bibitem[DKK{\etalchar{+}}19b]{diakonikolas2019sever}
Ilias Diakonikolas, Gautam Kamath, Daniel Kane, Jerry Li, Jacob Steinhardt, and
  Alistair Stewart.
\newblock Sever: A robust meta-algorithm for stochastic optimization.
\newblock In {\em International Conference on Machine Learning}, pages
  1596--1606, 2019.

\bibitem[DKK{\etalchar{+}}19c]{DiakonikolasKKP19}
Ilias Diakonikolas, Daniel Kane, Sushrut Karmalkar, Eric Price, and Alistair
  Stewart.
\newblock Outlier-robust high-dimensional sparse estimation via iterative
  filtering.
\newblock In Hanna~M. Wallach, Hugo Larochelle, Alina Beygelzimer, Florence
  d'Alch{\'{e}}{-}Buc, Emily~B. Fox, and Roman Garnett, editors, {\em Advances
  in Neural Information Processing Systems 32: Annual Conference on Neural
  Information Processing Systems 2019, NeurIPS 2019}, pages 10688--10699, 2019.

\bibitem[DKK20a]{DiakonikolasKK20}
Ilias Diakonikolas, Daniel Kane, and Daniel Kongsgaard.
\newblock List-decodable mean estimation via iterative multi-filtering.
\newblock In {\em Advances in Neural Information Processing Systems 33: Annual
  Conference on Neural Information Processing Systems 2020, NeurIPS 2020,
  December 6-12, 2020, virtual}, 2020.

\bibitem[DKK{\etalchar{+}}20b]{DiakonikolasKKLT20}
Ilias Diakonikolas, Daniel~M. Kane, Daniel Kongsgaard, Jerry Li, and Kevin
  Tian.
\newblock List-decodable mean estimation in nearly-pca time.
\newblock {\em CoRR}, abs/2011.09973, 2020.

\bibitem[DKS18]{DiakonikolasKS18}
Ilias Diakonikolas, Daniel~M. Kane, and Alistair Stewart.
\newblock List-decodable robust mean estimation and learning mixtures of
  spherical gaussians.
\newblock In {\em Proceedings of the 50th Annual {ACM} {SIGACT} Symposium on
  Theory of Computing, {STOC} 2018, Los Angeles, CA, USA, June 25-29, 2018},
  pages 1047--1060, 2018.

\bibitem[DKS19]{DiakonikolasKS19}
Ilias Diakonikolas, Weihao Kong, and Alistair Stewart.
\newblock Efficient algorithms and lower bounds for robust linear regression.
\newblock In Timothy~M. Chan, editor, {\em Proceedings of the Thirtieth Annual
  {ACM-SIAM} Symposium on Discrete Algorithms, {SODA} 2019}, pages 2745--2754.
  {SIAM}, 2019.

\bibitem[DL12]{devroye2012combinatorial}
Luc Devroye and G{\'a}bor Lugosi.
\newblock {\em Combinatorial methods in density estimation}.
\newblock Springer Science \& Business Media, 2012.

\bibitem[DS07]{dasgupta2007probabilistic}
Sanjoy Dasgupta and Leonard Schulman.
\newblock A probabilistic analysis of em for mixtures of separated, spherical
  gaussians.
\newblock {\em Journal of Machine Learning Research}, 8(Feb):203--226, 2007.

\bibitem[FSO06]{feldman2006pac}
Jon Feldman, Rocco~A Servedio, and Ryan O’Donnell.
\newblock Pac learning axis-aligned mixtures of gaussians with no separation
  assumption.
\newblock In {\em International Conference on Computational Learning Theory},
  pages 20--34. Springer, 2006.

\bibitem[GHK15]{ge2015learning}
Rong Ge, Qingqing Huang, and Sham~M Kakade.
\newblock Learning mixtures of gaussians in high dimensions.
\newblock In {\em Proceedings of the forty-seventh annual ACM symposium on
  Theory of computing}, pages 761--770, 2015.

\bibitem[HK13]{hsu2013learning}
Daniel Hsu and Sham~M Kakade.
\newblock Learning mixtures of spherical gaussians: moment methods and spectral
  decompositions.
\newblock In {\em Proceedings of the 4th conference on Innovations in
  Theoretical Computer Science}, pages 11--20, 2013.

\bibitem[HL18]{hopkins2018mixture}
Samuel~B Hopkins and Jerry Li.
\newblock Mixture models, robustness, and sum of squares proofs.
\newblock In {\em Proceedings of the 50th Annual ACM SIGACT Symposium on Theory
  of Computing}, pages 1021--1034, 2018.

\bibitem[HP15]{hardt2015tight}
Moritz Hardt and Eric Price.
\newblock Tight bounds for learning a mixture of two gaussians.
\newblock In {\em Proceedings of the forty-seventh annual ACM symposium on
  Theory of computing}, pages 753--760, 2015.

\bibitem[HRRS86]{HampelRRS86}
Frank~R. Hampel, Elvezio~M. Ronchetti, Peter~J. Rousseeuw, and Werner~A.
  Stahel.
\newblock {\em Robust Statistics: the Approach based on Influence Functions}.
\newblock Wiley Series in Probability and Mathematical Statistics, 1986.

\bibitem[Hub64]{huber1964robust}
Peter~J Huber.
\newblock Robust estimation of a location parameter.
\newblock {\em The Annals of Mathematical Statistics}, 35(1):73--101, 1964.

\bibitem[Hub04]{huber2004robust}
Peter~J Huber.
\newblock {\em Robust statistics}, volume 523.
\newblock John Wiley \& Sons, 2004.

\bibitem[JLT20]{jambulapati2020robust}
Arun Jambulapati, Jerry Li, and Kevin Tian.
\newblock Robust sub-gaussian principal component analysis and
  width-independent schatten packing.
\newblock In {\em Advances in Neural Information Processing Systems 33: Annual
  Conference on Neural Information Processing Systems 2020, NeurIPS 2020,
  December 6-12, 2020, virtual}, 2020.

\bibitem[Kan21]{kane2021robust}
Daniel~M Kane.
\newblock Robust learning of mixtures of gaussians.
\newblock In {\em Proceedings of the 2021 ACM-SIAM Symposium on Discrete
  Algorithms (SODA)}, pages 1246--1258. SIAM, 2021.

\bibitem[KK10]{kumar2010clustering}
Amit Kumar and Ravindran Kannan.
\newblock Clustering with spectral norm and the k-means algorithm.
\newblock In {\em 2010 IEEE 51st Annual Symposium on Foundations of Computer
  Science}, pages 299--308. IEEE, 2010.

\bibitem[KKK19]{karmalkar2019list}
Sushrut Karmalkar, Adam Klivans, and Pravesh Kothari.
\newblock List-decodable linear regression.
\newblock In {\em Advances in Neural Information Processing Systems}, pages
  7425--7434, 2019.

\bibitem[KKM18]{klivans2018efficient}
Adam Klivans, Pravesh~K Kothari, and Raghu Meka.
\newblock Efficient algorithms for outlier-robust regression.
\newblock In {\em Conference On Learning Theory}, pages 1420--1430, 2018.

\bibitem[KSS18]{kothari2018robust}
Pravesh~K Kothari, Jacob Steinhardt, and David Steurer.
\newblock Robust moment estimation and improved clustering via sum of squares.
\newblock In {\em Proceedings of the 50th Annual ACM SIGACT Symposium on Theory
  of Computing}, pages 1035--1046, 2018.

\bibitem[KSV08]{KannanSV08}
Ravindran Kannan, Hadi Salmasian, and Santosh~S. Vempala.
\newblock The spectral method for general mixture models.
\newblock {\em {SIAM} J. Comput.}, 38(3):1141--1156, 2008.

\bibitem[LAT{\etalchar{+}}08]{LiATSCRCBFCM08}
Jun~Z. Li, Devin~M. Absher, Hua Tang, Audrey~M. Southwick, Amanda~M. Casto,
  Sohini Ramachandran, Howard~M. Cann, Gregory~S. Barsh, Marcus Feldman,
  Luigi~L. Cavalli-Sforza, and Richard~M. Myers.
\newblock Worldwide human relationships inferred from genome-wide patterns of
  variation.
\newblock 319:1100--1104, 2008.

\bibitem[Li18]{Li18}
Jerry~Zheng Li.
\newblock {\em Principled approaches to robust machine learning and beyond}.
\newblock PhD thesis, Massachusetts Institute of Technology, 2018.

\bibitem[LM20]{liu2020settling}
Allen Liu and Ankur Moitra.
\newblock Settling the robust learnability of mixtures of gaussians.
\newblock {\em arXiv preprint arXiv:2011.03622}, 2020.

\bibitem[LRV16]{lai2016agnostic}
Kevin~A Lai, Anup~B Rao, and Santosh Vempala.
\newblock Agnostic estimation of mean and covariance.
\newblock In {\em 2016 IEEE 57th Annual Symposium on Foundations of Computer
  Science (FOCS)}, pages 665--674. IEEE, 2016.

\bibitem[LS17]{li2017robust}
Jerry Li and Ludwig Schmidt.
\newblock Robust and proper learning for mixtures of gaussians via systems of
  polynomial inequalities.
\newblock In {\em Conference on Learning Theory}, pages 1302--1382. PMLR, 2017.

\bibitem[LY20]{LiY20}
Jerry Li and Guanghao Ye.
\newblock Robust gaussian covariance estimation in nearly-matrix multiplication
  time.
\newblock In {\em Advances in Neural Information Processing Systems 33: Annual
  Conference on Neural Information Processing Systems 2020, NeurIPS 2020,
  December 6-12, 2020, virtual}, 2020.

\bibitem[Mur12]{murphy2012machine}
Kevin~P Murphy.
\newblock {\em Machine learning: a probabilistic perspective}.
\newblock MIT press, 2012.

\bibitem[MV10]{moitra2010settling}
Ankur Moitra and Gregory Valiant.
\newblock Settling the polynomial learnability of mixtures of gaussians.
\newblock In {\em 2010 IEEE 51st Annual Symposium on Foundations of Computer
  Science}, pages 93--102. IEEE, 2010.

\bibitem[MV18]{meister2018data}
Michela Meister and Gregory Valiant.
\newblock A data prism: Semi-verified learning in the small-alpha regime.
\newblock In {\em Conference On Learning Theory}, pages 1530--1546. PMLR, 2018.

\bibitem[MVW17]{mixon2017clustering}
Dustin~G Mixon, Soledad Villar, and Rachel Ward.
\newblock Clustering subgaussian mixtures by semidefinite programming.
\newblock {\em Information and Inference: A Journal of the IMA}, 6(4):389--415,
  2017.

\bibitem[Pat11]{patinkin2011method}
Seth Patinkin.
\newblock Method, apparatus, and system for clustering and classification,
  August~30 2011.
\newblock US Patent 8,010,466.

\bibitem[Pea94]{Pearson94}
Karl Pearson.
\newblock Contributions to the mathematical theory of evolution.
\newblock {\em Philosophical Transactions of the Royal Society of London},
  185:71--110, 1894.

\bibitem[PLJD10]{PaschouLJD10}
Peristera Paschou, Jamey Lewis, Asif Javed, and Petros Drineas.
\newblock Ancestry informative markers for fine-scale individual assignment to
  worldwide populations.
\newblock 47(12):835--847, 2010.

\bibitem[PSBR18]{prasad2018robust}
Adarsh Prasad, Arun~Sai Suggala, Sivaraman Balakrishnan, and Pradeep Ravikumar.
\newblock Robust estimation via robust gradient estimation.
\newblock {\em arXiv preprint arXiv:1802.06485}, 2018.

\bibitem[RH17]{RigolletH17}
Philippe Rigollet and Jan-Christian H\"utter.
\newblock {\em High-Dimensional Statistics}.
\newblock 2017.

\bibitem[RPW{\etalchar{+}}02]{RosenbergPWCKZF02}
Noah~A. Rosenberg, Jonathan~K. Pritchard, James~L. Weber, Howard~M. Cann,
  Kenneth~K. Kidd, Lev~A. Zhivotovsky, and Marcus~W. Feldman.
\newblock Genetic structure of human populations.
\newblock {\em Science}, 298:2381--2385, 2002.

\bibitem[RV17a]{RegevV17}
Oded Regev and Aravindan Vijayaraghavan.
\newblock On learning mixtures of well-separated gaussians.
\newblock In {\em 58th {IEEE} Annual Symposium on Foundations of Computer
  Science, {FOCS} 2017, Berkeley, CA, USA, October 15-17, 2017}, pages 85--96,
  2017.

\bibitem[RV17b]{regev2017learning}
Oded Regev and Aravindan Vijayaraghavan.
\newblock On learning mixtures of well-separated gaussians.
\newblock In {\em 2017 IEEE 58th Annual Symposium on Foundations of Computer
  Science (FOCS)}, pages 85--96. IEEE, 2017.

\bibitem[RY20]{raghavendra2020list}
Prasad Raghavendra and Morris Yau.
\newblock List decodable learning via sum of squares.
\newblock In {\em Proceedings of the Fourteenth Annual ACM-SIAM Symposium on
  Discrete Algorithms}, pages 161--180. SIAM, 2020.

\bibitem[SKL17]{SteinhardtKL17}
Jacob Steinhardt, Pang~Wei Koh, and Percy Liang.
\newblock Certified defenses for data poisoning attacks.
\newblock In {\em Advances in Neural Information Processing Systems 30: Annual
  Conference on Neural Information Processing Systems 2017, December 4-9, 2017,
  Long Beach, CA, {USA}}, pages 3517--3529, 2017.

\bibitem[Ste18]{Steinhardt18}
Jacob Steinhardt.
\newblock {\em Robust Learning: Information Theory and Algorithms}.
\newblock PhD thesis, Stanford University, 2018.

\bibitem[STM20]{santurkar2020breeds}
Shibani Santurkar, Dimitris Tsipras, and Aleksander Madry.
\newblock Breeds: Benchmarks for subpopulation shift.
\newblock {\em arXiv preprint arXiv:2008.04859}, 2020.

\bibitem[SVC16]{steinhardt2016avoiding}
Jacob Steinhardt, Gregory Valiant, and Moses Charikar.
\newblock Avoiding imposters and delinquents: Adversarial crowdsourcing and
  peer prediction.
\newblock In {\em Advances in Neural Information Processing Systems}, pages
  4439--4447, 2016.

\bibitem[TLM18]{tran2018spectral}
Brandon Tran, Jerry Li, and Aleksander Madry.
\newblock Spectral signatures in backdoor attacks.
\newblock In {\em Advances in Neural Information Processing Systems}, pages
  8000--8010, 2018.

\bibitem[Tuk60]{tukey1960survey}
John~W Tukey.
\newblock A survey of sampling from contaminated distributions.
\newblock {\em Contributions to probability and statistics}, pages 448--485,
  1960.

\bibitem[Tuk75]{tukey1975mathematics}
John~W. Tukey.
\newblock Mathematics and the picturing of data.
\newblock In {\em Proceedings of the International Congress of Mathematicians,
  Vancouver, 1975}, volume~2, pages 523--531, 1975.

\bibitem[VW04]{vempala2004spectral}
Santosh Vempala and Grant Wang.
\newblock A spectral algorithm for learning mixture models.
\newblock {\em Journal of Computer and System Sciences}, 68(4):841--860, 2004.

\bibitem[WK06]{WarmuthK06}
Manfred~K. Warmuth and Dima Kuzmin.
\newblock Randomized {PCA} algorithms with regret bounds that are logarithmic
  in the dimension.
\newblock In {\em Advances in Neural Information Processing Systems 19,
  Proceedings of the Twentieth Annual Conference on Neural Information
  Processing Systems, Vancouver, British Columbia, Canada, December 4-7, 2006},
  pages 1481--1488, 2006.

\bibitem[WK18]{wierzchon2018modern}
S{\l}awomir~T Wierzcho{\'n} and Mieczys{\l}aw~A K{\l}opotek.
\newblock {\em Modern algorithms of cluster analysis}.
\newblock Springer, 2018.

\end{thebibliography}
\begin{appendix}
  
\end{appendix}

\end{document}